\documentclass[12pt]{amsart}
\usepackage{calligra}
\usepackage{amsthm}
\usepackage{amsmath}
\usepackage{amssymb}
\usepackage{amscd}
\usepackage{bbm}
\usepackage[all]{xy}
\usepackage{mathrsfs}
\usepackage[top=72pt,bottom=96pt,left=72pt,right=72pt]{geometry}
\newtheorem{Lemma}{Lemma}[section]
\newtheorem{Remark}[Lemma]{Remark}

\newtheorem{Theorem}[Lemma]{Theorem}

\newtheorem{Definition}[Lemma]{Definition}
\newtheorem{Proposition}[Lemma]{Proposition}
\numberwithin{equation}{section}
\def\F{\mathfrak{F}}
\def\f{\mathfrak{f}}

\def\qS{\mathscr{S}}
\def\C{\mathbb{C}}

\def\F{\mathfrak{F}}
\def\qGG{\mathfrak{qGG}}
\def\Hor{\mathrm{Hor}}
\def\dvol{\mathrm{dvol}}
\def\id{\mathrm{id}}
\def\Ker{\mathrm{Ker}}
\def\Im{\mathrm{Im}}
\def\YM{\mathrm{YM}}
\def\SM{\mathrm{SM}}

\def\YMSM{\mathrm{YMSM}}

\def\H{\mathrm{Hor}}
\def\Mor{\textsc{Mor}}
\def\N{\mathbb{N}}

\def\R{\mathbb{R}}
\def\Z{\mathbb{Z}}
\def\C{\mathbb{C}}
\def\W{\mathbb{W}}
\def\A{\mathbb{A}}
\def\h{\mathbbm{h}}
\def\v{\mathbbm{v}}

\def\triv{\mathrm{triv}}

\def\l{\mathrm{L}}
\def\Id{\mathrm{Id}}

\def\ad{\mathrm{ad}}
\def\Ad{\mathrm{Ad}}
\def\class{\mathrm{class}}

\def\r{\mathrm{R}}
\def\U{\mathcal{U}}
\def\G{\mathcal{G}}
\def\X{\mathcal{X}}
\def\T{\mathcal{T}}

\def\SU{\mathcal{SU}}
\def\U{\mathcal{U}}

\def\m{\mathfrak{m}}
\def\su{\mathfrak{su}}
\def\u{\mathfrak{u}}

\begin{document}
\date{\today}
\title{On the Non--Commutative Geometry of the Electroweak Interaction}
\author{Gustavo Amilcar Salda\~na Moncada}
\address{Gustavo Amilcar Salda\~na Moncada\\
CIMAT, Unidad Guanajuato}
\email{gustavo.saldana@cimat.mx, gamilcar@ciencias.unam.mx}
\begin{abstract}
In differential geometry, the underlying mathematical structure behind the electroweak theory (the Glashow--Weinberg--Salam model) is a trivial principal $(SU(2)\times U(1))$--bundle over $\R^4$ with the Minkowski metric. The aim of this paper is to show that, using non--commutative geometry, the  electroweak theory can be described using a trivial quantum principal $SU(2)$--bundle over $C^\infty_\C(\R^4)$  with the Minkowski metric. 
 \begin{center}
  \parbox{300pt}{\textit{MSC 2010:}\ 46L87, 58B99.}
  \\[5pt]
  \parbox{300pt}{\textit{Keywords:}\ Quantum principal bundles, quantum principal connections, electroweak interaction.}
 \end{center}
\end{abstract}
\maketitle
\section{Introduction}

This paper arises from an \emph{elementary} observation in Yang--Mills theory: in differential geometry, the dimension of the differential calculus of the Lie group (which, in this case, coincides with the dimension of the group) used for the theory is exactly equal to the number of gauge boson fields of the theory \cite{diff1,na,nodg,gtvp}. For example, in the case of the electroweak theory, there are $4$ gauge boson fields, corresponding to the $4$ dimensions of the standard differential calculus on $SU(2)\times U(1)$, that is, the differential calculus on $SU(2)\times U(1)$ given by its space of differential forms. 

However, in non--commutative geometry, the dimension of a differential calculus on a Lie group need not coincide with the dimension of the group, since in this setting there is no restriction on using non--standard differential calculus. For example, in non--commutative geometry, for the Lie group $SU(2)$, there exists a $4$--dimensional differential calculus (of course, it is not the standard one given by the differential forms of $SU(2)$) \cite{stach}. This dimension is precisely the number of gauge boson fields of the electroweak theory. 

In this way, the aim of this paper is not to provide new information on the electroweak interaction, rather, the aim of this paper is to show that the non--commutative geometrical formulation of the Yang--Mills theory presented in~\cite{sald1,sald2,saldcon} allows one to describe the electroweak theory, including the Higgs mechanism, using solely a quantum principal $SU(2)$--bundle, by changing the standard differential calculus of $SU(2)$ given by its differential forms, to a concrete $4$--dimensional differential calculus on $SU(2)$. 

The importance of this paper lies in the following 3 points.
\begin{enumerate}
    \item The description of the full electroweak theory using only the Lie group $SU(2)$ and not $SU(2)\times U(1)$.
    \item The general theory used in this paper (see~\cite{sald1,sald2,saldcon}) was not developed specifically for this work; rather, this paper is an application of it. This illustrates both the scope and the generality of the theory in \cite{sald1,sald2,saldcon}.
    \item Grand unification theories (GUTs) are models in theoretical physics whose objective is to describe the physics of non--gravitational fundamental interactions of nature using a single principal $G$--bundle and symmetry--breaking mechanisms, for a Lie group $G$ that contains $SU(3)\times SU(2)\times U(1)$ as a subgroup, among other requirements on $G$ \cite{diff1}. However, at the time of writing, there is not yet sufficient theoretical and experimental evidence to conclude that such models can be satisfactorily realized. 
    
    In this way, points (1) and (2) open the door to the following question.\\
    
    \emph{What if GUTs fail at the theoretical level because we are only searching for a Lie group $G$ with the appropriate properties, whereas we should instead be looking for a Lie group $G$ together with \textbf{a non--standard differential calculus on it} in order to construct such theories in the framework of non--commutative geometry (not in the framework of differential geometry)?\\}
    
    We will address this question in forthcoming publications.\\
\end{enumerate}

This paper is organized into six sections. Following this introduction, Section~2 presents the canonical $\ast$--Hopf algebra $\SU(2)$ associated with the Lie group $SU(2)$, as well as the $4$--dimensional bicovariant differential calculus on $SU(2)$ (more precisely, on $\SU(2)$) that will be used throughout the paper.

In Section~3, we introduce the trivial quantum principal $\SU(2)$--bundle over $B:=C^\infty_\C(\R^4)$ $$\zeta_{4D}=(P,B,\Delta_P)$$ on which our model of the electroweak theory is developed, together with its differential calculus. In this section, we also show that the \emph{non--commutative geometrical} Bianchi identity of $\zeta_{4D}$ coincides with the \emph{classical} Bianchi identity of the trivial principal $G$--bundle over $\R^4$ $$\pi_\class : \R^4\times G\longrightarrow \R^4,\qquad (x,C)\longmapsto x,$$ where $G=SU(2)\times U(1)$. Furthermore, Theorem~\ref{prop3.6} proves that the gauge group $\mathfrak{GG}$ of $\pi_\class$ embeds into the \emph{quantum gauge group} $\qGG$ of $\zeta_{4D}$, while Proposition~3.13 shows that the action of $\mathfrak{GG}$ on the space of principal connections of $\pi_\class$ coincides with the action of $\qGG$ on the space of quantum principal connections of $\zeta_{4D}$.

Section~4 is devoted to proving that the \emph{non--commutative geometrical} Yang--Mills equation of $\zeta_{4D}$, in the sense of~\cite{sald2,saldcon}, coincides with the \emph{classical} Yang--Mills equation of $\pi_\class$, that is, the Yang--Mills equation of the electroweak theory. 

In Section~5, we incorporate scalar matter fields into our model. More precisely, we prove that the \emph{non--commutative geometrical} equations of motion for scalar matter fields coupled to Yang--Mills fields (i.e., gauge boson fields) of $\zeta_{4D}$ coincide with the \emph{classical} equations of motion for scalar matter fields coupled to Yang--Mills fields of $\pi_\class$. In particular, in Subsection~5.3 we show that $\zeta_{4D}$ provides an exact description of the Higgs mechanism in the electroweak theory. Finally, in the sixth section we show that $\zeta_{4D}$ provides a  description of fermionic matter fields coupled to Yang--Mills fields in the electroweak theory.

It is worth noticing that, throughout the text, the term \emph{quantum} is used to refer to the framework of non--commutative geometry and should not be confused with second quantization or related notions typically used in physics. Similarly, the term \emph{classical} refers to the framework of differential geometry. In the same way, throughout the text we work with associated left/right quantum vector bundles, whose elements are interpreted as \emph{left/right non--commutative geometrical} matter fields. Clearly, the term \emph{non--commutative geometrical} refers to the fact that we are working within the framework of non--commutative geometry, while the terms \emph{left} and \emph{right} refer to the underlying $B$--module structure. The reader should not confuse this terminology with that of chiral fields in physics. For example, in Section 6.3 we will verify that right--handed Lepton singlets can be described as \emph{left non--commutative geometrical} matter fields.

On the other hand, throughout the text we work with unital $\ast$--algebras, so the term \emph{unital} will be omitted.  In addition, although the symbol $\ast$ is commonly used to denote both the pull--back of a function and the dual space of a vector space, in this paper, the symbol $\ast$ is reserved for the antilinear involution of a $\ast$--algebra, or the antilinear involution of a vector space. Consequently, we use the symbol $\#$ to denote both the pull--back of a function and the dual space of a vector space. Furthermore, we use Sweedler notation throughout the text.

To conclude this section, it is worth mentioning the general theory used in this paper \cite{sald1,sald2,saldcon} was developed within Durdevich’s formulation of quantum principal bundles (\cite{micho1,micho2,micho3,stheve}); accordingly, this paper is also developed within that framework. Nevertheless, this formulation is not widely known. For this reason, in each section we have included a subsection titled \emph{General Theory} in order to provide the reader with a clear, coherent, and reasonable self--contained exposition of Durdevich’s formulation. This is necessary, as familiarity with the essential aspects of Durdevich’s formulation is required to correctly follow the theory developed for $\zeta_{4D}$ in each section.

\section{The Group}
In differential geometry, the structure group of a principal bundle plays a central role. The same holds in non--commutative geometry, and in this section we will address this aspect.

\subsection{General Theory}

As it is common in non--commutative geometry, we will use  quantum groups and $\ast$--Hopf algebras instead Lie groups (see references \cite{woro1,woro2}). However, the differential calculus on a quantum group will be the universal differential envelope $\ast$--calculus (not the Woronowicz high--order calculus of references \cite{woro2}). Here, we will present a brief summary of this calculus. For more details, see  references \cite{micho1,stheve}.

Let $$(\Lambda,d)$$ be a $\ast$--First--Order Differential Calculus (abbreviated $\ast$--FODC) over a $\ast$--Hopf algebra $$(A,\cdot,\mathbbm{1},\ast,\Delta,\epsilon,S)$$ and consider the graded vector space  ($k\in \N$)
$$\otimes^\bullet_A\Lambda:=\bigoplus_k (\otimes^k_A\Lambda)\quad \mbox{ with } \quad \otimes^0_A \Lambda=A,\quad  \otimes^k_A\Lambda:=\underbrace{\Lambda\otimes_A\cdots\otimes_A \Lambda}_{k\; times}$$ endowed with its canonical graded $\ast$--algebra structure, which is given by
$$(\vartheta_1\otimes_{A}\cdots\otimes_{A}\vartheta_k)\cdot(\vartheta'_1\otimes_{A}\cdots\otimes_{A}\vartheta'_l):=\vartheta_1\otimes_{A}\cdots\otimes_{A}\vartheta_k\otimes_{A}\vartheta'_1\otimes_{A}\cdots\otimes_{A}\vartheta'_l,$$ $$(\vartheta_1\otimes_{A}\cdots\otimes_{A}\vartheta_k)^{\ast}:=(-1)^{k(k-1)\over 2}\,\vartheta^{\ast}_k\otimes_{A}\cdots\otimes_{A}\vartheta^{\ast}_1, $$ for $\vartheta_1\otimes_{A}\cdots\otimes_{A}\vartheta_k$ $\in$ $\otimes^{k}_A\Lambda$ and $\vartheta'_1\otimes_{A}\cdots\otimes_{A}\vartheta'_l$ $\in$ $\otimes^{l}_A\Lambda$. Now, let us consider the quotient graded space
  \begin{equation}
      \label{ec.2.1}     \Lambda^\wedge:=\otimes^\bullet_A\Lambda/\mathcal{Q},
  \end{equation}
  where $\mathcal{Q}$ is the two--side ideal of $\otimes^\bullet_A\Lambda$ generated by elements
  \begin{equation}
      \label{ec.2.2}
     \sum_i dg_i\otimes_A dh_i \quad \mbox{ such that } \quad \sum_i g_i\,dh_i=0,
  \end{equation}
  for all $g_i$, $h_i$ $\in$ $A$. According to \cite{micho1,stheve}, the graded $\ast$--algebra structure of $\otimes^\bullet_A \Lambda$ endows $\Lambda^\wedge$ with structure of graded $\ast$--algebra. The product in $\Lambda^\wedge$ will be denoted simply by juxtaposition of elements. 
  
  On the other hand, for a given $t=\vartheta_1\cdots \vartheta_n$ $\in$ $\Lambda^{\wedge\,n }$ with $\vartheta_1$,..., $\vartheta_n$ $\in$ $\Lambda$, the linear map
  \begin{equation}
      \label{ec.2.3}
     d:\Lambda^\wedge\longrightarrow \Lambda^\wedge 
  \end{equation}
  given by $$d(t)=d(\vartheta_1\cdots \vartheta_n)=\displaystyle \sum^n_{j=1}(-1)^{j-1}\vartheta_1\cdots \vartheta_{j-1}\cdot d\vartheta_j\cdot \vartheta_{j+1}\cdots \vartheta_n \; \in \; \Lambda^{\wedge\,n+1 },$$ where $d\vartheta_j=\displaystyle \sum_l dg_l\,dh_l$ if $\vartheta_j=\displaystyle \sum_l g_l\,(dh_l)$, is well--defined, satisfies the graded Leibniz rule, $d^2=0$ and $d(t^\ast)=(dt)^\ast$ \cite{micho1,stheve}. In this way, 
  \begin{equation}
    \label{ec.2.4}
     (\Lambda^\wedge,d,\ast)
\end{equation}
  is a graded differential $\ast$--algebra generated by its degree 0 elements $\Lambda^{\wedge\,0}= A$  and it is called {\it the universal differential envelope $\ast$--calculus}.

 In references \cite{micho1,stheve}, the reader can find a proof of the following statement.

\begin{Proposition}
\label{prop2.0}
Suppose $(\Omega^\bullet=\bigoplus_k \Omega^k,d,\ast)$ is a graded differential $\ast$--algebra and $(\Lambda,d)$ is a $\ast$--FODC over $A$.

Let $$\phi^0: \Lambda^{\wedge 0}\longrightarrow \Omega^0$$ be a $\ast$--algebra morphism and $$\phi^1:\Lambda^{\wedge 1}\longrightarrow \Omega^1$$ be a linear map such that $$\phi^1(a \,db)=\phi^0(a)\,d(\phi^0(b))$$ for all $a$, $b$ $\in$ $A$. Then, there exist unique linear maps $$\phi^k: \Lambda^{\wedge k}\longrightarrow \Omega^k$$ for all $k\geq 2$ such that $$\phi:=\bigoplus_k\phi^k: \Lambda^\wedge\longrightarrow \Omega^\bullet$$ is a graded differential $\ast$--algebra morphism.
\end{Proposition}

For bicovariant $\ast$--FODC's, there is another construction of the universal differential envelope $\ast$--calculus which is very useful. In fact, let $(\Lambda,d)$ be a bicovariant $\ast$--FODC (\cite{stheve}). Then, it is well--known that  $$(\Lambda,d)\cong (A\otimes {\Ker(\epsilon)\over {\mathcal{R}}},d ) $$ for some right $A$--ideal ${\mathcal{R}}$ $\subseteq$ $\Ker(\epsilon)$ that satisfies $\Ad({\mathcal{R}})\subseteq {\mathcal{R}}\otimes A$, $S({\mathcal{R}})^\ast\subseteq {\mathcal{R}}$ \cite{stheve,woro2}. Let us define the \emph{quantum dual Lie algebra} as $$\mathfrak{qa}^\#:= {\Ker(\epsilon)\over {\mathcal{R}}}.$$ In addition, let us take
\begin{equation}
    \label{ec.2.5}
    \begin{aligned}
        \mathfrak{qa}^{\#\wedge}=\otimes^\bullet \mathfrak{qa}^\#/S^\wedge, \qquad \otimes^\bullet\mathfrak{qa}^\#:=\bigoplus_k (\otimes^k\mathfrak{qa}^\#)\\
    \otimes^0\mathfrak{qa}^\#=\C\mathbbm{1},\qquad \otimes^k\mathfrak{qa}^\#:=\underbrace{\mathfrak{qa}^\#\otimes\cdots\otimes \mathfrak{qa}^\#}_{k\; times}
    \end{aligned}
\end{equation}
($k\in \N$), where $S^\wedge$ is the graded two--side ideal of $\otimes^\bullet\mathfrak{qa}^\#$ generated by elements 
\begin{equation}
    \label{ec.2.6}
    \pi(g^{(1)})\otimes \pi(g^{(2)})\qquad \mbox{ for all }\qquad g \,\in\, {\mathcal{R}}.
\end{equation}
Then, in light of \cite{micho1}, we have
\begin{equation}
    \label{ec.2.7}
    (\Lambda^\wedge,d,\ast)\cong (A\otimes \mathfrak{qa}^{\#\wedge},d,\ast).
\end{equation}
For degree 0, the previous isomorphism is given by $A\cong A\otimes \mathbbm{1}$ in the canonical way and for degree $1$, the previous isomorphism matches with the well--known isomorphism $(\Lambda,d)\cong (A\otimes \mathfrak{qa}^\#,d)$ (\cite{micho1}). Furthermore, there is a Maurer--Cartan formula (\cite{micho1,stheve})
\begin{equation}
    \label{ec.2.8}
    d\pi(a)=-\pi(a^{(1)})\pi(a^{(2)})
\end{equation}
for all $a$ $\in$ $A$ (in Sweedler's notation), where
\begin{equation}
    \label{qgermsmap}
\pi:A\longrightarrow \mathfrak{qa}^\#,\qquad a\longmapsto \pi(a):=S(a^{(1)})da^{(2)}
\end{equation}
is the corresponding quantum germs map \cite{stheve}. The map $\pi$ has several useful properties, for example, $\pi|_{\Ker(\epsilon)}$ is surjective, and
\begin{equation}
\label{properties}
    \begin{aligned}
\ker(\pi)=\mathcal{R}_{4D}\oplus \C\mathbbm{1}, \quad & \qquad dg=g^{(1)}\pi(g^{(2)}),\quad  \qquad   \pi(g)^\ast=-\pi(S(g)^\ast)\\
\pi(g)=-(d&S(g^{(1)}))g^{(2)},\qquad dS(g)=-\pi(g^{(1)})S(g^{(2)}).&
    \end{aligned}
\end{equation}

If the $\ast$--FODC $(\Lambda,d)$ is bicovariant, then, the coproduct $\Delta$ of $A$ can be extended using Proposition \ref{prop2.0} to a graded differential $\ast$--algebra morphism (\cite{micho1,stheve})
\begin{equation}
\label{ec.2.9}
    \Delta: \Lambda^\wedge\longrightarrow \Lambda^\wedge\otimes \Lambda^\wedge.
\end{equation}
In the last tensor product, we have taken the tensor product of graded differential $\ast$--algebras. This map satisfies
\begin{equation}
\label{ec.2.10}
\Delta(\theta)=\ad(\theta)+\mathbbm{1}\otimes \theta 
\end{equation}
for every $\theta$ $\in$ $\mathfrak{qa}^\#$, where 
\begin{equation}
    \label{adgeneral}
    \ad: \mathfrak{qa}^\#\longrightarrow \mathfrak{qa}^\#\otimes A
\end{equation}
is the $A$--corepresentation that fulfills $$\ad\circ \pi=(\pi\otimes \id_A)\circ \Ad,$$ where $$\Ad:A\longrightarrow A\otimes A, \qquad a\longmapsto a^{(2)}\otimes S(a^{(1)})a^{(3)}$$ 
is the right adjoint coaction of $A$ \cite{micho1,stheve}. In addition, for bicovariant $\ast$--FODCs, the counit $\epsilon$ and the antipode $S$ can also be extended to $\Gamma^\wedge$
\begin{equation}
\label{ec.2.11}
\epsilon: \Lambda^\wedge \longrightarrow \C,
\end{equation}
\begin{equation}
\label{ec.2.12}
S: \Lambda^\wedge \longrightarrow \Lambda^\wedge.
\end{equation}
The extension of the coint is given by
$$\epsilon|_{\Lambda^{\wedge\,k}}=0 \quad \mbox{ for all } \quad k\geq 1 $$ and the extension of the antipode is as follows: for every $\theta$ $\in$ $\mathfrak{qa}^\#$ with $\theta=\pi(a)$ for some $a$ $\in$ $A$,  the formula
$$S(\theta)=S(\pi(a))=-\pi(a^{(2)})S(a^{(3)})S(S(a^{(1)}))$$ is well--defined. With this, it is possible to define $S:\Lambda\longrightarrow \Lambda$ such that
$$S(b\,\pi(a))=S(\pi(a))S(b),\qquad S(b\, da)=d(S(a))\,S(b)$$ and since $\Lambda^\wedge$ is generated by its degree 0, we can extend $S$ to whole space $\Lambda^\wedge$. For more details, see reference \cite{micho1}.

In light of \cite{micho1},  $$(\Lambda^\wedge,\cdot,\mathbbm{1},\ast,d,\Delta,\epsilon,S)$$ is a graded differential $\ast$--Hopf algebra. With this structure, the right $A$--corepresentation $\Ad$ can also be extended, by means of
\begin{equation}
\label{ec.2.13}
\Ad:\Lambda^\wedge \longrightarrow \Lambda^\wedge \otimes \Lambda^\wedge
\end{equation}
(the last tensor product, we have taken the tensor product of graded differential $\ast$--algebras) such that $$\Ad(t)=(-1)^{\partial t^{(1)}\partial t^{(2)}} t^{(2)}\otimes S(t^{(1)})t^{(3)}$$ for all $t$ $\in$ $\Gamma^\wedge$, where $\partial x$ denotes the grade of $x$ and $$(\id_{\Gamma^\wedge}\otimes \Delta)\Delta(t)=(\Delta\otimes \id_{\Gamma^\wedge})\Delta(t)=t^{(1)}\otimes t^{(2)}\otimes t^{(3)}.$$

\subsection{The Model}

For the model of the electroweak theory  that we want to develop using the advantages of  non--commutative geometry, the structure group will be $SU(2)$ or, to be more precisely, its canonical $\ast$--Hopf algebra. 

Consider the Lie group 
\begin{equation}
    \label{ec.2.14}
    SU(2)=\left\{ \begin{pmatrix}
a & -b^\ast \\
b & a^\ast 
\end{pmatrix} \in M_2(\C)\;\; \left|\right.\;\; |a|^2+ |b|^2=1 \right\}.
\end{equation}

According to reference \cite{woro1}, there is a canonical (matrix compact) quantum group associated to $SU(2)$. The dense $\ast$--Hopf algebra of this quantum group will be denoted by
\begin{equation}
    \label{ec.2.15}
    (\SU(2),\cdot,\mathbbm{1},\ast,\Delta,\epsilon,S),
\end{equation}
where $$\SU(2)$$ is the $\ast$--subalgebra of $C_\C(SU(2)):=\{g: SU(2)\longrightarrow \C \mid g \mbox{ is continous} \}$,  generated by the following smooth functions 
\begin{equation}
\label{ec.2.16}
    \begin{aligned}
\alpha: SU(2) &\longrightarrow \C, \qquad \qquad & \gamma: SU(2) &\longrightarrow \C \\
A &\longmapsto a_{11}, \qquad \qquad &  A &\longmapsto a_{21},
\end{aligned}
\end{equation}
where $A=(a_{ij})$. Notice that $\alpha$ and $\gamma$ satisfies 
\begin{equation}
\label{ec.2.17}
    \alpha\,\alpha^\ast+\gamma\,\gamma^\ast=\mathbbm{1}
\end{equation}
and of course, $\SU(2)$ is a commutative $\ast$--algebra. The structure of $\ast$--Hopf algebra is given by the pull--back of the product of the group, i.e., 
    \begin{equation}
        \label{ec.2.18}
        \Delta(g): SU(2)\times SU(2) \longmapsto \C, \qquad \Delta(g)(A,B)=g(AB)\quad \mbox{ with }\quad A,\,B \,\in \, SU(2)
    \end{equation}
for all $g$ $\in$ $\SU(2)$. In particular, in tensorial notation, we have
\begin{equation}
\label{ec.2.19}
    \Delta(\alpha)=\alpha\otimes \alpha-\gamma^\ast \otimes \gamma, \qquad \Delta(\gamma)=\gamma\otimes \alpha+\alpha^\ast\otimes \gamma.
\end{equation}
The counit and the coinverse are defined as follows: 
\begin{equation}
    \label{ec.2.20}
    \epsilon:\SU(2) \longrightarrow \C,\qquad g \longmapsto g(e),
\end{equation}
where $e$ $\in$ $SU(2)$ is the identity element; and
\begin{equation}
    \label{ec.2.21}
    S(g): SU(2) \longrightarrow \C,\qquad A \longmapsto S(g)(A)=g(A^{-1})=g(A^\ast)
\end{equation}
for all $g$ $\in$ $\SU(2)$, with $A^\ast$ the complex transpose matrix of $A$. In particular, we obtain
\begin{equation}
    \label{ec.2.22}
    \epsilon(\alpha)=1,\qquad \epsilon(\gamma)=0,
\end{equation}
\begin{equation}
    \label{ec.2.23}
    S(\alpha)=\alpha^\ast,\qquad S(\gamma)=-\gamma, \qquad S\circ \ast=\ast \circ S.
\end{equation}

\begin{Definition}
\label{def2.1}
    We define the (right) adjoint $\SU(2)$--corepresentation on  $\SU(2)$ as the linear map $$\Ad:\SU(2)\longrightarrow \SU(2)\otimes \SU(2),\qquad g \longmapsto g^{(2)}\otimes S(g^{(1)})g^{(3)},$$ where $\Delta(g)=g^{(1)}\otimes g^{(2)}$. 
\end{Definition}

However, we will not use the space of {\it classical} differential forms on $\SU(2)$. Let us consider the $\ast$--FODC
\begin{equation}
    \label{ec.2.85}
    (\Gamma,d)
\end{equation}
described in reference \cite{woro2} for $\mu= 1$. By definition, the corresponding right $\SU(2)$--ideal $$\mathcal{R}_{4D}\subseteq \Ker(\epsilon)$$ associated with $(\Gamma,d)$ is generated by the multiplets
$$\mathcal{R}_1:=\{(\alpha+\alpha^\ast-2\,\mathbbm{1})^2\},\qquad \mathcal{R}_2:=\{ (\alpha+\alpha^\ast-2\,\mathbbm{1})\,\gamma,(\alpha+\alpha^\ast-2\,\mathbbm{1})\,(\alpha-\alpha^\ast),(\alpha+\alpha^\ast-2\,\mathbbm{1})\,\gamma^\ast\},$$ $$\mathcal{R}_3:=\{\gamma^2,\gamma\,(\alpha-\alpha^\ast),\alpha^{\ast 2}-2(\alpha\,\alpha^\ast-\gamma\,\gamma^\ast)+\alpha^2,\gamma^\ast\,(\alpha-\alpha^\ast),\gamma^{\ast 2} \}.$$ As it is proven in \cite{stach}, the $\ast$--FODC $(\Gamma,d)$ is bicovariant and $4$--dimensional, and 
$$\{\pi(\gamma\gamma^\ast),\;\pi(\alpha-\alpha^\ast),\;\pi(\gamma),\;\pi(\gamma^\ast) \} $$ is a linear basis of
\begin{equation}
    \label{ec.2.86}
    \su(2,\C)_{4D}^{\#}:={\Ker(\epsilon)\over \mathcal{R}_{4D}}.
\end{equation}
Here 
\begin{equation}
    \label{ec.2.87}
    \pi: \SU(2)\longrightarrow \su(2,\C)_{4D}^{\#}, \qquad g \longmapsto S(g^{(1)})\,dg^{(2)}\cong [g-\epsilon(g)\mathbbm{1}]_{\mathcal{R}_{4D}},
\end{equation}
is the corresponding quantum germs map, where $[h]_{\mathcal{R}_{4D}}$ denotes the equivalence class of $h$ (see Section 6 of \cite{stheve}).

From $\mathcal{R}_3$ we get $$0=\pi(\alpha^{\ast2})-2\,\pi(\alpha\alpha^\ast-\gamma\gamma^\ast)+\pi(\alpha^2)=\pi(\alpha^{\ast2})+4\,\pi(\gamma\gamma^\ast)+\pi(\alpha^2) $$ and from $\mathcal{R}_1$ we obtain 
\begin{eqnarray*}
    0&=&\pi(\alpha^2)+\pi(\alpha^{\ast 2})+2\,\pi(\alpha\alpha^\ast)-4\,\pi(\alpha)-4\,\pi(\alpha^\ast)
    \\
    &=&
    \pi(\alpha^2)+\pi(\alpha^{\ast 2})-4\,\pi(\gamma\gamma^\ast)-4\,\pi(\alpha)-4\,\pi(\alpha^\ast)
    \\
    &=&
    -8\,\pi(\gamma\gamma^\ast)-4\,\pi(\alpha)-4\,\pi(\alpha^\ast);
\end{eqnarray*}
which implies that 
\begin{equation}
    \label{ec.2.88}
    \pi(\gamma\gamma^\ast)=-{1\over 2}\pi(\alpha+\alpha^\ast).
\end{equation}
In this way, we will fix the following linear basis of $\su(2,\C)_{4D}^{\#}$: 
\begin{equation}
    \label{ec.2.89}
    \beta^\#_{4D}:=\{\eta_1:=-{i\,\sqrt{2}\over q_w}\,\pi(\gamma),\quad \eta_2:={i\,\sqrt{2}\over q_w}\,\pi(\gamma^\ast),\quad \eta_3:=-{i\over q_w} \pi(\alpha-\alpha^\ast),\quad \eta_4:={i\over q_4}\pi(\alpha+\alpha^\ast) \},
\end{equation}
for $q_w$, $q_4$ some real constants. Consequently, $\beta^\#_{4D}$ is a left/right $\SU(2)$--basis of $\Gamma$, according to  \cite{stheve}. It is worth mentioning that, in accordance with Section 9 of reference \cite{stach}, we have
\begin{equation}
    \label{ec.2.90}
    \pi(\alpha-\alpha^\ast)\,b=b\,\pi(\alpha-\alpha^\ast), \quad \pi(\gamma)\,b=b\,\pi(\gamma),\quad \pi(\gamma^\ast)\,b=b\,\pi(\gamma^\ast)
\end{equation}
for all $b$ $\in$ $\SU(2)$.

There is a right $\SU(2)$--module structure on $\su(2,\C)_{4D}^{\#}$ given by (\cite{stheve})
\begin{equation}
    \label{ec.2.91}
    \pi(g_1)\diamondsuit g_2:=\pi(g_1g_2-\epsilon(g_1)g_2))=S(g^{(1)}_2)\pi(g_1) g^{(2)}_2.
\end{equation}
For example, 
\begin{eqnarray*}
    \pi(\alpha-\alpha^\ast)\diamondsuit \alpha=S(\alpha^{(1)})\pi(\alpha-\alpha^\ast)\alpha^{(2)}&=&\alpha^\ast \pi(\alpha-\alpha^\ast)\alpha+\gamma^{\ast}\pi(\alpha-\alpha^\ast)\gamma \nonumber
    \\
    &=&
    (\alpha^\ast \alpha+\gamma\gamma^\ast)\pi(\alpha-\alpha^\ast)
    \\
    &=&
    \pi(\alpha-\alpha^\ast)\nonumber
\end{eqnarray*}
and 
\begin{eqnarray*}
    \pi(\alpha-\alpha^\ast)\diamondsuit \alpha=\pi(\alpha^2-\alpha\alpha^\ast)=\pi(\alpha^2)-\pi(\alpha\alpha^\ast)&=&\pi(\alpha^2)+\pi(\gamma\gamma^\ast)
    \\
    &=&
    \pi(\alpha^2)-{1\over 2}\pi(\alpha+\alpha^\ast).
\end{eqnarray*}
Thus
\begin{equation}
    \label{ec.2.92}
    \pi(\alpha^2)={1\over 2}\pi(\alpha+\alpha^\ast)+\pi(\alpha-\alpha^\ast).
\end{equation}
Since $0=\pi(\alpha^{\ast2})+4\,\pi(\gamma\gamma^\ast)+\pi(\alpha^2)$, it follows that
\begin{equation}
    \label{ec.2.93}
    \pi(\alpha^{\ast 2})={3\over 2}\pi(\alpha+\alpha^\ast)-\pi(\alpha-\alpha^\ast).
\end{equation}

A direct calculation using the definition of $\mathcal{R}_{4D}$ and the properties of $\pi$ 
proves that
\begin{Proposition}
    \label{prop2.4}
    The following relations hold
    $$ \pi(\gamma^2)=\pi(\gamma^{\ast 2})=0,\quad \pi(\alpha\gamma)=\pi(\alpha^\ast\gamma)=\pi(\gamma),\quad \pi(\alpha\gamma^\ast)=\pi(\alpha^\ast\gamma^\ast)=\pi(\gamma^\ast), $$ $$\pi(\alpha\alpha^\ast-\gamma\gamma^\ast)=\pi(\alpha+\alpha^\ast), $$ $$\pi(\gamma)^\ast=\pi(\gamma^\ast),\quad \pi(\gamma^\ast)^\ast=\pi(\gamma),\quad \pi(\alpha-\alpha^\ast)^\ast=-\pi(\alpha-\alpha^\ast),\quad \pi(\alpha+\alpha^\ast)^\ast=-\pi(\alpha+\alpha^\ast).$$
\end{Proposition}

The following map will play the role of the adjoint action of a Lie group on its Lie algebra. 
\begin{Definition}
\label{def2.5}
    We define the adjoint $\SU(2)$--corepresentation $$\ad:\su(2,\C)_{4D}^{\#}\longrightarrow \su(2,\C)_{4D}^{\#}\otimes \SU(2)$$ as the one define for the formula $$\ad\circ \pi=(\pi\otimes \id)\Ad.$$ 
\end{Definition}

In the {\it classical} case, the following map coincides with the pull--back of the Lie bracket of the Lie  algebra of a Lie group \cite{appendix}.

\begin{Definition}
\label{def2.6}
    We define the quantum Lie bracket, also known as the transpose commutator of $\su(2,\C)_{4D}^{\#}$, as the linear map $$c^T:=(\id\otimes \pi)\circ \ad:\su(2,\C)^\#_{4D} \longrightarrow \su(2,\C)^\#_{4D}\otimes\su(2,\C)^\#_{4D}.$$
\end{Definition}
 
 Also, we have

\begin{Proposition}
    \label{prop2.5}
    The following relations hold
\begin{equation*}
    \ad(\eta_1)=\eta_1\otimes \alpha^2+\eta_2\otimes -\gamma^2+\eta_3\otimes -\sqrt{2}\,\alpha\,\gamma,
\end{equation*}
\begin{equation*}
    \ad(\eta_2)=\eta_1\otimes -\gamma^{\ast 2}+\eta_2\otimes \alpha^{\ast 2}+\eta_3\otimes -\sqrt{2}\,\alpha^\ast\,\gamma^\ast,
\end{equation*}
\begin{equation*}
    \ad(\eta_3)=\eta_1\otimes \sqrt{2}\,\alpha\,\gamma^\ast+\eta_2\otimes \sqrt{2}\,\alpha^\ast\,\gamma+\eta_3\otimes(\alpha\,\alpha^\ast-\gamma\,\gamma^\ast).
\end{equation*}
\begin{equation*}
    \ad(\eta_4)=\eta_4\otimes \mathbbm{1}.
\end{equation*}
\end{Proposition}

\begin{proof}
The proof consists of direct calculations. In fact,
    $$\Delta(\alpha )=\alpha\otimes \alpha -\gamma^\ast \otimes \gamma,\qquad \Delta(\alpha^\ast)=\alpha^\ast \otimes \alpha^\ast -\gamma \otimes \gamma^\ast , $$
    $$\Delta(\gamma)=\gamma \otimes \alpha +\alpha^\ast \otimes \gamma , \qquad \Delta (\gamma^\ast )=\gamma^\ast \otimes \alpha^\ast +\alpha \otimes \gamma^\ast;$$
    so
    \begin{eqnarray*}
        (\Delta\otimes \id)(\Delta(\alpha))&=& \Delta(\alpha)\otimes \alpha-\Delta(\gamma^\ast)\otimes \gamma
        \\
        &=&
        \alpha\otimes \alpha\otimes \alpha\otimes \mathbbm{1}-\gamma^\ast\otimes \gamma \otimes \alpha- \gamma^\ast \otimes \alpha^\ast \otimes \gamma -\alpha \otimes \gamma^\ast \otimes \gamma,
    \end{eqnarray*}
    \begin{eqnarray*}
        (\Delta\otimes \id)(\Delta(\alpha^\ast))&=& \Delta(\alpha^\ast)\otimes \alpha^\ast-\Delta(\gamma)\otimes \gamma^\ast
        \\
        &=&
        \alpha^\ast\otimes \alpha^\ast\otimes \alpha^\ast-\gamma\otimes \gamma^\ast\otimes \alpha^\ast- \gamma\otimes \alpha\otimes \gamma^\ast-\alpha^\ast\otimes \gamma\otimes \gamma^\ast,
    \end{eqnarray*}
    \begin{eqnarray*}
        (\Delta\otimes \id)(\Delta(\gamma))&=& \Delta(\gamma)\otimes \alpha+\Delta_\otimes(\alpha^\ast)\otimes \gamma
        \\
        &=&
        \gamma\otimes \alpha\otimes \alpha+\alpha^\ast\otimes \gamma\otimes \alpha+\alpha^\ast\otimes \alpha^\ast\otimes \gamma+\gamma\otimes \gamma^\ast\otimes \gamma,
    \end{eqnarray*}
    \begin{eqnarray*}
        (\Delta\otimes \id)(\Delta(\gamma^\ast))&=& \Delta(\gamma^\ast)\otimes \alpha^\ast+\Delta_\otimes(\alpha)\otimes \gamma^\ast
        \\
        &=&
        \gamma^\ast\otimes \alpha^\ast\otimes \alpha^\ast+\alpha\otimes \gamma^\ast\otimes \mathbbm{1}\otimes \alpha^\ast+\alpha\otimes \alpha\otimes \gamma^\ast-\gamma^\ast\otimes \gamma\otimes \gamma^\ast.
    \end{eqnarray*}
Furthermore, since
$$S(\alpha)=\alpha^\ast, \qquad S(\alpha^\ast)=\alpha,\qquad S(\gamma)=-\gamma, \qquad S(\gamma^\ast)=-\gamma^\ast,$$ we get
\begin{eqnarray*}
        \Ad(\alpha)&=& 
         \alpha\otimes S(\alpha) \alpha-  \gamma\otimes \otimes S(\gamma^\ast)\alpha-\alpha^\ast\otimes S(\gamma^\ast)\gamma-  \gamma^\ast\otimes S(\alpha)\gamma
          \\
        &=&
        \alpha\otimes \alpha\,\alpha^\ast+  \gamma\otimes \alpha\,\gamma^\ast +\alpha^\ast\otimes \gamma\,\gamma^\ast -  \gamma^\ast\otimes \alpha^\ast\,\gamma,
    \end{eqnarray*}
\begin{eqnarray*}
        \Ad(\alpha^\ast)&=&  \alpha^\ast\otimes S(\alpha^\ast) \alpha^\ast-  \gamma^\ast\otimes S(\gamma)\alpha^\ast-\alpha\otimes S(\gamma)\gamma^\ast-  \gamma\otimes S(\alpha^\ast)\gamma^\ast
        \\
        &=&
      \alpha^\ast\otimes \alpha\,\alpha^\ast+  \gamma^\ast\otimes  \alpha^\ast\,\gamma+\alpha\otimes \gamma\,\gamma^\ast-  \gamma\otimes \alpha\,\gamma^\ast,
    \end{eqnarray*}
    and 
    \begin{eqnarray*}
        \Ad(\alpha-\alpha^\ast)&=&(\alpha-\alpha^\ast)\otimes (\alpha\,\alpha^\ast- \gamma\,\gamma^\ast)+\gamma\otimes 2\,\alpha\,\gamma^\ast+\gamma^\ast \otimes -2\,\alpha^\ast\,\gamma,
    \end{eqnarray*}
    $$\Ad(\alpha+\alpha^\ast)= (\alpha+\alpha^\ast)\otimes \mathbbm{1}.$$
    Moreover,
    \begin{eqnarray*}
        \Ad(\gamma)
        &=&
        \alpha\otimes S(\gamma)\alpha+  \gamma\otimes S(\alpha^\ast)\alpha+\alpha^\ast\otimes S(\alpha^\ast)\gamma+  \gamma^\ast\otimes S(\gamma)\gamma
        \\
        &=&
        -\alpha\otimes \alpha\,\gamma+  \gamma\otimes  \alpha^2 
        +
        \alpha^\ast\otimes \alpha\,\gamma-  \gamma^\ast\otimes   \gamma^2 
        \\
        &=&
        (\alpha-\alpha^\ast) \otimes -\alpha\,\gamma +\gamma \otimes \alpha^2 +\gamma^\ast \otimes \gamma^2,
    \end{eqnarray*}
\begin{eqnarray*}
        \Ad(\gamma^\ast)
        &=&
          \alpha^\ast\otimes S(\gamma^\ast)(\alpha^\ast)+  \gamma^\ast\otimes S(\alpha)\alpha^\ast
        +
         \alpha\otimes S(\alpha)\gamma^\ast-  \gamma\otimes S(\gamma^\ast)\gamma^\ast
         \\
         &=&
         -\alpha^\ast\otimes  \alpha^\ast\,\gamma^\ast +  \gamma^\ast\otimes  \alpha^{\ast 2}
         +
         \alpha\otimes  \alpha^\ast\,\gamma^\ast+  \gamma\otimes   \gamma^{\ast 2}
         \\
         &=&
         (\alpha-\alpha^\ast)\otimes \alpha^\ast\,\gamma^\ast+  \gamma\otimes   \gamma^{\ast 2}+  \gamma^\ast\otimes  \alpha^{\ast 2}.
    \end{eqnarray*}
Therefore $$\ad(\pi(\alpha-\alpha^\ast))=\pi(\alpha-\alpha^\ast)\otimes (\alpha\,\alpha^\ast- \gamma\,\gamma^\ast)+\pi(\gamma)\otimes 2\,\alpha\,\gamma^\ast+\pi(\gamma^\ast) \otimes -2\,\alpha^\ast\,\gamma,$$ $$\ad(\pi(\alpha+\alpha^\ast))=\pi(\alpha+\alpha^\ast)\otimes \mathbbm{1},$$ $$\ad(\pi(\gamma))= \pi(\alpha-\alpha^\ast) \otimes -\alpha\,\gamma +\pi(\gamma) \otimes \alpha^2 +\pi(\gamma^\ast) \otimes \gamma^2,$$ $$\ad(\pi(\gamma^\ast))=\pi(\alpha-\alpha^\ast)\otimes \alpha^\ast\,\gamma^\ast+  \pi(\gamma)\otimes   \gamma^{\ast 2}+  \pi(\gamma^\ast)\otimes  \alpha^{\ast 2}$$ and the proposition follows by considering the definitions of $\eta_1,$ $\eta_2$, $\eta_3$, $\eta_4$.
\end{proof}

Next proposition follows directly by a straightforward calculation using Propositions \ref{prop2.4}, \ref{prop2.5}, so we will omit its proof.

\begin{Proposition}
    \label{prop2.6}
    The following relations hold
    $$c^T(\eta_1)=-{i\,q_4\over 2}\,\eta_1\otimes\eta_4+i\,q_w\,(\eta_1\otimes\eta_3-\eta_3\otimes \eta_1),$$ $$c^T(\eta_2)=-{i\,3\,q_4\over 2}\,\eta_2\otimes \eta_4+i\,q_w\,(\eta_3\otimes\eta_2-\eta_2\otimes\eta_3),$$ $$c^T(\eta_3)=-i\,q_4\,\eta_3\otimes\eta_4+i\,q_w\,(\eta_2\otimes \eta_1-\eta_1\otimes\eta_2),$$ $$c^T(\eta_4)=0.$$
\end{Proposition}

Let us consider the universal differential envelope $\ast$--calculus
\begin{equation}
    \label{ec.2.95}
    (\Gamma^\wedge,d,\ast)
\end{equation}
of $(\Gamma,d)$ (see reference \cite{micho1,stheve} or Section 2.1 for a brief summary).

\begin{Proposition}
    \label{prop2.7}
    The following relations hold $$\pi(\gamma)^2=\pi(\gamma^\ast)^2=\pi(\alpha+\alpha^\ast)^2=0, $$ $$\pi(\alpha-\alpha^\ast)^2={1\over 2}\left(\pi(\alpha+\alpha^\ast)\pi(\alpha-\alpha^\ast)+\pi(\alpha-\alpha^\ast)\pi(\alpha+\alpha^\ast) \right),$$  $$\pi(\gamma^\ast)\pi(\gamma)=-\pi(\gamma)\pi(\gamma^\ast), $$ $$\pi(\gamma)\pi(\alpha-\alpha^\ast)=-\pi(\alpha-\alpha^\ast)\pi(\gamma),\quad \pi(\gamma^\ast)\pi(\alpha-\alpha^\ast)=-\pi(\alpha-\alpha^\ast)\pi^\ast(\gamma^\ast),$$ $$\pi(\gamma)\pi(\alpha+\alpha^\ast)=-\pi(\alpha+\alpha^\ast)\pi(\gamma),\quad \pi(\gamma^\ast)\pi(\alpha+\alpha^\ast)=-\pi(\alpha+\alpha^\ast)\pi^\ast(\gamma^\ast).$$ 
\end{Proposition}

\begin{proof}
    Consider the two--side ideal $S^\wedge$ of equation (\ref{ec.2.5}). Then, the proof of this proposition consists in a straightforward but tedious calculation applying the operator $(\pi\otimes \pi)\circ \Delta$ on every single  generator element of $\mathcal{R}_{4D}$. For example, $\gamma^2$ $\in$ $\mathcal{R}_3$; so 
    \begin{eqnarray*}
        (\pi\otimes \pi)\Delta(\gamma^2)&=&\pi(\gamma^2)\otimes \pi(\alpha^2)+\pi(\alpha^\ast\gamma)\otimes \pi(\alpha\gamma)+\pi(\alpha^\ast\gamma)\otimes \pi(\alpha\gamma)+\pi(\alpha^{\ast 2})\otimes \pi(\gamma^2)
        \\
        &=&
        \pi(\alpha^\ast\gamma)\otimes \pi(\alpha\gamma)+\pi(\alpha^\ast\gamma)\otimes \pi(\alpha\gamma)
        \\
        &=&
        \pi(\gamma)\otimes \pi(\gamma)+\pi(\gamma)\otimes \pi(\gamma).
    \end{eqnarray*}
    This implies that $\pi(\gamma)\otimes \pi(\gamma)+\pi(\gamma)\otimes \pi(\gamma)$ $\in$ $S^\wedge$ and hence, in $\Gamma^\wedge$, we get $$0=\pi(\gamma) \pi(\gamma)+\pi(\gamma) \pi(\gamma)=2\,\pi(\gamma)^2 \quad \Longrightarrow  \quad \pi(\gamma)^2=0.$$
\end{proof}

\noindent In addition

\begin{Proposition}
    \label{prop2.8}
    We have
    $$d\eta_1=-i\,q_w\,\eta_1\,\eta_3,$$
    $$d\eta_2=i\,q_w\,\eta_2\,\eta_3,$$ $$d\eta_3=i\,q_w\,\eta_1\,\eta_2+{i\,q_4\over 2}\,(\eta_3\,\eta_4+\eta_3\,\eta_4),$$ $$d\eta_4=-{i\,q_w\over 4}\,(\eta_3\,\eta_4+\eta_4\,\eta_3). $$
\end{Proposition}

\begin{proof}
    By equations (\ref{ec.2.8}), (\ref{ec.2.19}) and Proposition \ref{prop2.7}, we obtain 
    \begin{eqnarray*}
        d\pi(\gamma)&=&-\pi(\gamma)\,\pi(\alpha)-\pi(\alpha^\ast)\,\pi(\gamma)
        \\
        &=&
        -\pi(\gamma)\,({1\over 2}\pi(\alpha+\alpha^\ast)+{1\over 2}\pi(\alpha-\alpha^\ast))-({1\over 2}\pi(\alpha+\alpha^\ast)-{1\over 2}\pi(\alpha-\alpha^\ast))\,\pi(\gamma)
        \\
        &=&
      {1\over 2}(-\pi(\gamma)\,\pi(\alpha+\alpha^\ast)-\pi(\gamma)\,\pi(\alpha-\alpha^\ast)-\pi(\alpha+\alpha^\ast)\,\pi(\gamma)+\pi(\alpha-\alpha^\ast)\,\pi(\gamma)) 
      \\
        &=&
        -\pi(\gamma)\,\pi(\alpha-\alpha^\ast)
    \end{eqnarray*}
    and hence  $d\eta_1=-i\,q_w\,\eta_1\,\eta_3.$ 
    
    Similarly, 
   \begin{eqnarray*}
        d\pi(\gamma^\ast)&=&-\pi(\gamma^\ast)\,\pi(\alpha^\ast)-\pi(\alpha)\,\pi(\gamma^\ast)
        \\
        &=&
        -\pi(\gamma^\ast)\,({1\over 2}\pi(\alpha+\alpha^\ast)-{1\over 2}\pi(\alpha-\alpha^\ast))-({1\over 2}\pi(\alpha+\alpha^\ast)+{1\over 2}\pi(\alpha-\alpha^\ast))\,\pi(\gamma^\ast)
        \\
        &=&
      {1\over 2}(-\pi(\gamma^\ast)\,\pi(\alpha+\alpha^\ast)+\pi(\gamma^\ast)\,\pi(\alpha-\alpha^\ast)-\pi(\alpha+\alpha^\ast)\,\pi(\gamma^\ast)-\pi(\alpha-\alpha^\ast)\,\pi(\gamma^\ast)) 
      \\
        &=&
        \pi(\gamma^\ast)\,\pi(\alpha-\alpha^\ast)
    \end{eqnarray*}
    and therefore $d\eta_2=i\,q_w\,\eta_2\,\eta_3.$ 
    
    \noindent Moreover, 
    \begin{eqnarray*}
        d\pi(\alpha-\alpha^\ast)&=&-\pi(\alpha)\,\pi(\alpha)+\pi(\gamma^\ast)\,\pi(\gamma)+\pi(\alpha^\ast)\,\pi(\alpha^\ast)-\pi(\gamma)\,\pi(\gamma^\ast)
        \\
        &=&
        -({1\over 2}\pi(\alpha+\alpha^\ast)+{1\over 2}\pi(\alpha-\alpha^\ast))({1\over 2}\pi(\alpha+\alpha^\ast)+{1\over 2}\pi(\alpha-\alpha^\ast))-2\,\pi(\gamma)\,\pi(\gamma^\ast)
        \\
        &+&
        ({1\over 2}\pi(\alpha+\alpha^\ast)-{1\over 2}\pi(\alpha-\alpha^\ast))({1\over 2}\pi(\alpha+\alpha^\ast)-{1\over 2}\pi(\alpha-\alpha^\ast))
        \\
        &=&
        -{1\over 2}(\pi(\alpha+\alpha^\ast)\,\pi(\alpha-\alpha^\ast)+\pi(\alpha-\alpha^\ast)\,\pi(\alpha+\alpha^\ast))-2\,\pi(\gamma)\,\pi(\gamma^\ast)
    \end{eqnarray*}
    and then  $d\eta_3=\displaystyle i\,q_w\,\eta_1\,\eta_2+{i\,q_4\over 2}\,(\eta_3\,\eta_4+\eta_3\,\eta_4).$ 

    \noindent Finally, 
    \begin{eqnarray*}
        d\pi(\alpha+\alpha^\ast)&=&-\pi(\alpha)\,\pi(\alpha)+\pi(\gamma^\ast)\,\pi(\gamma)-\pi(\alpha^\ast)\,\pi(\alpha^\ast)+\pi(\gamma)\,\pi(\gamma^\ast)
        \\
        &=&
        -({1\over 2}\pi(\alpha+\alpha^\ast)+{1\over 2}\pi(\alpha-\alpha^\ast))({1\over 2}\pi(\alpha+\alpha^\ast)+{1\over 2}\pi(\alpha-\alpha^\ast))
        \\
        &-&
        ({1\over 2}\pi(\alpha+\alpha^\ast)-{1\over 2}\pi(\alpha-\alpha^\ast))({1\over 2}\pi(\alpha+\alpha^\ast)-{1\over 2}\pi(\alpha-\alpha^\ast))
        \\
        &=&
        -{1\over 2}\,\pi(\alpha-\alpha^\ast)^2
        \\
        &=&
        -{1\over 4}\,(\pi(\alpha+\alpha^\ast)\,\pi(\alpha-\alpha^\ast)+\pi(\alpha-\alpha^\ast)\,\pi(\alpha+\alpha^\ast)).
    \end{eqnarray*}
    Thus  $d\eta_4=\displaystyle-{i\,q_w\over 4}\,(\eta_3\,\eta_4+\eta_4\,\eta_3).$ 
\end{proof}

In differential geometry, the quantum Lie bracket $c^T$ is actually the pull--back of the Lie bracket of the Lie algebra $\mathfrak{g}$ of a Lie group $G$ (\cite{appendix}) and in this case, it satisfies
$$\ad^{\#\otimes 2}\circ c^T=(c^T\otimes \id)\circ \ad^\#$$ where $\ad^\#$ is the pullback of the adjoint action of the Lie group on its Lie algebra, and $\ad^{\#\otimes 2}$ is the corepresentation tensor product of $\ad^\#$ with itself. Furthermore, the Maurer--Cartan 
equation (\cite{nodg}) is expressed as
$$d\theta=m\left(-{1\over 2}c^T(\theta)\right),$$ where $\theta$ $\in$ $\mathfrak{g}^\#$ (the dual space of $\mathfrak{g}$) and $m$ is the product map of differential forms.

In our case, it is easy to verify that (\cite{stheve})
$$\ad^{\otimes 2}\circ c^T=(c^T\otimes \id)\circ \ad,$$ where $\ad^{\otimes 2}$ is the corepresentation tensor product of $\ad$ with itself; however, by comparing Proposition \ref{prop2.6} with Proposition \ref{prop2.8}, it is clear that the Maurer–Cartan equation is no longer satisfied:
\begin{equation}
    \label{ec.2.96}
    d\theta\not=m_{\Gamma^\wedge}\left(-{1\over 2}c^T(\theta)\right),
\end{equation}
where $$m_{\Gamma^\wedge}:\Gamma^\wedge\otimes \Gamma^\wedge\longrightarrow \Gamma^\wedge$$ is the product map. Equation (\ref{ec.2.96}) motivates the following definition.

\begin{Definition}
    \label{def2.7}
    We define an embedded differential as a linear map $$\Theta: \su(2,\C)^\#_{4D}\longrightarrow  \su(2,\C)^\#_{4D}\otimes \su(2,\C)^\#_{4D}$$ such that
    \begin{enumerate}
        \item $\ad^{\otimes 2}\circ \Theta=(\Theta\otimes \id)\circ \ad$, where $\ad^{\otimes 2}$ is the corepresentation tensor product of $\ad$ with itself.
        \item $d\theta=m_{\Gamma^\wedge}(\Theta(\theta))$  for all $\theta$ $\in$ $\su(2,\C)^\#_{4D}$.
    \end{enumerate}
\end{Definition}

Notice that in differential geometry, $\displaystyle -{1\over 2}c^T$ is an embedded differential. Hence, we can regard an embedded differential $\Theta$ as a generalization of $\displaystyle -{1\over 2}c^T$, in order to preserve the Maurer–Cartan equation. In other words, an embedded differential provides {\it a way to correctly combine the differential structure with the Lie algebra structure} in the non--commutative geometrical setting. 

In the general case, there may exist multiple embedded differentials associated with a given $\ast$--Hopf algebra and a given bicovariant $\ast$--FODC. Fortunately, in some cases there exists a \emph{kind of canonical} embedded differential.

In light of Section 5 of reference \cite{micho1}, if $L\subseteq \Ker(\epsilon)$ is a $\ast$--$S$--invariant $\Ad$--invariant complement of $\mathcal{R}_{4D}$, then $$\Theta=-(\pi\otimes \pi)\circ \Delta \circ \pi^{-1}|_L $$ defines an embedded differential. In this way, taking into account the  basis $\beta^\#_{4D}$ of equation (\ref{ec.2.89}),  there is a {\it canonical} space $L$ given by 
\begin{equation}
    \label{ec.2.97}
    L=\mathrm{span}_{\C}\{\gamma,\,\gamma^\ast,\,\alpha-\alpha^\ast,\,\alpha+\alpha^\ast \}.
\end{equation}
Indeed, since $\beta^\#_{4D}$ is a linear basis of $\su(2,\C)^\#_{4D}$, then, by definition, we have $$L\oplus \mathcal{R}_{4D}=\Ker(\epsilon).$$
Our calculations of Propositions \ref{prop2.5}  show that $$\Ad(L)\subseteq L\otimes \SU(2)$$ and it is easy to see that $$S(L)^\ast\subseteq L.$$ Therefore, for our case, in which we are interesting in working with the basis $\beta^\#_{4D}$, there is a {\it canonical} embedded differential
\begin{equation}
    \label{ec.2.98}
    \Theta:  \su(2,\C)^\#_{4D}\longrightarrow  \su(2,\C)^\#_{4D}\otimes \su(2,\C)^\#_{4D}
\end{equation}
given by
\begin{equation}
    \label{ec.2.99} 
    \Theta(\eta_1):={i\,\sqrt{2}\over q_w}\,(\pi\otimes \pi)\Delta(\gamma)={i\,q_4\over 2}\,(\eta_1\otimes \eta_4+\eta_4\otimes \eta_1)+{i\,q_w\over 2}\,(\eta_3\otimes \eta_1- \eta_1\otimes\eta_3),
\end{equation}
\begin{equation}
    \label{ec.2.100} 
    \Theta(\eta_2):=-{i\,\sqrt{2}\over q_w}\,(\pi\otimes \pi)\Delta(\gamma^\ast)={i\,q_4\over 2}\,(\eta_2\otimes \eta_4+\eta_4\otimes \eta_2)+{i\,q_w\over 2}\,(\eta_2\otimes \eta_3- \eta_3\otimes \eta_2),
\end{equation}
\begin{equation}
    \label{ec.2.101} 
    \Theta(\eta_3):={i\over q_w}\,(\pi\otimes \pi)\Delta(\alpha-\alpha^\ast)={i\,q_4\over 2}\,(\eta_3\otimes \eta_4+\eta_4\otimes \eta_3)+{i\,q_w\over 2}\,(\eta_1\otimes \eta_2- \eta_2\otimes \eta_1),
\end{equation}
\begin{equation}
    \label{ec.2.102} 
    \Theta(\eta_4):=-{i\over q_4}\,(\pi\otimes \pi)\Delta(\alpha+\alpha^\ast)={i\,q_4\over 2}\,\eta_4\otimes \eta_4+{i\,q^2_w\over 2\,q_4}\,(\eta_1\otimes \eta_2+\eta_2\otimes \eta_1+\eta_3\otimes \eta_3).
\end{equation}

\begin{Definition}
    \label{def2.8}
    We define the inner product $$\langle -|-\rangle_{4D}: \su(2,\C)^\#_{4D}\times \su(2,\C)^\#_{4D}\longrightarrow \C$$ such that $\beta^\#_{4D}$  is an orthonormal basis of  $\su(2,\C)^\#_{4D}$. 
\end{Definition}
\noindent 

Thus, we have
\begin{Proposition}
    \label{prop2.9}
    The $\SU(2)$--corepresentation $\ad$ on $\su(2,\C)^\#_{4D}$ is unitary with respect to $\langle -|-\rangle_{4D}$.
\end{Proposition}
\begin{proof}
    By Proposition \ref{prop2.5}, we get that the matrix associated with $\ad$ with respect to the basis $\beta^\#_{4D}$ is
   $$U=\left(
\begin{array}{cccc}  
\alpha^2 &  -\gamma^{\ast 2} & \sqrt{2}\,\alpha\,\gamma^\ast &0 \\
-\gamma^2 & \alpha^{\ast 2} & \sqrt{2}\,\alpha^\ast\,\gamma &0 \\
-\sqrt{2}\,\alpha\,\gamma & -\sqrt{2}\,\alpha^\ast\,\gamma^\ast & \alpha\alpha^\ast-\gamma\gamma^\ast & 0\\
0 & 0 &0 & 1
\end{array}
\right)$$
An straightforward calculation shows that $$U\,U^\dagger=U^\dagger\,U=\Id_3,$$ where $\Id_3$ is the identity matrix of dimension $3$ and $U^\dagger$ is composition of the usual matrix transpose operation and the $\ast$ operation of $\G$. Hence, $\ad$ is unitary with respect to $\langle-|-\rangle_{4D}$.
\end{proof}

Notice that the linear basis of $\su(2,\C)^\#_{4D}$
\begin{equation}
    \label{ec.2.103}
    \widetilde{\beta}^\#_{4D}:=\{ \beta_1:={1\over \sqrt{2}}(\eta_1+\eta_2),\quad \beta_2:={i\over \sqrt{2}}(\eta_1-\eta_2),\quad \beta_3:=\eta_3,\quad \beta_4:=\eta_4\}
\end{equation}
is also an orthonormal basis with respect to $\langle -|-\rangle_{4D}$ and a simple change of basis proves the next propositions
\begin{Proposition}
    \label{prop2.3.2}
    The following relations hold
    $$\Theta(\beta_1)={i\,q_4\over 2}\,(\beta_1\otimes \beta_4+\beta_4\otimes \beta_1)-{q_w\over 2}\,(\beta_2\otimes \beta_3- \beta_3\otimes\beta_2),$$ $$\Theta(\beta_2)={i\,q_4\over 2}\,(\beta_2\otimes \beta_4+\beta_4\otimes \beta_2)-{q_w\over 2}\,(\beta_3\otimes \beta_1- \beta_1\otimes \beta_3),$$
    $$\Theta(\beta_3)={i\,q_4\over 2}\,(\beta_3\otimes \beta_4+\beta_4\otimes \beta_3)-{q_w\over 2}\,(\beta_1\otimes \beta_2- \beta_2\otimes \beta_1),$$ $$\Theta(\beta_4)={i\,q_4\over 2}\,\beta_4\otimes \beta_4+{i\,q^2_w\over 2\,q_4}\,(\beta_1\otimes \beta_1+ \beta_2\otimes \beta_2+\beta_3\otimes \beta_3).$$ 
\end{Proposition}
\begin{Proposition}
    \label{prop2.3.3}
    The following relations hold
    $$c^T(\beta_1)=-i\,q_4\,\beta_1\otimes \beta_4+{q_4\over 2}\,\beta_2\otimes \beta_4-q_w\,(\beta_3\otimes\beta_2-\beta_2\otimes\beta_3),$$ $$c^T(\beta_2)=-{\,q_4\over 2}\,\beta_1\otimes \beta_4-i\,q_4\,\beta_2\otimes\beta_4-q_w\,(\beta_1\otimes\beta_3-\beta_3\otimes\beta_1),$$ $$c^T(\beta_3)=-i\,q_4\,\beta_3\otimes \beta_4-q_w\,(\beta_2\otimes\beta_1-\beta_1\otimes\beta_2),$$ $$c^T(\beta_4)=0.$$
\end{Proposition}

It is well--known that there exists another $4$--dimensional differential calculus on $\SU(2)$ (see references \cite{woro1,stach}). The one presented in this subsection is known as the $4D_+$--differential calculus, while the other one is known as the $4D_-$--differential calculus. We have chosen to work with the $4D_+$--differential calculus rather than the $4D_-$--differential calculus because $4D_+$ has the following property: let us denote by $d_{dR}$ the de--Rham differential of $SU(2)$  and let us denote by $d$ the differential of the $\ast$--FODC of this subsection. In light Section 9 of reference \cite{stach}, we get
\begin{equation}
    \label{ec.2.104}
    dg=d_{dR}(g)+w\,\triangle(g)\,\eta_4
\end{equation}
for every $g$ $\in$ $\SU(2),$ with $w$ $\in$ $\C$. Here, $\triangle(g)$ denotes the Laplacian of the function $g$. In the $4D_-$--differential calculus of $\SU(2)$, equation (\ref{ec.2.104}) does not hold. Equation~(\ref{ec.2.104}) shows that the $4D_+$--differential calculus on $\SU(2)$ is closely related to its {\it classical} counterpart.

\section{The Principal Bundle}

In the {\it classical} case, the electroweak theory is developed in $\R^4$. Since $\R^4$ is contractible, every bundle on it is trivializable and therefore, in this paper, we also use a trivial quantum principal bundle.

\subsection{General Theory}

Durdevich's formulation of quantum principal bundles is presented in \cite{micho1,micho2,micho3,stheve} among other references. Let us present a brief summary.

Let  $(A,\cdot,\mathbbm{1},\ast,\Delta,\epsilon,S)$ be a $\ast$--Hopf algebra.  In Durdevich's formulation, a quantum principal $A$--bundle (abbreviated qpb) is a triple 
\begin{equation}
    \label{ec.3.1}
    \zeta=(P,B,\Delta_P),
\end{equation}
where $(P,\cdot,\mathbbm{1},\ast)$ is a $\ast$--algebra called the {\it the quantum total space}, $(B,\cdot,\mathbbm{1},\ast)$ is a $\ast$--subalgebra called {\it the quantum base space} and $$\Delta_P:P\longrightarrow P\otimes A $$ is a $\ast$--algebra morphism that satisfies
\begin{enumerate}
\item $\Delta_P$ is a $A$--corepresentation.
\item $\Delta_P(x)=x\otimes \mathbbm{1}$ if and only if $x$ $\in$ $B$.
\item The linear map $\beta:P\otimes P\longrightarrow P\otimes G$ given by $$\beta(x\otimes y):=x\cdot \Delta_P(y):=(x\otimes \mathbbm{1})\cdot \Delta_P(y) $$ is surjective. 
\end{enumerate}
 
Given $\zeta$ a qpb, a {\it differential calculus} on it is:
 \begin{enumerate}
 \item A graded differential $\ast$--algebra $(\Omega^\bullet(P),d,\ast)$ generated by $\Omega^0(P)=P$ ({\it quantum differential forms of $P$}).
 \item  A bicovariant $\ast$--FODC $(\Gamma,d)$ over $A$.
 \item The map $\Delta_P$ is extensible to a graded differential $\ast$--algebra morphism $$\Delta_{\Omega^\bullet(P)}:\Omega^\bullet(P)\longrightarrow \Omega^\bullet(P)\otimes \Lambda^{\wedge},$$ where $(\Lambda^\wedge,d,\ast)$ is the universal differential envelope $\ast$--calculus. Here, we have considered that $\otimes$ is the tensor product of graded differential $\ast$--algebras.
 \end{enumerate}

Notice that if $\Delta_{\Omega^\bullet(P)}$ exists, then it is unique because all our graded differential $\ast$--algebras are generated by their degree $0$ elements.

Let $\zeta$ be a qpb with a differential calculus. The space 
\begin{equation}
    \label{ec.3.2}
    \Hor^\bullet\,P:=\{ \psi \in \Omega^\bullet(P)\mid \Delta_{\Omega^\bullet(P)}(\psi)\in \Omega^\bullet(P)\otimes A\}
\end{equation}
is a graded $\ast$--algebra called the space of horizontal forms. Additionally, we define the map
\begin{equation}
    \label{ec.3.3}
    \Delta_\Hor:=\Delta_{\Omega^\bullet(P)}|_{\Hor^\bullet\,P}: \Hor^\bullet\,P\longrightarrow \Hor^\bullet\,P\otimes A,
\end{equation}
which, according to \cite{stheve}, is an $A$--corepresentation.

On the other hand, the space 
\begin{equation}
    \label{ec.3.4}
    \Omega^\bullet(B):=\{ \psi \in \Omega^\bullet(P)\mid \Delta_{\Omega^\bullet(P)}(\psi)=\psi\otimes \mathbbm{1}\}
\end{equation}
is a graded differential $\ast$--subalgebra of $(\Omega^\bullet(P),d,\ast)$ which is called the space of base forms. Moreover, we define a quantum principal connection (abbreviated qpc) on $\zeta$ as a linear map 
\begin{equation}
    \label{ec.3.5}
    \omega:\mathfrak{qa}^\#\longrightarrow \Omega^1(P)
\end{equation}
such that $$\Delta_{\Omega^\bullet(P)}(\omega(\theta))=(\omega\otimes \id_A)\ad(\theta)+\mathbbm{1}\otimes \theta, \qquad \omega(\theta^\ast)=\omega(\theta)^\ast $$ for all $\theta$ $\in$ $\mathfrak{qa}^\#$ (see equation (\ref{adgeneral})). The space of all qpc will be denoted by
\begin{equation}
    \label{ec.3.6}
    \mathfrak{qpc}(\zeta).
\end{equation}
In addition, according to reference \cite{stheve}, $\mathfrak{qpc}(\zeta)$ is a real affine space modeled by the $\R$--vector space
\begin{equation}
    \label{ec.3.7}
    \overrightarrow{\mathfrak{qpc}(\zeta)}:=\{\lambda:\mathfrak{qa}^\# \longrightarrow \Hor^1 P\mid \lambda \mbox{ is linear and }
    (\lambda\otimes \id_A)\circ \ad=\Delta_\Hor \circ \lambda,\;\; \lambda\circ \ast=\ast\circ \lambda \}.
\end{equation}

A qpc $\omega$ is called regular if 
\begin{equation}
    \label{ec.3.7.1}
    \omega(\theta)\,\varphi=(-1)^k\,\varphi^{(0)}\,\omega(\theta\diamondsuit \varphi^{(1)}),
\end{equation}
for all $\varphi$ $\in$ $\Hor^k\,P$, where $\Delta_\Hor(\varphi)=\varphi^{(0)}\otimes \varphi^{(1)}$ and $\theta\diamondsuit g$ is defined as in equation (\ref{ec.2.91}).

For a given qpc $\omega$, we define its covariant derivative as the first--order linear map 
\begin{equation}
    \label{ec.3.8}
    D^\omega: \Hor^\bullet\,P\longrightarrow \Hor^\bullet\,P
\end{equation}
such that for all $\varphi$ $\in$ $\Hor^k\,P$ we have $$D^\omega(\varphi)=d\varphi-(-1)^k\,\varphi^{(0)}\,\omega(\pi(\varphi^{(1)})),$$ with $\Delta_\Hor(\varphi)=\varphi^{(0)}\otimes\varphi^{(1)}$ and $\pi$ being the quantum germs map (see equation (\ref{qgermsmap}). Additionally, we define the dual covariant derivative as the first--order linear map 
\begin{equation}
    \label{ec.3.8.1}
    \widehat{D}^\omega:=\ast \circ D^\omega \circ \ast: \Hor^\bullet\,P\longrightarrow \Hor^\bullet\,P.
\end{equation}
In concrete
\begin{equation}
    \label{ec.3.8.2}
    \widehat{D}^\omega(\varphi)=d\varphi+\omega(\pi(S^{-1}(\varphi^{(1)})))\,\varphi^{(0)}.
\end{equation}
It is worth mentioning that, according to \cite{micho3}, we get
\begin{equation}
    \label{ec.3.8.3}
    \widehat{D}^\omega=D^\omega\;\;\Longleftrightarrow \;\; \omega \;\mbox{ is regular.} 
\end{equation}
In fact, for every $\varphi$ $\in$ $\Hor^k\,P$, we obtain
\begin{equation}
\label{iguales}
    \widehat{D}^{\omega}(\varphi)=D^{\omega}(\varphi)+\ell^{\omega}(\pi(S^{-1}(\varphi^{(1)})),\varphi^{(0)}),
\end{equation}
where 
\begin{equation*}
\ell^{\omega}:\mathfrak{qg}^\#\times \Hor^\bullet P \longrightarrow \Hor^\bullet P, \qquad
(\theta,\varphi)\longmapsto \omega(\theta)\varphi-(-1)^k \varphi^{(0)}\omega(\theta\diamondsuit \varphi^{(1)}).
\end{equation*}

 The map $\ell^{\omega}$ {\it measures the degree of non--regularity} of $\omega$, in the sense of $\ell^{\omega}=0$ if and only if $\omega$ is regular \cite{micho3}. For regular qpc's we get $D^\omega=\widehat{D}^{\omega}$, which is the situation for qpc's arising from {\it classical} principal connections (\cite{micho1}). In other words, $D^\omega$ and $\widehat{D}^\omega$ are two different horizontal operators that generalize the covariant derivative of a principal connection in differential geometry. In the next sections, we will work with both operators.

Direct calculations shows that (\cite{micho3})
\begin{equation}
\label{2.f32}
    D^{\omega}(\varphi\psi)=D^{\omega}(\varphi)\psi+(-1)^k\varphi D^{\omega}(\psi)+ (-1)^k \varphi^{(0)}\ell^{\omega}(\pi(\varphi^{(1)}),\psi),
\end{equation}
and 
\begin{equation}
\label{2.f32.1}
\widehat{D}^{\omega}(\varphi\psi)=\widehat{D}^{\omega}(\varphi)\psi+(-1)^k\varphi \widehat{D}^{\omega}(\psi)
+ \ell^{\omega}(\pi(S^{-1}(\psi^{(1)}))\circ S^{-1}(\varphi^{(1)}),\varphi^{(0)})\psi^{(0)},
\end{equation}
for all $\varphi$ $\in$ $\Hor^k P$, $\psi$ $\in$ $\Hor^\bullet P$.

The notion of embedded differential (see Definition~\ref{def2.7}) can be defined for every bicovariant $\ast$--FODC over any $\ast$--Hopf algebra; this concept is not exclusive to the $\ast$--FODC presented in the previous section \cite{micho1,micho2,stheve}. In this way, in Durdevich's formulation of qpb's, if $\Theta$ is an embedded differential, we define the curvature of $\omega$ as the linear map
\begin{equation*}
    R^\omega:\mathfrak{qa}^\#\longrightarrow \Omega^2(P)
\end{equation*}
given by 
\begin{equation}
    \label{ec.3.10}
    R^\omega=d\omega-\langle\omega,\omega\rangle,\qquad \langle\omega,\omega\rangle:=m_\Omega\circ (\omega\otimes \omega)\circ \Theta,
\end{equation}
where $m_\Omega:\Omega^\bullet(P)\otimes \Omega^\bullet(P)\longrightarrow \Omega^\bullet(P)$ is the product map. 

In differential geometry, the image of the curvature of a principal connection is the Lie algebra $\mathfrak{g}$ of the structure group $G$ of the bundle, so in non--commutative geometry, it is natural to think that the the domain of the curvature of a quantum principal connection has to be the quantum dual Lie algebra $\mathfrak{qa}^\#$. This is why in Durdevich's formulation of qpb’s,  the curvature is defined by equation (\ref{ec.3.10}) and not as in reference \cite{libro}: 
 \begin{equation}
    \label{ec.3.11}
    r^\omega: \Ker(\epsilon)\longrightarrow \Omega^2(P),\qquad g\longmapsto d\omega(\pi(g))+\omega(\pi(g^{(1)}))\,\omega(\pi(g^{(2)})) 
\end{equation}
with $\Delta(g)=g^{(1)}\otimes g^{(2)}$. 

Let $\omega_\class$ be a principal connection on a principal $G$--bundle, and let $\Omega^{\omega_\class}$ be its curvature. If $\omega$ denotes the pullback of $\omega_\class$, then the pullback of $\Omega^{\omega_\class}$ is given by 
\begin{equation}
    \label{ec.3.12}
    d\omega+{1\over 2}[\omega,\omega],\qquad [\omega,\omega]:=m_\Omega\circ (\omega\otimes \omega)\circ c^T.
\end{equation}
Thus, one could consider equation~(\ref{ec.3.12}) as the definition of the curvature of a qpc $\omega$. However, due to equation~(\ref{ec.2.96}) and the discussion below Definition~\ref{def2.7}, it should be clear that, in general, the curvature of a qpc cannot be defined in terms of the transpose commutator $c^T$; instead, it must be defined through an embedded differential $\Theta$. This is why in Durdevich's formulation of qpb’s, the curvature is defined by equation~(\ref{ec.3.10}), which can be viewed as a \emph{non--commutative geometrical generalization} of equation~(\ref{ec.3.12}). The reader can consult references \cite{micho2,stheve} for more details about this definition of the curvature. 
 
 It is worth mentioning that in differential geometry, the curvature $\Omega^{\omega_\class}$  of $\omega_\class$ is a basic differential form of type $\ad^\#$, where $\ad^\#$ is the adjoint action of $G$ on its Lie algebra $\mathfrak{g}$ (\cite{nodg}). The non--commutative geometrical counterpart of the space of basic differential forms of type $\ad^\#$ is the space
\begin{equation}
    \label{ec.3.13}
    \begin{aligned}
        \Mor(\ad,\Delta_\Hor):=\{\tau:\; &\mathfrak{qa}^\#\longrightarrow \Hor^\bullet\,P\mid
        \\ 
        &\tau \mbox{ is linear such that } (\tau\otimes \id_{\mathfrak{qa}^\#})\circ \ad=\Delta_\Hor\circ \tau\}
    \end{aligned}
\end{equation}
and according to Section 12.8 of reference \cite{stheve}, we have 
\begin{equation}
    \label{ec.3.14}
    \Im(R^\omega)\subseteq \Hor^2\,P, \qquad R^\omega\in \Mor(\ad,\Delta_\Hor), \qquad R^\omega(\theta^\ast)=R^\omega(\theta)^\ast,
\end{equation}
for all $\theta$ $\in$ $\mathfrak{qa}^\#$. In other words, we can regard $R^\omega$ as a \emph{quantum basic differential form of type} $\ad$.

Let $(\X,m_\X,\mathbbm{1},\ast)$ be a $\ast$-algebra, where $m_\X:\X\otimes \X\longrightarrow \X$ is the product map. In general, given two linear maps $$L_1,\,L_2:\mathfrak{qa}^\#\longrightarrow \X,$$ we define the linear maps
\begin{equation}
    \label{brakets}
    [L_1,L_2]:=m_\X\circ (L_1\otimes L_2)\circ c^T,\qquad \langle L_1,L_2\rangle:= m_\Omega\circ (\omega\otimes \omega)\circ \Theta.
\end{equation}

Let us consider the operator (\cite{micho2,sald2})
\begin{equation}
    \label{ec.3.18}
    \begin{aligned}
        S^\omega:\Mor(\ad,\Delta_\Hor)\longrightarrow \Mor(\ad,\Delta_\Hor),\qquad
        \tau\longmapsto S^\omega(\tau)
    \end{aligned}
\end{equation}
given by 
\begin{equation}
    \label{ec.3.19}
    S^\omega(\tau):=\langle\omega,\tau\rangle-(-1)^k\langle \tau,\omega\rangle-(-1)^k[\tau,\omega]
\end{equation}
for every $\tau$ $\in$ $\Mor(\ad,\Delta_\Hor)$ such that $\Im(\tau)\subseteq \Hor^k P$; and the operator (\cite{micho2,sald2})
\begin{equation}
    \label{ec.3.17}
    \begin{aligned}
        D^\omega:\Mor(\ad,\Delta_\Hor)\longrightarrow \Mor(\ad,\Delta_\Hor),\qquad
        \tau\longmapsto D^\omega(\tau)
    \end{aligned}
\end{equation}
given by $$D^\omega(\tau)(\theta):=D^\omega(\tau(\theta))$$ for every $\theta$ $\in$ $\mathfrak{qa}^\#$. In particular, by Proposition 4.7 of reference \cite{micho2}, we get
$$D^\omega(\tau)=d\tau-(-1)^k[\tau,\omega].$$

\begin{Definition}
    \label{twisted}
    We define the twisted covariant derivative as the operator
    $$DS^\omega:\Mor(\ad,\Delta_\Hor)\longrightarrow \Mor(\ad,\Delta_\Hor), \qquad \tau \longmapsto DS^\omega(\tau):=D^\omega(\tau)-S^\omega(\tau).$$ Explicitly, we have (\cite{sald2}) $$DS^\omega(\tau):=d\tau-\langle\omega,\tau\rangle+(-1)^k\langle \tau,\omega\rangle $$ for every $\tau$ $\in$ $\Mor(\ad,\Delta_\Hor)$ such that $\Im(\tau)\subseteq \Hor^k P$.
\end{Definition}

In accordance with the analysis below Definition (\ref{def2.7}), the operator $S^\omega$ may be interpreted as a measure of the \emph{discrepancy between the differential structure and the Lie algebra structure} in the non--commutative geometrical setting. For example, in differential geometry, the graded differential $\ast$--algebras involved are the algebras of differential forms, which are graded--commutative and the only embedded differential is given by $\displaystyle -\frac{1}{2}c^T$. Consequently, $S^\omega\equiv0$; which reflects the \emph{compatibility between the differential structure and the Lie algebra structure} in the \emph{classical} case.

In the same sense, the operator $$DS^\omega$$
can be interpreted as the \emph{correct} covariant derivative on quantum basic  differential forms of type $\ad$  because it corrects  this \emph{discrepancy}.

In differential geometry, given a principal connection $\omega_\class$ of a principal $G$--bundle, its curvature $\Omega^{\omega_\class}$ satisfies the (second) Bianchi identity
\begin{equation}
    \label{ec.3.15}
    D^{\omega_\class}\Omega^{\omega_\class}=0,
\end{equation}
where $D^{\omega_\class}$ is the covariant derivative of $\omega_\class$ \cite{nodg}. In light of Proposition 4.9 of reference \cite{micho2}, for every qpc $\omega$, the non--commutative geometrical (second) Bianchi identity is given by 
\begin{equation}
    \label{ec.3.16}
DS^\omega(R^\omega)=\langle\omega,\langle\omega,\omega\rangle\rangle-\langle\langle\omega,\omega\rangle,\omega\rangle.
\end{equation}

\noindent Notice that in the \emph{classical} case (\cite{micho2}) $$S^\omega\equiv 0 \qquad \mbox{ and }\qquad \langle\omega,\langle\omega,\omega\rangle\rangle-\langle\langle\omega,\omega\rangle,\omega\rangle=0$$ and hence,  equation (\ref{ec.3.16}) turns into (the pull--back of) equation (\ref{ec.3.15}). In other words, equation (\ref{ec.3.16}) is indeed, a \emph{non--commutative geometrical generalization} of equation (\ref{ec.3.15}). 

It is worth mentioning that in non--commutative geometry, in general, the equation
\begin{equation}
    \label{notbianchi}
    D^\omega R^\omega=0
\end{equation}
\emph{does not hold}. In Remark \ref{rema3.2} we will discuss this fact in more detail in our particular case of study.

According to Section 12 of reference \cite{stheve}, since the following identity holds 
\begin{equation}
    \label{ec.3.19.1}
    (\ast\otimes \ast)\circ \ad=\ad\circ \ast,
\end{equation}
there is a canonical involution on $\Mor(\ad,\Delta_\Hor)$ given by
\begin{equation}
    \label{ec.3.19.2}
    \wedge: \Mor(\ad,\Delta_\Hor)\longrightarrow \Mor(\ad,\Delta_\Hor), \qquad \tau\longmapsto \widehat{\tau}:=\ast \circ \tau\circ \ast.
\end{equation}
In this way, the dual covariant derivative, as an operator of $\Mor(\ad,\Delta_\Hor)$, satisfies
\begin{equation}
    \label{ec.3.19.3}
    \widehat{D}^\omega=\wedge\circ D^\omega \circ \wedge: \Mor(\ad,\Delta_\Hor)\longrightarrow \Mor(\ad,\Delta_\Hor);
\end{equation}
and defining the dual $S^\omega$ operator as
\begin{equation}
    \label{ec.3.19.4}
    \widehat{S}^\omega=\wedge\circ S^\omega \circ \wedge: \Mor(\ad,\Delta_\Hor)\longrightarrow \Mor(\ad,\Delta_\Hor),
\end{equation}
we obtain the dual twisted covariant derivative
\begin{equation}
    \label{dualtwisted}
    \widehat{DS}^\omega=\wedge\circ DS^\omega \circ \wedge: \Mor(\ad,\Delta_\Hor)\longrightarrow \Mor(\ad,\Delta_\Hor).
\end{equation}

In light of Section 2.2 of reference \cite{sald2}, we have
\begin{equation}
    \label{ec.3.19.5}
    \widehat{DS}^\omega(\tau)=DS^\omega(\tau) 
\end{equation}
for all $\tau$ $\in$ $\Mor(\ad,\Delta_\Hor)$ such that 
\begin{equation}
    \label{ec.3.19.5.1}
    \widehat{\tau}=\tau.
\end{equation}
In other words, if we define
\begin{equation}
    \label{ec.3.19.5.2}
    \Mor(\ad,\Delta_\Hor)^\dagger:=\{\tau\in \Mor(\ad,\Delta_\Hor)\mid \widehat{\tau}=\tau\},
\end{equation}
we get
\begin{equation}
    \label{ec.3.19.5.3}
    DS^\omega=\widehat{DS}^\omega :\Mor(\ad,\Delta_\Hor)^\dagger\longrightarrow \Mor(\ad,\Delta_\Hor)^\dagger.
\end{equation}
In particular, 
\begin{equation}
    \label{ec.3.19.5.5}
    \overrightarrow{\mathfrak{qpc}(\zeta)}\subseteq \Mor(\ad,\Delta_\Hor)^\dagger, \qquad R^\omega\in\Mor(\ad,\Delta_\Hor)^\dagger
\end{equation}
and for instance, we obtain
\begin{equation}
    \label{ec.3.19.6}
    DS^\omega(R^\omega) =\widehat{DS}^\omega(R^\omega)=\langle\omega,\langle\omega,\omega\rangle\rangle-\langle\langle\omega,\omega\rangle,\omega\rangle.
\end{equation}
It is worth mentioning that, when $\omega$ is regular, we get (see Lemma 4.8 of reference \cite{micho2})
\begin{equation}
    \label{ec.3.19.7}
    S^\omega=\widehat{S}^\omega=0.
\end{equation}

\begin{Definition}
\label{3.def1}
    We define the quantum gauge group of the qpb $\zeta$ as
    $$\qGG=\{\F:(P\oplus\Omega^1(P))\longrightarrow (P\oplus\Omega^1(P))\mid \F  \mbox{ is a } \mbox{graded left } (B\oplus \Omega^1(B))-\mbox{module isomorphism  }$$   $$\mbox{such that}\;\;\;\F(\mathbbm{1})=\mathbbm{1}, \;\;\; \Delta_{\Omega^\bullet(P)}\circ \F=(\F\otimes \id_{\Gamma^\wedge})\circ \Delta_{\Omega^\bullet(P)}$$ $$ 
     \mbox{ and }\;\;\F(\Im(\omega)^\ast)=\F(\Im(\omega))^\ast \;\mbox{ for all }\; \omega \in \mathfrak{qpc}(\zeta) \}.$$
\end{Definition}

It is worth mentioning that in general, elements of $\qGG$ do not commute with the differential of $\Omega^\bullet(P)$.

The {\it a priori} motivation for our definition of $\qGG$ is the fact that in differential geometry, gauge transformations are vertical principal bundle automorphisms. In this way, the definition \ref{3.def1} was derived by {\it dualizing} this {\it classical} fact, while ensuring that $\qGG$ is defined for every degree in the most general manner, without imposing any unnecessary condition except to guaranty an action of $\qGG$ on $\mathfrak{qpc}(\zeta)$.

According to \cite{sald1}, $\qGG$ is isomorphic to a subgroup of the group of all convolution--invertible maps $$\f:A\oplus \Lambda\longrightarrow P\oplus \Omega^1(P)$$ such that $$\f(\mathbbm{1})=\mathbbm{1}\quad \mbox{ and }(\f\otimes \id_{\Gamma^\wedge})\circ \Ad=\Delta_{\Omega^\bullet(P)}\circ \f,$$ where the map $\Ad$ is given in equation (\ref{ec.2.13}). Elements of $\qGG$ are called {\it quantum gauge transformations}. In Durdevich's framework of qpb's, the action of $\qGG$ on $\mathfrak{qpc}(\zeta)$ given by (\cite{libro}) $$ \omega \longmapsto \f^{-1}\ast \omega \ast \f+ \f^{-1}\ast (d\circ \f)$$ is not well--defined. Furthermore, even if we extend the domain of $\omega$ to $A$, the induced action on the curvature (\cite{libro}) $$\f^{-1}\ast R^\omega \ast \f $$ remains ill-defined. However, in accordance with Theorem 4.7 points 1 and 2 of reference \cite{sald1}, $\qGG$ has a well--defined group action on  $\mathfrak{qpc}(\zeta)$ by 
\begin{equation}
    \label{ec.3.20}
    \F^\circledast \omega:=\F\circ \omega.
\end{equation}
 The reader is encouraged to consult the reference \cite{sald1} for more details. Notice that $\F^\circledast \omega$ is only the {\it dualization} of the action of the gauge group on principal connections via the pull--back in the {\it classical} case.

\subsection{The Model}

For the model of the electroweak interaction that we want to develop using the advantages of  non--commutative geometry, we will consider the trivial quantum principal $\SU(2)$--bundle given by
\begin{equation}
    \label{ec.3.50}
    \zeta_{4D}=(P,B,\Delta_P),
\end{equation}
where $$B:=C^\infty_\C(\R^4),\qquad P:=B\otimes \SU(2), \qquad \Delta_P:=\id_B\otimes \Delta.$$
Furthermore, the pair 
\begin{equation}
    \label{ec.3.51}
    (\Omega^\bullet(P):=\Omega^\bullet(B)\otimes \Gamma^\wedge,\;\Delta_{\Omega^\bullet(P)}:=\id_{\Omega^\bullet(B)}\otimes \Delta)
\end{equation}
constitutes a differential calculus for $\zeta_{4D}$, where $$(\Omega^\bullet(B):=\Omega^\bullet_\C(\R^4),d,\ast) $$ is the graded differential $\ast$--algebra of $\C$--valued differential forms of $\R^4$, and $$(\Gamma^\wedge,d,\ast)$$ is the universal differential $\ast$--calculus of the bicovariant $\ast$--FODC of $\SU(2)$ of equation (\ref{ec.2.85}), and the symbol $\otimes$ in  $\Omega^\bullet(B)\otimes \Gamma^\wedge$ denotes the tensor product of graded differential $\ast$--algebras. Moreover, $\Delta: \Gamma^\wedge\longrightarrow \Gamma^\wedge\otimes \Gamma^\wedge$ is the extension of the coproduct of $\SU(2)$ given in equation (\ref{ec.2.9}).

By equation (\ref{ec.3.2}), we have
\begin{equation}
    \label{ec.3.52}
    \Hor^\bullet\,P:=\Omega^\bullet(B)\otimes \SU(2)
\end{equation}
and by equation (\ref{ec.3.3}) we get that
\begin{equation}
    \label{ec.3.53}
    \Omega^\bullet(B)\otimes \mathbbm{1}\cong \Omega^\bullet(B)
\end{equation}
is space of base forms. 

In light of Lemma $6.11$ of reference \cite{micho2}, the linear map
\begin{equation}
\label{ec.3.54}
\omega^\triv:\mathfrak{su}(2,\C)^\#_{4D} \longrightarrow \Omega^1(P),\qquad
\theta \longmapsto \mathbbm{1}\otimes \theta.
\end{equation}
is a qpc and it will be called the trivial qpc. Furthermore,  by Lemma $6.11$ of reference \cite{micho2}, there is a bijection between
\begin{equation}
    \label{ec.3.55}
    \begin{aligned}
        \mathrm{Hom}_\C(\mathfrak{su}(2,\C)^\#_{4D},\Omega^1(B))^\dagger=\{ A^\omega:\;&\mathfrak{su}(2,\C)^\#_{4D}\longrightarrow \Omega^1(B)\mid \\ & A^\omega \mbox{ is linear such that }\, \ast\circ A^\omega=A^\omega\circ \ast\}
    \end{aligned}
\end{equation}
 and the set $\mathfrak{qpc}(\zeta_{4D})$. This bijection is given by 
\begin{equation}
\label{ec.3.56}
\omega=(A^\omega\otimes \id_{\SU(2)})\circ \ad+\omega^\triv,
\end{equation}
The hermitian linear map 
\begin{equation}
    \label{ec.3.57}
    A^\omega:\mathfrak{su}(2,\C)^\#_{4D}\longrightarrow \Omega^1(B)
\end{equation}
can be interpreted as the {\it quantum gauge potential} of $\omega$.

\begin{Proposition}
    \label{regular}
    A quantum principal connection $\omega$ is regular if and only if $A^\omega(\eta_4)=0$.
\end{Proposition}
\begin{proof}
    In light of Lemma 6.11 of reference \cite{micho2}, $\omega$ is regular if and only if $A^\omega$ satisfies
    \begin{enumerate}
        \item $A^\omega(\theta)\,\mu=(-1)^k\,\mu\,A^\omega(\theta)$ for all $\mu$ $\in$ $\Omega^k(B)$ and all $\theta$ $\in$ $\mathfrak{su}(2,\C)^\#_{4D}$.
        \item $A^\omega(\theta \diamondsuit g)=\epsilon(g)A^\omega(\theta)$ for all $\theta$ $\in$ $\mathfrak{su}(2,\C)^\#_{4D}$ and all $g$ $\in$ $\SU(2)$ (see equation (\ref{ec.2.91})).
    \end{enumerate}
    The first point follows for every quantum gauge potential  because $\Omega^\bullet(B)$ is graded--commutative. 
    
    On the other hand, it is easy to see that every quantum gauge potential $A^\omega$ satisfies 
    \begin{equation}
        \label{regec}
        A^\omega(\eta_j \diamondsuit g)=\epsilon(g)A^\omega(\eta_j)
    \end{equation}
    for all $g$ $\in$ $\SU(2)$ and $j=1,2,3$, where $\eta_j$ is defined in equation (\ref{ec.2.89}). 

    Now assume that $\omega$ is regular. Then $$A^\omega(\eta_1\diamondsuit \gamma^\ast)=-{i\,\sqrt{2}\over q_w}A^{\omega}(\pi(\gamma)\diamondsuit \gamma^\ast)=-{i\,\sqrt{2}\over q_w}A^{\omega}(\pi(\gamma  \gamma^\ast))= {i\,\sqrt{2}\over 2\,q_w}A^\omega(\pi(\alpha+\alpha^\ast)),$$ and $$\epsilon(\gamma^\ast)\,A^\omega(\eta_1)=0.$$ This implies that $A^\omega(\pi(\alpha+\alpha^\ast))=0$ and hence $A^\omega(\eta_4)=0$. Reciprocally, if $A^\omega$ satisfies $A^\omega(\eta_4)=0$, by equation (\ref{regec}) we get that $\omega$ is regular.
\end{proof}

\begin{Remark}
    \label{reg}
    In the classical case, every principal connection is regular; so in general, a regular qpc can be interpreted as a classical qpc in the non--commutative geometrical setting.  Hence, the previous proposition shows that qpc’s which do not vanish on the fourth dimension of $\mathfrak{su}(2,\mathbb{C})^\#_{4D}$ (corresponding to $\eta_4$) are genuinely non--classical objects, with no analog in differential geometry in any sense. Moreover, for these qpc’s we have $D^\omega \neq \widehat{D}^\omega$.
\end{Remark}

Let us consider the embedded differential $\Theta$ of equation (\ref{ec.2.98}). Then, we can define the curvature of a qpc by means of equation (\ref{ec.3.10}). In light of Lemma 6.12 of reference \cite{micho2}, the bijection of equation (\ref{ec.3.55}) extends naturally to the curvature by  
\begin{equation}
    \label{ec.3.58}
    R^{\omega}=(F^\omega\otimes \id_{\SU(2)})\circ \ad,
\end{equation}
where 
\begin{equation}
    \label{ec.3.59}
    F^\omega:\mathfrak{su}(2,\C)^\#_{4D}\longrightarrow \Omega^2(B)
\end{equation}
is the linear map defined as (see equation (\ref{brakets}))
\begin{equation}
    \label{ec.3.60}
    F^\omega:=dA^\omega-\langle A^\omega,A^\omega\rangle.
\end{equation}
 The map $F^\omega$ can be interpreted as the {\it quantum field strength} of $\omega$. Notice that $F^{\omega^\triv}=0$ (because $A^{\omega^\triv}=0$) and it follows that $\omega^\triv$ is flat. 
 
 Let $\omega$ be a qpc and consider the linear basis $\widetilde{\beta}^\#_{4D}$ of $\mathfrak{su}(2,\C)^\#_{4D}$ given in equation (\ref{ec.2.103}). Since $\omega$ is a linear map with domain $\mathfrak{su}(2,\C)^\#_{4D}$ and image $\Omega^1(B)$, we identify four  $\C$--valued differential $1$--forms of $\R^4$ associated with $\omega$: $$A^\omega(\beta_1),\; A^\omega(\beta_2),\; A^\omega(\beta_3)\;\;\;\mbox{ and }\;\;\;A^\omega(\beta_4).$$

\begin{Proposition}
    \label{prop3.4}
    For the basis $\widetilde{\beta}^\#_{4D}$, the quantum field strength is given by 
    $$F^\omega(\beta_1)=dA^\omega(\beta_1)+{q_w}\, A^\omega(\beta_2)\wedge A^\omega(\beta_3),$$ $$F^\omega(\beta_2)=dA^\omega(\beta_2)+{q_w}\, A^\omega(\beta_3)\wedge A^\omega(\beta_1), $$ $$F^\omega(\beta_3)=dA^\omega(\beta_3)+{q_w}\, A^\omega(\beta_1)\wedge A^\omega(\beta_2),$$ $$F^\omega(\beta_4)=dA^\omega(\beta_4).$$ Here, $\wedge$ denotes the usual wedge product of $\Omega^\bullet(B)$.
\end{Proposition}

\begin{proof}
    The statement easily follows from Proposition \ref{prop2.3.2} and the fact that $\Omega^\bullet(B)$ is graded--commutative. For example, by Proposition \ref{prop2.3.2} we get
    \begin{eqnarray*}
        F^\omega(\beta_1)&=&dA^\omega(\beta_1)-\langle A^\omega,A^\omega\rangle(\beta_1)
        \\
        &=&
        dA^\omega(\beta_1)-{i\,q_4\over 2}(A^\omega(\beta_1)\wedge A^\omega(\beta_4)+A^\omega(\beta_4)\wedge A^\omega(\beta_1))
        \\
        &+&
        {q_w\over 2}(A^\omega(\beta_2)\wedge A^\omega(\beta_3)-A^\omega(\beta_3)\wedge A^\omega(\beta_2))
    \end{eqnarray*}
and since $\Omega^\bullet(B)$ is graded--commutative, we obtain $$A^\omega(\beta_1)\wedge A^\omega(\beta_4)+A^\omega(\beta_4)\wedge A^\omega(\beta_1)=0,$$ $$A^\omega(\beta_2)\wedge A^\omega(\beta_3)-A^\omega(\beta_3)\wedge A^\omega(\beta_2)=2\,A^\omega(\beta_2)\wedge A^\omega(\beta_3).$$ Hence $$ F^\omega(\beta_1)=dA^\omega(\beta_1)+{q_w}\, A^\omega(\beta_2)\wedge A^\omega(\beta_3).$$    
\end{proof}

Also, notice that
\begin{equation}
    \label{ec.3.42.1.1}
    F^\omega(\beta_j)^\ast=F^\omega(\beta^\ast_j)=F^\omega(\beta_j)
\end{equation}
for $j=1,2,3,4$.

In order to prove that the map $F^\omega$ correctly describes the curvature--tensor of the electroweak theory (\cite{notation}), we have
\begin{Proposition}
    \label{prop3.4.1}
    All the equations of Proposition \ref{prop3.4} can be summarized as 
    \begin{equation}
        \label{ec.3.46.8}
        F^j_{\mu\nu}=\partial_\mu A^j_\nu-\partial_\nu A^j_\mu+q_w\,\sum^4_{k,l=1}\,f^{jkl}\,A^k_\mu\,A^l_\nu,
    \end{equation}
    for $\mu,\nu=0,1,2,3$ and $j=1,2,3,4$ in index notation, where 
\begin{equation}
    \label{ec.3.46.10}
    f^{jkl}=\epsilon^{jkl} \quad \mbox{ if } \quad j,k,l \,\in\,\{1,2,3\} \quad \mbox{ and }\quad f^{jkl}=0 \;\; \mbox{ if some index is 4}, 
\end{equation}
with $\epsilon^{jkl}$ the $3$--dimensional Levi--Civita symbol.
\end{Proposition}
\begin{proof}
    Let $$A^\omega: \mathfrak{su}(2,\C)^\#_{4D}\longrightarrow \Omega^1(B)$$ be a quantum gauge potential. In terms of the linear basis $\widetilde{\beta}^\#_{4D}$ of $\mathfrak{su}(2,\C)^\#_{4D}$, we get 
    \begin{equation}
        \label{ec.3.46.1}
A^\omega(\beta_j)=A^j_0\,dx^0+A^j_1\,dx^1+A^j_2\,dx^2+A^j_3\,dx^3
    \end{equation}
for $j=1,2,3,4$. In the usual tensor index notation, we have
\begin{equation}
    \label{ec.3.46.2}
    A^\omega(\beta_j)\;\;\longleftrightarrow \;\; A^j_\mu=(A^j_0,A^j_1,A^j_2,A^j_3).
\end{equation}
A straightforward calculation substituting equation (\ref{ec.3.46.1}) in the result of the Proposition \ref{prop3.4} shows that
\begin{equation}
    \label{ec.3.46.3}
    F^\omega(\beta_j)=\sum^3_{\mu,\nu=0} {1\over 2}\,F^j_{\mu\nu}\,dx^\mu\wedge dx^\nu,
\end{equation}
where 
\begin{equation*}
    F^1_{\mu\nu}=\partial_\mu A^1_\nu-\partial_\nu A^1_\mu+q_w\,(A^2_\mu\,A^3_\nu-A^3_\nu\,A^2_\mu),
\end{equation*}
\begin{equation*}
    F^2_{\mu\nu}=\partial_\mu A^2_\nu-\partial_\nu A^2_\mu+q_w\,(A^3_\mu\,A^1_\nu-A^1_\nu\,A^3_\mu),
\end{equation*}
\begin{equation*}
    F^3_{\mu\nu}=\partial_\mu A^3_\nu-\partial_\nu A^3_\mu+q_w\,(A^1_\mu\,A^2_\nu-A^2_\nu\,A^1_\mu),
\end{equation*}
\begin{equation*}
    F^4_{\mu\nu}=\partial_\mu A^4_\nu-\partial_\nu A^4_\mu.
\end{equation*}
The previous equations can be summarized by
\begin{equation*}
    F^j_{\mu\nu}=\partial_\mu A^j_\nu-\partial_\nu A^j_\mu+q_w\,\sum^4_{k,l=1}\,f^{jkl}\,A^k_\mu\,A^l_\nu
\end{equation*}
for $j=1,2,3,4$ in index notation, with $f^{jkl}$ as in equation (\ref{ec.3.46.10}).

In the usual tensor index notation, we get
\begin{equation}
    \label{ec.3.46.11}
    F^\omega(\beta_j)\;\;\longleftrightarrow \;\; F^j_{\mu\nu}
\end{equation}
\end{proof}

It is worth mentioning that equation (\ref{ec.3.46.8}) has the usual form of the gauge curvature tensor in electroweak theory found in the literature (\cite{diff1,notation}). Of course, this expression depends on the choice of basis for $\mathfrak{u}(2)$. Likewise, the expression for the curvature in equation~(\ref{ec.3.46.8}) depends on the choice of basis for $\mathfrak{su}(2,\C)^\#_{4D}$. Nevertheless, $F^j_{\mu\nu}$ will always coincide with the corresponding gauge curvature tensor in electroweak theory, for any choice of bases of $\mathfrak{su}(2,\C)^\#_{4D}$ and $\mathfrak{u}(2)$, under the linear isomorphism $$\beta_j\;\longleftrightarrow\;\hat{\varrho}^j:={\sigma_j\over 2} \;\;\;\mbox{ for j=1,2,3}\qquad \mbox{ and }\qquad \beta_4\;\longleftrightarrow \;\hat{\varrho}^4:= {\Id_2\over 2},$$ where $\sigma_1,\sigma_2,\sigma_3$ are the Pauli matrices and $\Id_2$ is the identity matrix of order $2$. In fact, let $$\{ \alpha_j\}^4_{j=1}$$ be another linear basis of $\mathfrak{su}(2,\C)^\#_{4D}$. Then, there exists an invertible matrix $C$ such that $C\beta_j=\alpha_j$. Thus, $$\{\alpha'^j=C\,\hat{\varrho}^j \} $$ is a linear basis of $\mathfrak{u}(2)$ and the gauge curvature tensor in electroweak theory calculated in terms of $\alpha'^j$ coincides with the components of $F^\omega(\alpha_j)$.

According to Lemma $6.10$ of reference \cite{micho2}, there is a bijection between 
\begin{equation}
    \label{ec.3.61}
    \mathrm{Hom}_\C(\mathfrak{su}(2,\C)^\#_{4D},\Omega^\bullet(B))=\{ L^\tau:\mathfrak{su}(2,\C)^\#_{4D}\longrightarrow \Omega^\bullet(B)\mid  L^\tau \mbox{ is linear} \}
\end{equation}
and $\Mor(\ad,\Delta_\Hor)$ (see equation (\ref{ec.3.13})).
This bijection is given by 
\begin{equation}
    \label{ec.3.62}
    \tau=(L^\tau\otimes\id_{\SU(2)})\circ \ad
\end{equation}
for every $\tau$ $\in$ $\Mor(\ad,\Delta_\Hor)$. Furthermore, Lemma $6.12$ of reference \cite{micho2} shows that 
\begin{equation}
    \label{ec.3.63}
    D^\omega(\tau) =(D^\omega_B(L^\tau)\otimes \id_{\SU(2)}
    )\circ \ad,
\end{equation}
where
\begin{equation}
    \label{used3}
    D^\omega_B:\mathrm{Hom}_\C(\mathfrak{su}(2,\C)^\#_{4D},\Omega^\bullet(B))\longrightarrow \mathrm{Hom}_\C(\mathfrak{su}(2,\C)^\#_{4D},\Omega^\bullet(B))
\end{equation}
is given by 
\begin{equation*}
    L^\tau\;\;\longmapsto \;\; dL^\tau-(-1)^{k}[L^\tau,A^\omega]
\end{equation*}
for every $\tau$ $\in$ $\Mor(\ad,\Delta_\Hor)$ such that $\Im(\tau)\subseteq \Hor^k P$; and Proposition 5.4 of reference \cite{sald2} proves that
\begin{equation}
    \label{ec.3.64}
    S^\omega(\tau)=(S^\omega_B(L^\tau)\otimes \id_{\SU(2)})\circ \ad,
\end{equation}
where 
\begin{equation}
    \label{used4}
    S^\omega_B:\mathrm{Hom}_\C(\mathfrak{su}(2,\C)^\#_{4D},\Omega^\bullet(B))\longrightarrow \mathrm{Hom}_\C(\mathfrak{su}(2,\C)^\#_{4D},\Omega^\bullet(B))
\end{equation}
is given by
$$L^\tau\;\;\longmapsto \;\;\langle A^\omega,L^\tau\rangle-(-1)^k\langle L^\tau,A^\omega\rangle-(-1)^k[L^\tau,A^\omega].$$
Therefore, the twisted covariant derivative (see Definition \ref{twisted}) satisfies  
\begin{equation}
\label{twsted.12}
    DS^\omega(\tau)=(DS^\omega_B(L^\tau)\otimes \id_{\mathcal{SU}(2)})\circ \ad,
\end{equation}
where 
\begin{equation}
    \label{twisted.13}
    DS^\omega_B:=D^\omega_B-S^\omega_B:\mathrm{Hom}_\C(\mathfrak{su}(2,\C)^\#_{4D},\Omega^\bullet(B))\longrightarrow \mathrm{Hom}_\C(\mathfrak{su}(2,\C)^\#_{4D},\Omega^\bullet(B))
\end{equation}
is given by
$$L^\tau\;\;\longmapsto\;\; dL^\tau-\langle A^\omega,L^\tau\rangle+(-1)^k\langle L^\tau,A^\omega\rangle.$$

\begin{Proposition}
\label{bianchi}
    The non--commutative geometrical (second) Bianchi identity \begin{equation}
    \label{ec.3.65}
DS^\omega(R^\omega)=\langle\omega,\langle\omega,\omega\rangle\rangle-\langle\langle\omega,\omega\rangle,\omega\rangle
\end{equation} can be written as
\begin{equation}
    \label{ec.3.66}
    \begin{aligned}
        (DS^\omega_B(F^\omega)\otimes \id_{\SU(2)})\circ \ad=((\langle A^\omega,\langle A^\omega,A^\omega\rangle\rangle-\langle \langle A^\omega,A^\omega\rangle,A^\omega\rangle)\otimes \id_{\SU(2)})\circ \ad 
    \end{aligned}
\end{equation}
\end{Proposition}
\begin{proof}
   Equation (\ref{twsted.12}) guaranties that the left--hand side of equation (\ref{ec.3.65}) is $$(DS^\omega_B(F^\omega)\otimes \id_{\SU(2)})\circ \ad.$$ On the other hand, since $\langle\omega,\langle\omega,\omega\rangle\rangle-\langle\langle\omega,\omega\rangle,\omega\rangle$ is equal to $DS^\omega(R^\omega)$, and  $DS^\omega(R^\omega)$ $\in$ $\Mor(\ad,\Delta_\Hor)$, it follows that $$\langle\omega,\langle\omega,\omega\rangle\rangle-\langle\langle\omega,\omega\rangle,\omega\rangle\;\in\;\Mor(\ad,\Delta_\Hor).$$ 
   
   Let $\theta$ $\in$ $\mathfrak{su}(2,\C)^\#_{4D}$ and $\Theta(\theta)=\displaystyle\sum_{i,j}\theta_i\otimes \theta'_j$, for some $\theta_i$, $\theta_j$ $\in$ $\mathfrak{su}(2,\C)^\#_{4D}$. Then 
   \begin{eqnarray*}
\langle\omega,\omega\rangle(\theta)&=&\sum_{i,j}\omega(\theta_i)\, \omega(\theta'_j)
\\
&=&
\sum_{i,j}(A^\omega(\theta^{(0)}_i)\otimes \theta^{(1)}_i+\omega^\triv(\theta_i))\,(A^\omega(\theta'^{(0)}_j)\otimes \theta'^{(1)}_j+\omega^\triv(\theta'_j))
\\
&=&
\sum_{i,j}A^\omega(\theta^{(0)}_i)\,A^\omega(\theta'^{(0)}_j)\otimes \theta^{(1)}_i\,\theta'^{(1)}_j+\omega^\triv(\theta_i)\,((A^\omega\otimes \id_{\SU(2)})\ad(\theta'_j))
\\
&+&
((A^\omega\otimes \id_{\SU(2)})\ad(\theta_i))\,\omega^\triv(\theta'_j)+\omega^\triv(\theta_i)\,\omega^\triv(\theta'_j).
   \end{eqnarray*}

Notice that $$\sum_{i,j} \omega^\triv(\theta_i)\,((A^\omega\otimes \id_{\SU(2)})\ad(\theta'_j))=\langle \omega^\triv,(A^\omega\otimes \id_{\SU(2)})\circ\ad\rangle(\theta),$$ $$\sum_{i,j} ((A^\omega\otimes \id_{\SU(2)})\ad(\theta_i))\,\omega^\triv(\theta'_j)=\langle(A^\omega\otimes \id_{\SU(2)})\circ \ad,\omega^\triv\rangle(\theta),$$
$$\displaystyle\sum_{i,j}\omega^\triv(\theta_i)\,\omega^\triv(\theta'_j)=\langle \omega^\triv,\omega^\triv\rangle(\theta).$$  Furthermore, since $\Theta$ $\in$ $\Mor(\ad,\ad^{\otimes^2})$, it follows that
\begin{eqnarray*}
    \sum^m_{i,j=1} A^\omega(\theta^{(0)}_i)\,A^\omega(\theta'^{(0)}_j)\otimes \theta^{(1)}_i\theta'^{(1)}_j&=& ((m_{\Omega^\bullet}\circ (A^\omega\otimes A^\omega))\otimes \id_H) \ad^{\otimes 2}(\Theta(\theta))
    \\
    &=&
    ((m_{\Omega^\bullet}\circ (A^\omega\otimes A^\omega)\circ \Theta)\otimes \id_H)\ad(\theta)
    \\
    &=&
    (\langle A^\omega,A^\omega\rangle \otimes \id_H)\ad(\theta),
\end{eqnarray*}
 where $m_{\Omega^\bullet}:\Omega^\bullet(B)\otimes \Omega^\bullet(B)\longrightarrow \Omega^\bullet(B)$ is the product map. Thus
 \begin{eqnarray*}
     \langle \omega,\omega\rangle&=& (\langle A^\omega,A^\omega\rangle \otimes \id_H)\circ \ad+\langle \omega^\triv,(A^\omega\otimes \id_{\SU(2)})\circ\ad\rangle
     \\
     &+&
     \langle(A^\omega\otimes \id_{\SU(2)})\circ \ad,\omega^\triv\rangle+\langle \omega^\triv,\omega^\triv\rangle.
 \end{eqnarray*}
 A similar calculation proves that 
 \begin{eqnarray}
 \label{cuentas}
     \langle\omega,\langle\omega,\omega\rangle\rangle-\langle\langle\omega,\omega\rangle,\omega\rangle&=&((\langle A^\omega,\langle A^\omega,A^\omega\rangle\rangle-\langle \langle A^\omega,A^\omega\rangle,A^\omega\rangle)\otimes \id_{\SU(2)})\circ \ad \nonumber
     \\
     &+&
     \langle (A^\omega\otimes \id_{\SU(2)})\circ\ad,\langle \omega^\triv,(A^\omega\otimes \id_{\SU(2)})\circ\ad\rangle\rangle  \nonumber
     \\
     &\vdots& \nonumber
     \\
     &+&
     \langle \omega^\triv,\langle\omega^\triv,\omega^\triv\rangle\rangle 
     \\
     &-&
     \langle \langle\omega^\triv,(A^\omega\otimes \id_{\SU(2)})\circ\ad\rangle,(A^\omega\otimes \id_{\SU(2)})\circ\ad\rangle  \nonumber
     \\
     &\vdots& \nonumber
     \\
     &-&
     \langle \langle\omega^\triv,\omega^\triv\rangle,\omega^\triv\rangle.  \nonumber
 \end{eqnarray}
Now, notice that 
\begin{eqnarray*}
    \Omega^3(P)&=&(\Omega^3(B)\otimes \SU(2))\oplus(\Omega^2(B)\otimes \Gamma)\oplus (\Omega^1(B)\otimes \Gamma^{\wedge\,2})\oplus(B\otimes \Gamma^{\wedge\,3}) 
    \\
    &=&
    \Hor^3\,P\oplus(\Omega^2(B)\otimes \Gamma)\oplus (\Omega^1(B)\otimes \Gamma^{\wedge\,2})\oplus(B\otimes \Gamma^{\wedge\,3}).
\end{eqnarray*}
So, for every $\varphi$, $\psi$ $\in$ $\Hor^\bullet\,P=\Omega^\bullet(B)\otimes \SU(2)$ and every $\theta_1$, $\theta_2$, $\theta_3$ $\in$ $\mathfrak{su}(2,\C)^\#_{4D}\subset \Gamma$, if the following products are not zero, then we have 
$$\varphi\,\omega^\triv(\theta_1)\,\psi=\varphi\,(\mathbbm{1}\otimes \theta_1)\,\psi\,\not\in\, \Hor^\bullet\,P,$$ $$\varphi\,\psi\,\omega^\triv(\theta_1)=\varphi\,\psi\,(\mathbbm{1}\otimes \theta_1)\,\not\in\, \Hor^\bullet\,P, $$ $$\omega^\triv(\theta_1)\,\varphi\,\psi=(\mathbbm{1}\otimes \theta_1)\,\varphi\,\psi\,\not\in\,\Hor^\bullet\,P, $$ $$\varphi\,\omega^\triv(\theta_1)\,\omega^\triv(\theta_2)=\varphi\,(\mathbbm{1}\otimes \theta_1)\,(\mathbbm{1}\otimes \theta_2)\,\not\in\,\Hor^\bullet\,P,$$ $$ \omega^\triv(\theta_1)\,\omega^\triv(\theta_2)\,\varphi=(\mathbbm{1}\otimes \theta_1)\,(\mathbbm{1}\otimes \theta_2)\,\varphi\,\not\in\,\Hor^\bullet\,P,$$ $$\omega^\triv(\theta_1)\,\varphi\,\omega^\triv(\theta_2)=(\mathbbm{1}\otimes \theta_1)\,\varphi\,(\mathbbm{1}\otimes \theta_2)\,\not\in\,\Hor^\bullet\,P, $$ $$\omega^\triv(\theta_1)\,\omega^\triv(\theta_2)\,\omega^\triv(\theta_3)=(\mathbbm{1}\otimes \theta_1)\,(\mathbbm{1}\otimes \theta_2)\,(\mathbbm{1}\otimes \theta_3)=\mathbbm{1}\otimes \theta_1\,\theta_2\,\theta_3 \,\not\in\,\Hor^\bullet\,P.$$ Since $\Im(\langle\omega,\langle\omega,\omega\rangle\rangle-\langle\langle\omega,\omega\rangle,\omega\rangle) \subseteq \Hor^3\,P,$ it follows that the right--hand side of equation (\ref{cuentas}) turns into 
 $$((\langle A^\omega,\langle A^\omega,A^\omega\rangle\rangle-\langle \langle A^\omega,A^\omega\rangle,A^\omega\rangle)\otimes \id_{\SU(2)})\circ \ad.$$
\end{proof}

In this way, by Proposition \ref{bianchi}, in terms of the maps $A^\omega$, $F^\omega$,  the non--commutative Bianchi identity is given by
\begin{equation}
    \label{ec.3.67}
     DS^\omega_B(F^\omega)=\langle A^\omega,\langle A^\omega,A^\omega\rangle\rangle-\langle \langle A^\omega,A^\omega\rangle,A^\omega\rangle.
\end{equation}

At first glance, equation~(\ref{ec.3.67}) is not the \emph{classical} Bianchi identity of the electroweak interaction. However, we have
\begin{Proposition}
    \label{prop3.5}
     For the basis $\widetilde{\beta}^\#_{4D}$, the non--commutative geometrical Bianchi identity is given by $$dF^\omega(\beta_1)-q_w\,(F^\omega(\beta_2)\wedge A^\omega(\beta_3)-F^\omega(\beta_3)\wedge A^\omega(\beta_2))=0,$$ $$dF^\omega(\beta_2)-q_w\,(F^\omega(\beta_3)\wedge A^\omega(\beta_1)-F^\omega(\beta_1)\wedge A^\omega(\beta_3))=0,$$ $$dF^\omega(\beta_3)-q_w\,(F^\omega(\beta_1)\wedge A^\omega(\beta_2)-F^\omega(\beta_2)\wedge A^\omega(\beta_1))=0,$$ $$dF^\omega(\beta_4)=0.$$ Here, $\wedge$ denotes the usual wedge product of $\Omega^\bullet(B)$.
\end{Proposition}

\begin{proof}
    The statement readily follows from Proposition~\ref{prop2.3.2} together with the fact that $\Omega^\bullet(B)$ is graded--commutative. For example, 
    \begin{eqnarray*}
        \langle A^\omega,\langle A^\omega,A^\omega\rangle\rangle(\beta_1)&=&{i\,q_4\over 2}\,(A^\omega(\beta_1)\, \langle A^\omega,A^\omega\rangle(\beta_4)+A^\omega(\beta_4)\, \langle A^\omega,A^\omega\rangle(\beta_1))
        \\
        &-&
        {q_w\over 2}\,(A^\omega(\beta_2)\, \langle A^\omega,A^\omega\rangle(\beta_3)- A^\omega(\beta_3)\,\langle A^\omega,A^\omega\rangle(\beta_2))
        \\
        &=&
        {i\,q_4\over 2}\,A^\omega(\beta_4)\, \langle A^\omega,A^\omega\rangle(\beta_1)
        \\
        &=&
        -{i\,q_4\,q_w\over 4}\,(A^\omega(\beta_4)\,A^\omega(\beta_2)\, A^\omega(\beta_3)- A^\omega(\beta_4)\,A^\omega(\beta_3)\,A^\omega(\beta_2))
        \\
        &=&
        {i\,q_4\,q_w\over 2}\,A^\omega(\beta_4)\,A^\omega(\beta_3)\, A^\omega(\beta_2);
    \end{eqnarray*}
and
\begin{eqnarray*}
   \langle \langle A^\omega,A^\omega\rangle,A^\omega\rangle(\beta_1)&=&{i\,q_4\over 2}\,(\langle A^\omega,A^\omega\rangle(\beta_1)\, A^\omega(\beta_4)+\langle A^\omega,A^\omega\rangle(\beta_4)\, A^\omega(\beta_1))
   \\
   &-&
   {q_w\over 2}\,(\langle A^\omega,A^\omega\rangle(\beta_2)\, A^\omega(\beta_3)- \langle A^\omega,A^\omega\rangle(\beta_3)\,A^\omega(\beta_2))
   \\
   &=&
     {i\,q_4\over 2}\,\langle A^\omega,A^\omega\rangle(\beta_1)\, A^\omega(\beta_4) 
    \\
   &=&
     -{i\,q_4\,q_w\over 4}\,(A^\omega(\beta_2)\, A^\omega(\beta_3)\,A^\omega(\beta_4)- A^\omega(\beta_3)\,A^\omega(\beta_2)\,A^\omega(\beta_4))
     \\
   &=&
   {i\,q_4\,q_w\over 2}\,A^\omega(\beta_4)\,A^\omega(\beta_3)\,A^\omega(\beta_2);
\end{eqnarray*}
so $$ \langle \langle A^\omega,A^\omega\rangle,A^\omega\rangle(\beta_1)-\langle \langle A^\omega,A^\omega\rangle,A^\omega\rangle(\beta_1)=0.$$ On the other hand, by equation (\ref{twisted.13}) we get
\begin{eqnarray*}
    DS^\omega_B(F^\omega)(\beta_1)&=&dF^\omega(\beta_1)-\langle A^\omega,F^\omega\rangle(\beta_1)+\langle F^\omega,A^\omega\rangle(\beta_1)
    \\
    &=&
    dF^\omega(\beta_1)
    \\
    &-&
    {i\,q_4\over 2}\,(A^\omega(\beta_1)\wedge F^\omega(\beta_4)+A^\omega(\beta_4)\wedge F^\omega(\beta_1))
    \\
    &+&
    {q_w\over 2}\,(A^\omega(\beta_2)\wedge F^\omega(\beta_3)- A^\omega(\beta_3)\wedge F^\omega(\beta_2))
    \\
    &+&
    {i\,q_4\over 2}\,(F^\omega(\beta_1)\wedge A^\omega(\beta_4)+F^\omega(\beta_4)\wedge A^\omega(\beta_1))
    \\
    &-&
   {q_w\over 2}\,(F^\omega(\beta_2)\wedge A^\omega(\beta_3)- F^\omega(\beta_3)\wedge A^\omega(\beta_2))
   \\
    &=&
    dF^\omega(\beta_1)-q_w\,(F^\omega(\beta_2)\wedge A^\omega(\beta_3)-F^\omega(\beta_3)\wedge A^\omega(\beta_2))
\end{eqnarray*}
Hence $$ dF^\omega(\beta_1)-q_w\,(F^\omega(\beta_2)\wedge A^\omega(\beta_3)-F^\omega(\beta_3)\wedge A^\omega(\beta_2))=0.$$
\end{proof}

The proof of the following theorem consists merely of substituting the result of Proposition~\ref{prop3.4.1} into that of Proposition~\ref{prop3.5}. Therefore, we omit it.
\begin{Theorem}
    \label{prop.3.5.1}
    All the equations of Proposition \ref{prop3.5} can be summarized as 
    \begin{equation}
    \label{ec.3.46.14.4}
\partial_\mu\,\widetilde{F}^{j\,\mu\nu}+q_w\,\sum^4_{k,l=1}\,f^{jkl}\,A^k_\mu\,\widetilde{F}^{l\,\mu\nu}=0,
\end{equation}
where
\begin{equation}
    \label{ec.3.46.13}
    \widetilde{F}^{j\,\mu\nu}= {1\over 2}\epsilon^{\mu\nu\alpha\beta}F^j_{\alpha\beta}
\end{equation}
with $\epsilon^{\mu\nu\alpha\beta}$ the $4$--dimensional Levi--Civita symbol for $\mu,\nu=0,1,2,3$ and $j=1,2,3,4$.
\end{Theorem}

There are four remarks to be made concerning the last result. 

\begin{Remark}
    \label{rema3.1}
    It is worth emphasizing that, when expressed in terms of the maps $A^\omega$ and $F^\omega$, the non--commutative geometrical Bianchi identity coincides with the Bianchi identity of the electroweak interaction \cite{diff1}. Proposition \ref{prop.3.5.1} shows that, for the quantum model of the electroweak theory that we aim to construct, the information provided by the non--commutative geometrical Bianchi identity is indeed, the correct one.
\end{Remark}

\begin{Remark}
    \label{rema3.0.2}
    Assume that the correct Bianchi identity in the non--commutative geometrical setting is 
    \begin{equation}
        \label{falsebianchi}
        D^\omega R^\omega = 0
    \end{equation}
    as in the {\it classical} case. By equation (\ref{used3}), when expressed in terms of the maps $A^\omega$ and $F^\omega$, the previous equation takes the form $$D^\omega_B(F^\omega)=0 \qquad \mbox{ with }\qquad D^\omega_B(F^\omega)=dF^\omega - [F^\omega, A^\omega].$$
In particular, by evaluating in $\beta_1$ and using Proposition \ref{prop2.3.3}, we obtain
\begin{eqnarray}
\label{incorrec1}
0 &=& dF^\omega(\beta_1) - [F^\omega, A^\omega](\beta_1)\nonumber 
       \\
      &=&
      dF^\omega(\beta_1)-q_w\,(F^\omega(\beta_2)\wedge A^\omega(\beta_3)-F^\omega(\beta_3)\wedge A^\omega(\beta_2))
      \\
      &-&
      i\,q_4\,F^\omega(\beta_1)\wedge A^\omega(\beta_4)+{q_4\over 2}\,F^\omega(\beta_2)\wedge A^\omega(\beta_4)
      \\
      &=&
      -i\,q_4\,F^\omega(\beta_1)\wedge A^\omega(\beta_4)+{q_4\over 2}\,F^\omega(\beta_2)\wedge A^\omega(\beta_4),
\end{eqnarray}
since $$dF^\omega(\beta_1)-q_w\,(F^\omega(\beta_2)\wedge A^\omega(\beta_3)-F^\omega(\beta_3)\wedge A^\omega(\beta_2))=0.$$
In conclusion, equation (\ref{falsebianchi}) for $\beta_1$ is equivalent to  
\begin{equation}
    \label{incorrect}
    0=i\,q_4\,F^\omega(\beta_1)\wedge A^\omega(\beta_4)-{q_4\over 2}\,F^\omega(\beta_2)\wedge A^\omega(\beta_4), 
\end{equation}
which is clearly incorrect for all maps $A^\omega$. Therefore, as we commented in Section 3.1, in non--commutative geometry, in general, the equation $$D^\omega R^\omega=0$$ does not hold. This is important since in the author's opinion, most of the researchers of the area still believe that equation (\ref{falsebianchi}) always holds in non--commutative geometry, but, this is incorrect.
\end{Remark}

\begin{Remark}
\label{rema3.2}
As we commented in Section 3.1, in Proposition 4.9 of reference \cite{micho2} it is proven that the equation that is satisfied for every qpc $\omega$ (and hence for every $A^\omega$) is $$DS^\omega(R^\omega)=\langle\omega,\langle\omega,\omega\rangle\rangle-\langle\langle\omega,\omega\rangle,\omega\rangle.$$  In other words, the geometry of the quantum principal bundle and its differential calculus necessarily entails the presence of the operator $S^\omega$ in order to get a Bianchi identity in the non--commutative geometrical setting. 

The reader should not be surprised that the twisted covariant derivative $DS^\omega$ is the correct expression for the covariant derivative on elements of $\Mor(\ad,\Delta_\Hor)$, since, as discussed in Section 3.1, there is a discrepancy between the differential structure of $(\Gamma^\wedge,d,\ast)$ and its Lie algebra structure, and the operator $DS^\omega$ precisely corrects this discrepancy.

What is more striking is that the action of the twisted covariant derivative $DS^\omega$ on $\Mor(\ad,\Delta_\Hor)$ of $\zeta_{4D}$ coincides with the action of the classical covariant derivative on basic differential forms of type $\ad^{\mathrm{class}}$ of the principal bundle of the electroweak interaction. 
\end{Remark}

\begin{Remark}
    \label{rema3.2.1}
By equation (\ref{incorrect}), in $\zeta_{4D}$, the typical correspondence in index notation doing in textbooks
$$D^\omega \;\;\longleftrightarrow\;\; D_\mu:=\partial_\mu+q_w\,\sum \,f^{jkl}A^k_\mu$$
is incorrect in our model. The correct correspondence is 
\begin{equation}
    \label{ec.3.85}
    DS^\omega \;\;\;\longleftrightarrow\;\;\; DS_\mu:= \partial_\mu + q_w\,\sum \,f^{jkl}A^k_\mu.
\end{equation}
In other words, the covariant derivative of a hermitian $2$--tensor $C^{j\,\mu\nu}$ is no longer given by $$\partial_\mu C^{j\,\mu\nu}+q_w\,\sum^4_{k,l=1}\,f^{jkl}A^k_\mu\, C^{l\,\mu\nu}.$$
Instead, it is the action of the twisted covariant derivative  $DS^\omega$  on $C^{j\,\mu\nu}$ that corresponds to
$$\partial_\mu C^{j\,\mu\nu}+q_w\,\sum^4_{k,l=1}\,f^{jkl}A^k_\mu\, C^{l\,\mu\nu}.$$
In particular, Theorem \ref{prop3.5} can be summarized as follows:
\begin{equation}
    \label{ec.3.86}
    DS_\mu\,(\widetilde{F}^{j\,\mu\nu})=0
\end{equation}
for $\mu,\nu=0,1,2,3$ and $j=1,2,3,4$.
\end{Remark}

Let $\qGG$ be the quantum gauge group of the quantum principal bundle $\zeta_{4D}$. In general, the quantum gauge group is \emph{larger} than its \emph{classical} counterpart (\cite{sald1}). We will take advantage of this property in non--commutative geometry to prove the following result, which will be very useful for our purposes.

\begin{Theorem}
    \label{prop3.6}
   Consider the principal bundle of the electroweak interaction $$\pi_\class:\R^4\times G\longrightarrow \R^4, \qquad (x,C)\longmapsto x$$ with $G=SU(2)\times U(1)$ and consider its gauge group $\mathfrak{GG}$. Then, $\mathfrak{GG}$ can be embedded in $\qGG$.
\end{Theorem}

\begin{proof}
Let 
\begin{equation*}
    F_\class: \R^4\times G\longrightarrow \R^4\times G, \qquad (x,C)\longmapsto (x,\psi(x)\,C)
\end{equation*}
be a gauge transformation of $\pi_\class:\R^4\times G\longrightarrow \R^4$, where $$\psi:\R^4\longrightarrow G $$ is a smooth map. It is well--known that $F_\class$ is equivalent to the $G$--equivariant map 
\begin{equation*}
    f_\class:\R^4\times G\longrightarrow G,\qquad (x,C)\longmapsto C^{-1}\,\psi(x)\,C.
\end{equation*}
We also have the map 
\begin{equation*}
    f^{-1}_\class:\R^4\times G\longrightarrow G,\qquad (x,C)\longmapsto C^{-1}\,\psi^{-1}(x)\,C.
\end{equation*}

Consider 
$$\f_\class:\Gamma^\wedge_G\longrightarrow \Omega^\bullet(P),\qquad \f^{-1}_\class:\Gamma^\wedge_G\longrightarrow \Omega^\bullet(P),$$
$$\f_\psi:\Gamma^\wedge_G\longrightarrow \Omega^\bullet(B),\qquad \f^{-1}_\psi:\Gamma^\wedge_G\longrightarrow \Omega^\bullet(B),$$ the pull--back of $f_\class$, $f^{-1}_\class$, $\psi$ and $\psi^{-1}$ on $\C$--valued differential forms, respectively. Here, $$(\Gamma^\wedge_G,d,\ast)$$  is the universal differential envelope $\ast$--calculus of the \emph{classical} $\ast$--FODC of $\C$--valued differential forms of $G$, i.e.,  $$\Gamma_G=(\SU(2)\otimes\U(1))\otimes (\mathfrak{su}(2,\C)^\#\oplus \mathfrak{u}(1,\C)^\#),$$ where $\mathfrak{g}^\#_\C=\mathfrak{su}(2,\C)^\#\oplus \mathfrak{u}(1,\C)^\#$ is the dual space of the complexification of the Lie algebra $\mathfrak{g}=\mathfrak{su}(2,\R)\oplus \mathfrak{u}(1,\R)$ of $G$, and $$\U(1):=\langle z,z^{-1}=z^\ast\rangle_{\mathrm{alg}}$$ is the canonical $\ast$--Hopf algebra associated to the Lie group $U(1)$ (\cite{woro1}), where $\langle h\rangle_{\mathrm{alg}}$ denotes the algebra generated by $h$, and $$z: U(1)\longrightarrow \C,\qquad \mathrm{e}^{it}\longmapsto \mathrm{e}^{it}.$$ By definition, we have
\begin{equation}
    \label{ec.pull}
    \f_\class(\varsigma)=\f_\psi(\varsigma^{(2)})\otimes S(\varsigma^{(1)})\,\varsigma^{(3)}
\end{equation}
for all $\varsigma$ $\in$ $\Gamma^\wedge_G$ and $\f_\class$ is a convolution invertible map (see Section 4 of reference \cite{sald1}).

We know that
\begin{equation}
    \label{ec.2.80}
    \widetilde{\beta}_\class:=\{ \hat{\varrho}^1:= {\sigma_1\over 2 }, \quad \hat{\varrho}^2:={\sigma_2\over 2 }, \quad \hat{\varrho}^3:={\sigma_3\over 2}, \quad \hat{\varrho}^4:={\Id\over 2 } \}
\end{equation}
is a linear basis of $\su(2,\C)\oplus \u(1,\C)$ with  $\sigma_1$, $\sigma_2$, $\sigma_3$ the Pauli matrices, and hence, its dual basis 
\begin{equation}
    \label{ec.2.80.1}
    \widetilde{\beta}^\#_\class=\{\varrho_1,\,\varrho_2,\,\varrho_3,\,\varrho_4 \}
\end{equation}
is a linear basis of $\mathfrak{g}^\#_\C=\su(2,\C)^\#\oplus \u(1,\C)^\#$. 

Let 
\begin{equation}
    \label{ec.3.97}
    \widehat{\f}:\SU(2)\oplus \Gamma\longrightarrow P\oplus  \Omega^1(P), \qquad \widehat{\f}^{-1}:\SU(2)\oplus \Gamma\longrightarrow P\oplus  \Omega^1(P)
\end{equation}
be the graded linear maps defined as follows: Consider the Weyl decomposition of $\SU(2)$ $$\SU(2)=\bigoplus_{\tau\in \T}\langle \{g^{\tau}_{ij}\}\rangle_{\mathrm{lin}},$$ where $\langle h\rangle_{\mathrm{lin}}$ denotes the linear space generated by $h$ and $\T$ is a complete set of mutually non--equivalent unitary (necessarily finite--dimensional) $\SU(2)$--corepresentations such that the trivial $\SU(2)$--corepresentation $\delta^\C$ on $\C$ belongs to $\T$, the standard $\SU(2)$--corepresentation $\delta^{\C^2}$ belongs to $\T$, and $\ad$ $\in$ $\T$ \cite{woro1}. For example, $$\{g^{\delta^\C}_{ij} \}=\{\mathbbm{1}\},\qquad \{g^{\delta^{\C^2}}_{ij}\}=\{\alpha,\gamma,-\gamma^\ast,\alpha^\ast\}$$ and $\{g^{\ad}_{ij}\}$ is the set of  elements of the matrix $U$ of Proposition (\ref{prop2.9}). For an element $g$ of $\{g^{\delta^{\C^2}}_{ij}\}$, define
\begin{equation}
    \label{ec.3.97.1}
\widehat{\f}(g):=\f_\class(g\otimes z^\ast), \qquad \widehat{\f}^{-1}:=\f^{-1}_\class(g\otimes z^\ast)
\end{equation}
and for any other element $g$ of $\{g^{\tau}_{ij}\}$, define
\begin{equation}
    \label{ec.3.97.2}
    \widehat{\f}(g):=\f_\class(g\otimes \mathbbm{1}), \qquad \widehat{\f}^{-1}:=\f^{-1}_\class(g\otimes \mathbbm{1}).
\end{equation}
Now extend $\widehat{\f}$ and $\widehat{\f}^{-1}$ by linearity to all of $\SU(2)$. Furthermore, define
\begin{equation}
    \label{ec.3.98}
 \widehat{\f}(\beta_j)=\f_\class(\varrho_j),\qquad   \widehat{\f}^{-1}(\beta_j)=\f^{-1}_\class(\varrho_j)
\end{equation}
for $j=1,2,3,4$. Now, we extend $\widehat{\f}$ and $\widehat{\f}^{-1}$ to $\Gamma$ by
$$\widehat{\f}(g\,\beta_j)=\widehat{\f}(g)\,\widehat{\f}(\beta_j),\qquad \widehat{\f}^{-1}(g\,\beta_j)=\widehat{\f}^{-1}(g)\,\widehat{\f}^{-1}(\beta_j).$$ It is worth mentioning that $\widehat{\f}$ preserves the $\ast$ operation except on elements of $\{g^{\delta^{\C^2}}_{ij}\}$.

Also, by defining the graded linear maps 
\begin{equation}
    \label{ecpullback}
     \widehat{\f}_\psi:\SU(2)\oplus \Gamma\longrightarrow B\oplus \Omega^1(B),\qquad \widehat{\f}^{-1}_\psi:\SU(2)\oplus \Gamma\longrightarrow B\oplus \Omega^1(B)
\end{equation}
in the same manner as $\widehat{\f}$, but replacing $\f_\class$ with $\f_\psi$, we have
\begin{equation}
    \label{necesito}
    \widehat{\f}(\vartheta)=\widehat{\f}_\psi(\vartheta^{(2)})\otimes S(\vartheta^{(1)})\,\vartheta^{(3)}
\end{equation}
for all $\vartheta$ $\in$ $\SU(2)\oplus \Gamma$ and $\widehat{\f}$ is actually a convolution invertible map \cite{sald1}.

Consider the graded linear map
\begin{equation}
    \label{ec.3.100}
    \F_\psi:P\oplus \Omega^1(P)\longrightarrow P\oplus \Omega^1(P)
\end{equation}
given by $$\F_\psi(\mu\otimes \vartheta)=(\mu\otimes\vartheta^{(1)})\,\widehat{\f}(\vartheta^{(2)})=\mu\,\widehat{\f}_\psi(\vartheta^{(1)})\otimes \vartheta^{(2)}$$ for $\mu$ $\in$ $B\oplus\Omega^\bullet(B)$, $\vartheta$ $\in$ $\SU(2)\oplus\Gamma$, $\Delta(\vartheta)=\vartheta^{(1)}\otimes \vartheta^{(2)}$. It is straightforward to prove that $\F_\psi$ is a $(B\oplus\Omega^1(B))$--module isomorphism with inverse
\begin{equation}
    \label{ec.3.100.1}
    \F^{-1}_\psi(\mu\otimes \vartheta):=\F_{\psi^{-1}}(\mu\otimes \vartheta)=\mu\,\widehat{\f}_{\psi^{-1}}(\vartheta^{(1)})\otimes \vartheta^{(2)},
\end{equation}
such that 
$$\F_{\psi}(\mathbbm{1})=\mathbbm{1},\qquad \Delta_{\Omega^\bullet(P)}\circ \F_\psi=(\F_\psi\otimes \id_{\Gamma^\wedge})\circ \Delta_{\Omega^\bullet(P)}.$$

In Section 4 of reference \cite{sald1}, quantum gauge transformations are defined without imposing any condition on the $\ast$ operation. Therefore, the map $\F_\psi$ is a quantum gauge transformation in the sense of~\cite{sald1}.

Nevertheless, for this paper, we need to prove another condition. Let $\omega$ be a qpc. According to Proposition 4.8 of reference \cite{sald1}, we have
\begin{equation}
    \label{ec.3.gauge6}
    \F^\circledast_\psi \omega=m_\Omega\circ (\omega\otimes \widehat{\f})\circ \ad+\widehat{\f},
\end{equation}
where $m_\Omega:\Omega^\bullet(P)\otimes \Omega^\bullet(P)\longrightarrow \Omega^\bullet(P)$ is the product map. By equation (\ref{ec.3.98}) and the fact that $\widehat{\f}$ preserves the $\ast$ operation on elements of $\{g^{\ad}_{ij}\}$, it follows that $$\F(\Im(\omega)^\ast)=\F(\Im(\omega))^\ast $$ and we conclude $\F$ $\in$ $\qGG$.

Finally, it is easy to see that the group morphism $$\iota: \mathfrak{GG}\longrightarrow \qGG, \qquad F_\class\longmapsto \F_{\psi} $$ is injective.
\end{proof}

Of course, in differential geometry it is impossible to embed the gauge group of a principal $(SU(2)\times U(1))$--bundle into the gauge group of a principal $SU(2)$--bundle. 

In Section 5, we will see that equation (\ref{ec.3.97.1}) in the definition of $\widehat{\f}$ ensures that all the symmetries of the \emph{classical} case are preserved in our \emph{quantum} model.

\begin{Proposition}
    \label{prop3.gauge1}
    Consider the map $\F_\psi$ of equation (\ref{ec.3.100}). In terms of the map $A^\omega$, equation (\ref{ec.3.gauge6}) is given by
    \begin{equation}
    \label{ec.3.gauge7}
     A^{\F^\circledast_\psi \omega}=m_{\Omega^\bullet}\circ (A^\omega\otimes \widehat{\f}_\psi)\circ \ad+\widehat{\f}_\psi,
\end{equation}
where $m_{\Omega^\bullet}:\Omega^\bullet(B)\otimes \Omega^\bullet(B)\longrightarrow \Omega^\bullet(B)$ is the product map.
\end{Proposition}

\begin{proof}
    Since  $$\F^\circledast_\psi \omega=m_\Omega\circ (\omega\otimes \widehat{\f})\circ \ad+\widehat{\f} $$ is again a qpc (see Section 4 of reference \cite{sald1}), by Lemma 6.11 of reference \cite{micho2}, there exists a unique hermitian linear map $$A^{\F^\circledast_\psi \omega}: \mathfrak{su}(2,\C)^\#_{4D}\longrightarrow \Omega^1(B)$$ such that 
    $$\F^\circledast_\psi \omega=(A^{\F^\circledast_\psi \omega}\otimes \id_{\SU(2)})\circ \ad+\omega^\triv.$$ In this way, considering the extension of the counit $\epsilon$ of equation (\ref{ec.2.11}), we have  $$(\id_{\Omega^\bullet(B)}\otimes \epsilon)(\F^\circledast_\psi \omega)=A^{\F^\circledast_\psi \omega}.$$ 

    Thus, for all $\theta$ $\in$ $\mathfrak{su}(2,\C)^\#_{4D}$, by equation (\ref{necesito}) we obtain 
    \begin{eqnarray*}
        \F^\circledast_\psi \omega(\theta)&=&m_\Omega(\omega\otimes \widehat{\f}) \ad(\theta)+\widehat{\f}(\theta)
        \\
        &=&
        m_\Omega(((A^\omega\otimes \id_{\SU(2)})\circ \ad+\omega^\triv )\otimes \widehat{\f}) \ad(\theta)+\widehat{\f}(\theta)
        \\
        &=&(A^\omega(\theta^{(0)})\otimes \theta^{(1)})\,\widehat{\f}(\theta^{(2)})+(\mathbbm{1}\otimes \theta^{(0)})\,\widehat{\f}(\theta^{(1)})+\widehat{f}(\theta)
        \\
        &=&
        (A^\omega(\theta^{(0)})\otimes \theta^{(1)})\,(\widehat{\f}_\psi(\theta^{(3)})\otimes S(\theta^{(2)})\, \theta^{(4)})+(\mathbbm{1}\otimes \theta^{(0)})\,(\widehat{\f}_\psi(\theta^{(2)})\otimes S(\theta^{(1)})\,\theta^{(3)})
        \\
        &+&
        \widehat{f}_\psi(\theta^{(2)})\otimes S(\theta^{(1)})\, \theta^{(3)}
        \\
        &=&
    A^\omega(\theta^{(0)})\,\widehat{\f}_\psi(\theta^{(1)})\otimes \theta^{(2)}+\widehat{\f}_\psi(\theta^{(2)})\otimes \theta^{(0)}\, S(\theta^{(1)})\,\theta^{(3)}
    +
    \widehat{f}_\psi(\theta^{(2)})\otimes S(\theta^{(1)})\,\theta^{(3)}
    \end{eqnarray*}
    and hence $$(\id_{\Omega^\bullet(B)}\otimes \epsilon)(\F^\circledast_\psi \omega)(\theta)=A^\omega(\theta^{(0)})\,\widehat{\f}_\psi(\theta^{(1)})+\widehat{\f}_\psi(\theta)=m_{\Omega^\bullet}(A^\omega\otimes \widehat{\f}_\psi)\ad(\theta)+\widehat{\f}_\psi(\theta).$$ We conclude that 
    $$A^{\F^\circledast_\psi \omega}=m_{\Omega^\bullet}\circ (A^\omega\otimes \widehat{\f}_\psi)\circ \ad+\widehat{\f}_\psi.$$
\end{proof}

Consider the principal $G$--bundle $$\pi_\class:\R^4\times G\longrightarrow \R^4,\qquad (x,C)\longmapsto x$$ and its \emph{canonical} section $$s:\R^4\longrightarrow \R^4\times G,\qquad x\longmapsto (x,e),$$ where $e$ is the identity element of $G=SU(2)\times U(1)$. For a given principal connection $\omega_\class$ of this bundle, we have
\begin{equation}
    \label{ec.con.1}
    s^\#\omega_\class: T\R^4\longrightarrow \mathfrak{g}, \qquad X_x\longmapsto \omega_\class(ds(X_x)),
\end{equation}
where $\mathfrak{g}=\su(2,\R)\oplus \mathfrak{u}(1,\R)$ is the Lie algebra of $G$. It is well--known that the action of the action of a gauge transformation $$F_\class:\R^4\times G\longrightarrow \R^4\times G,\qquad (x,C)\longmapsto (x,\psi(x)\,C) $$ of $\pi_\class:\R^4\times G\longrightarrow \R^4$  on $s^\#\omega_\class$ is given by (\cite{diff1})
\begin{equation}
\label{ec.3.gauge8}
(s^\#\omega_\class,\psi)\longmapsto \psi s^\#\omega_\class
:= \ad^\class_{\psi}\circ s^\#\omega_\class+\mathrm{MC}\circ d\psi,
\end{equation}
where $$\ad^\class: \mathfrak{g}\times G\longrightarrow \mathfrak{g},\qquad (v,C)\longmapsto \left.\ad^\class_C(v):={d\, C^{-1}\,\mathrm{exp}(t\,v)\,C\over dt }\right|_{t=0}$$ is the \emph{classical} adjoint action of $G$ on $\mathfrak{g}$, and $$\mathrm{MC}: TG\longrightarrow \mathfrak{g}, \qquad Y_C\longmapsto dL_{C^{-1}}(Y_C)$$ is the Maurer--Cartan form. In a more \emph{traditional} notation, we get
\begin{equation}
    \label{ec.3.gauge10}
    \psi s^\#\omega_\class=\psi^{-1}\cdot s^\#\omega_\class\cdot \psi+\psi^{-1} \cdot d\psi.
\end{equation}

The pull--back of the complexification of $s^\#\omega_\class$ gives us a hermitian linear map 
\begin{equation}
    \label{ec.con.2}
    A_\class: \mathfrak{g}^\#_\C:=\mathfrak{su}(2,\C)^\#\oplus \mathfrak{u}(1,\C)^\#\longrightarrow  \Omega^\bullet(B)
\end{equation}
and we can define another hermitian linear map
\begin{equation}
    \label{ec.con.3}
    A^\omega: \su(2,\C)^\#_{4D}\longrightarrow \Omega^1(B)
\end{equation}
given by $$A^\omega(\beta_j):=A_\class(\varrho_j)$$ for $j=1,2,3,4$ (see equation (\ref{ec.2.80.1})). By equation (\ref{ec.3.56}), the map $A^\omega$ defines a qpc $\omega$ on our quantum bundle $\zeta_{4D}$.

\begin{Proposition}
\label{propgauge}
    The pull--back of the complexification of the map $\psi s^\#\omega_\class$ (equation \ref{ec.3.gauge8}) coincides with equation (\ref{ec.3.gauge7}), for the map $ A^\omega$ of equation (\ref{ec.con.3}).
\end{Proposition}

\begin{proof}
Notice that $$\ad^\class_\psi  \circ s^\#\omega_\class: T\R^4\longrightarrow \mathfrak{g}, \qquad X_x\longmapsto \ad^\class_{\psi(x)}(\omega_\class(ds(X_x))) $$ and hence we can consider that $$\ad^\class_\psi  \circ s^\#\omega_\class=\ad^\class\circ (s^\#\omega_\class\times\psi ):T\R^4\longrightarrow \mathfrak{g}$$ is given by  $$X_x\longmapsto (s^\#\omega_\class(X_x),\psi(x))\longmapsto \ad^\class_{\psi(x)}(\omega_\class(ds(X_x))). $$ In this way, the pull--back of the complexification of $\ad^\class_\psi  \circ s^\#\omega_\class$ is given by $$m_{\Omega^\bullet}\circ (A_\class\otimes \f_\psi)\circ \ad^{\#},$$ where $$\ad^{\#}:\mathfrak{g}^\#_\C\longrightarrow \mathfrak{g}^\#_\C\otimes (\SU(2)\otimes \U(1)), \qquad \theta \longmapsto \theta\circ \ad^\class$$  is the pull--back of $\ad^{\class}$ on the complexification $\mathfrak{g}^\#_\C$.

    In order to make easier the calculations, consider the linear basis of $\mathfrak{g}_\C$ (see equation (\ref{ec.2.80})) 
    $$\{\hat{\xi}^1:={1\over \sqrt{2}} (\hat{\varrho}^1 -i\,\hat{\varrho}^2),\quad \hat{\xi}^2:={1\over \sqrt{2}}(\hat{\varrho}^1 +i\,\hat{\varrho}^2),\quad  \hat{\xi}^3:=\hat{\varrho}^3, \quad \hat{\xi}^4:=\hat{\varrho}^4 \}$$ and its dual basis in $\mathfrak{g}^\#_\C$ $$\{{\xi}_1:={1\over \sqrt{2}} ({\varrho}_1 -i\,{\varrho}_2),\quad {\xi}_2:={1\over \sqrt{2}}({\varrho}_1 +i\,{\varrho}_2),\quad  {\xi}_3:={\varrho}_3, \quad {\xi}_4:={\varrho}_4  \}.$$ A direct calculation shows that 
    \begin{equation*}
    \ad^{\#}(\xi_1)=(\xi_1\otimes \alpha^2+\xi_2\otimes -\gamma^2+\xi_3\otimes -\sqrt{2}\,\alpha\,\gamma)\otimes \mathbbm{1},
\end{equation*}
\begin{equation*}
    \ad^{\#}(\xi_2)=(\xi_1\otimes -\gamma^{\ast 2}+\xi_2\otimes \alpha^{\ast 2}+\xi_3\otimes -\sqrt{2}\,\alpha^\ast\,\gamma^\ast)\otimes \mathbbm{1},
\end{equation*}
\begin{equation*}
    \ad^{\#}(\xi_3)=(\xi_1\otimes \sqrt{2}\,\alpha\,\gamma^\ast+\xi_2\otimes \sqrt{2}\,\alpha^\ast\,\gamma+\xi_3\otimes(\alpha\,\alpha^\ast-\gamma\,\gamma^\ast))\otimes \mathbbm{1}.
\end{equation*}
\begin{equation*}
    \ad^{\#}(\xi_4)=\xi_4\otimes \mathbbm{1}\otimes \mathbbm{1}.
\end{equation*}
Notice that the matrix associated with $\ad^{\#}$ is $$\left(
\begin{array}{cccc}  
\alpha^2 &  -\gamma^{\ast 2} & \sqrt{2}\,\alpha\,\gamma^\ast &0 \\
-\gamma^2 & \alpha^{\ast 2} & \sqrt{2}\,\alpha^\ast\,\gamma &0 \\
-\sqrt{2}\,\alpha\,\gamma & -\sqrt{2}\,\alpha^\ast\,\gamma^\ast & \alpha\alpha^\ast-\gamma\gamma^\ast & 0\\
0 & 0 &0 & 1
\end{array}
\right)\otimes \mathbbm{1},$$ which is equivalent to the one of Proposition  \ref{prop2.9} and therefore, we conclude that $$\ad^{\#}(\xi_j)\otimes \mathbbm{1}=\ad(\eta_j)$$ for $j=1,2,3,4$. Since $$\{{\eta}_1:={1\over \sqrt{2}} ({\beta}_1 -i\,{\beta}_2),\quad {\eta}_2:={1\over \sqrt{2}}({\beta}_1 +i\,{\beta}_2),\quad  {\eta}_3:={\beta}_3, \quad {\eta}_4:={\beta}_4  \},$$ we get $$\ad^{\#}(\varrho_j)\otimes \mathbbm{1}=\ad(\beta_j).$$ By definition, $$A_\class(\varrho_j)=A^\omega(\beta_j),\qquad  \f_\psi(\varrho_j)= \widehat{\f}_\psi(\beta_j) $$ and hence $$m_{\Omega^\bullet}\circ (A_\class\otimes \f_\psi)\circ \ad^{\#}=m_{\Omega^\bullet}\circ (A^\omega\otimes \widehat{\f}_\psi)\circ \ad.$$

On the other hand, under the isomorphism $$G\times \mathfrak{g}\;\longleftrightarrow\; TG,\qquad (C,v)\;\longleftrightarrow \; dL_C(v),$$ the Maurer--Cartan form is $$\mathrm{MC}:TG\longrightarrow \mathfrak{g},\qquad (C,v)\longmapsto v. $$ Thus, the pull--back of $\mathrm{MC}$ on $\mathfrak{g}^\#_\C$ is given by $$\mathrm{MC}^\#:\mathfrak{g}^\#_\C \longrightarrow \Gamma_G=(\SU(2)\otimes \U(1))\otimes \mathfrak{g}^\#_\C\cong (\SU(2)\otimes \U(1))\, \mathfrak{g}^\#_\C,\qquad \theta \longmapsto \theta \circ \mathrm{MC}= \theta$$ and therefore, the pull--back of $\mathrm{MC}\circ d\psi$ is $$(\mathrm{MC}\circ d\psi)^\#:\mathfrak{g}^\#_\C\longrightarrow \Omega^\bullet(B),\qquad \theta \longmapsto \theta\circ \mathrm{MC}\circ d\psi=  \theta\circ d\psi=\f_\psi(\theta).$$  By definition, we know $$(\mathrm{MC}\circ d\psi)^\#(\varrho_j)=\widehat{\f}_\psi(\beta_j) $$ for $j=1,2,3,4$, and the proposition follows.
\end{proof}

According to the previous proposition, the action of quantum gauge transformations of the form $\F_\psi$ on the space $\mathfrak{qpc}(\zeta_{4D})$ is the dualization, via pull--back, of the action  of the gauge group $\mathfrak{GG}$ associated with the principal bundle $\pi_\class:\R^4\times G\longrightarrow \R^4$ on its space of principal connections, with $G=SU(2)\times U(1)$.

\section{The Yang--Mills Theory of the Electroweak Interaction}

As in the other sections, we will start this section with a brief summary of the general theory in non--commutative geometry.

\subsection{General Theory}
In Durdevich's formulation of quantum principal bundles, there is a Yang--Mills theory, as the reader can check in reference \cite{sald2,saldcon}. In order to use this theory, in accordance with reference \cite{sald2}, it is necessary
    \begin{enumerate}
    \item A quantum principal $A$--bundle $\zeta=(P,B,\Delta_P)$ for which $B$ is a $C^\ast$--algebra or can be completed to one. 
    \item A bicovariant $\ast$--FODC $(\Gamma,d)$ for which the quantum dual Lie algebra $\mathfrak{qa}^\#$ is finite--dimensional. 
    \item An inner product $\langle-|-\rangle_{\mathfrak{qa}^\#}$ of $\mathfrak{qa}^\#$ for which the $A$--corepresentation $\ad$ is unitary. 
    \item A differential calculus on $\zeta$ using the previous bicovariant $\ast$--FODC for which the space of base forms $\Omega^\bullet(B)$ satisfies
    \begin{enumerate}
        \item There exists $n$ $\in$ $\N$ for which $$\Omega^n(B)=B\,\dvol, \qquad \Omega^{n+1}(B)=0,$$ where $\dvol$ $\in$ $\Omega^n(B)$ fulfills $$b\,\dvol=0\;\;\Longrightarrow\;\; b=0.$$
        \item There is a family of $B$--valued non--degenerated sesquilinear form (antilinear in the first coordinate) $$\{\langle-,-\rangle_\r:\Omega^k(B)\times\Omega^k(B)\longrightarrow B\}$$ where for $k=0$ 
\begin{equation*}
\begin{aligned}
\langle-,-\rangle_\r\;:\; B\times B&\longrightarrow B\\
(b_1\,,\,b_2)&\longmapsto b^\ast_1\,b_2
\end{aligned}
\end{equation*} 
and for $k=n$
\begin{equation*}
\begin{aligned}
\langle-,-\rangle_\r:\Omega^n(B)\times \Omega^n(B) \quad &\longrightarrow \quad B\\
(\;b_1\,\dvol\;,\;b_2\,\dvol\;)\;&\longmapsto \quad b^\ast_1\,b_2,
\end{aligned}
\end{equation*} 
such that $$\langle b\,\mu_1,\mu_2\rangle_\r=\langle \mu_1,b^\ast\,\mu_2\, \rangle_\r $$ for all $\mu_1$, $\mu_2$ $\in$ $\Omega^k(B)$, $b$ $\in$ $B$ and $ 1  \leq k< n$. We can combine these maps into a $B$--valued non--degenerated sesquilinear map form   $$\langle-,-\rangle_\r:\Omega^\bullet(B)\times \Omega^\bullet(B)\longrightarrow B$$ by postulating the orthogonality between elements of different degrees.

\item Let $s$ be a faithful state of $B$. We define a quantum integral on $B$ as 
\begin{equation*}
\begin{aligned}
\int_{B}: \Omega^{n}(B)\longrightarrow \C,\qquad 
b\,\dvol\longmapsto s(b) 
\end{aligned}
\end{equation*}
and we can interpret that a given quantum integral satisfies the Stokes theorem by explicitly defining
\begin{equation*}
\begin{aligned}
\int_{\partial B}: \Omega^{n-1}(B)\longrightarrow \C,\qquad
\mu\longmapsto \int_B d\mu.
\end{aligned}
\end{equation*}
We require the existence of a quantum integral for which $\Im(d)$ $\subseteq$ $\Ker\left(\displaystyle\int_B\right)$. 
\item There exists a linear isomorphism  $$\star_\r: \Omega^k(B)\longrightarrow \Omega^{n-k}(B) $$ and it satisfies $$\eta^\ast\,\mu=\langle \eta,\star^{-1}_\r\mu\rangle_\r\,\dvol.$$
    \end{enumerate}
\item For every qpc $\omega$ of $\zeta$, the operator 
\begin{equation}
    \label{ec.4.1.1}
d^{\widehat{S}^{\omega}}:=\widehat{\Upsilon}_{\mathfrak{qa}^\#}\circ \widehat{S}^{\omega} \circ \widehat{\Upsilon}^{-1}_{\mathfrak{qa}^\#}
\end{equation}
is adjointable or formally adjointable in (see equation (\ref{ec.3.19.5.2})) $$ \widehat{\Upsilon}_{\mathfrak{qa}^\#}(\Mor(\ad,\Delta_\H)^\dagger)$$  with respect to the non--degenerated sesquilinar form 
\begin{equation}
    \label{ec.4.2.1}
    \langle-|-\rangle^\bullet_\r: (\Mor(\ad,\Delta_P)\otimes_B \Omega^\bullet(B) )\times (\Mor(\ad,\Delta_P)\otimes_B \Omega^\bullet(B)) \longrightarrow \C
\end{equation}
given by $$\langle T_1\otimes_B \mu_1\mid T_2\otimes_B \mu_2 \rangle^\bullet_\r=\int_B \langle \mu_1,\,(T_1,T_2)_\r\,\mu_2\rangle_\r\,\dvol,$$ where $\mu_1$, $\mu_2$ $\in$ $\Omega^\bullet(B)$, $T_1$, $T_2$ $\in$ $\Mor(\ad,\Delta_P)$. Here, 
\begin{equation}
    \label{ec.4.4.1}
    (-,-)_\r: \Mor(\ad,\Delta_P)\times \Mor(\ad,\Delta_P)\longrightarrow B 
\end{equation}
is given by $$(T_1,T_2)_\r=\sum^m_{k=1} T_1(\theta_k)^\ast\,T_2(\theta_k),$$ with $\{\theta_k \}^m_{k=1}$ an orthonormal basis of $\mathfrak{qa}^\#$ with respect to $\langle-|-\rangle_{\mathfrak{qa}^\#}$, and the map 
\begin{equation}
    \label{ec.4.5.1}
    \widehat{\Upsilon}_{\mathfrak{qa}^\#}: \Mor(\ad,\Delta_\Hor)\longrightarrow \Mor(\ad, \Delta_P)\otimes_B \Omega^\bullet(B)
\end{equation}
is the right $\Omega^\bullet(B)$--module isomorphism defined in equation (87) of reference \cite{sald2}.
\end{enumerate}

Under the previous situation, we have

\begin{Definition}
\label{4.def1}
Considering the second part of equation (\ref{ec.3.14}), we define the right non--commutative geometrical Yang--Mills action as the association 
\begin{equation*}
\begin{aligned}
\qS_{\YM_\r}:\mathfrak{qpc}(\zeta) &\longrightarrow \R\\
\omega &\longmapsto -\dfrac{1}{2}\,||R^\omega||^{\bullet\,2}_\r:=-\dfrac{1}{2}\,\langle \widehat{\Upsilon}_{\mathfrak{qa}^\#}(R^\omega)\mid \widehat{\Upsilon}_{\mathfrak{qa}^\#}(R^\omega)\rangle^\bullet_\r.
\end{aligned}
\end{equation*}
\end{Definition}
\noindent In addition, in total analogy with the {\it classical} case, we get

\begin{Definition}
\label{4.def3}
 A stationary point of $\qS_{\YM_\r}$ is an element $\omega$ $\in$ $\mathfrak{qpc}(\zeta)$ such that for any $\lambda$ $\in$ $\overrightarrow{\mathfrak{qpc}(\zeta)}$ (see equation (\ref{ec.3.7}))  $$\left.\dfrac{d}{d t}\right|_{t=0}\qS_{\YM_\r}(\omega + t\,\lambda)=0$$ for $t$ $\in$ $\R$. Stationary points are also called Yang--Mills qpc's or non--commutative geometrical Yang--Mills fields. In terms of a traditional physical interpretation, they should be considered as {\it free--interaction gauge boson fields}. 
\end{Definition}

It is worth recalling that $\overrightarrow{\mathfrak{qpc}(\zeta)}$ is a $\R$--vector space, so the expression $t\,\lambda$ makes sense for $t$ $\in$ $\R$. Moreover, $\qS_\YM(\omega+t\,\lambda)$ depends polynomially on $t$, so we interpret the derivative with respect to $t$ of $\qS_\YM
$ as a formal derivative in the obvious way \cite{sald2}.

According to Theorem 5.16 of reference \cite{saldcon}, a qpc $\omega$ is a critical point of $\qS_{\YM_\r}$ if and only if 
\begin{equation}
    \label{ec.4.6.1}
    (d^{\widehat{\nabla}^\omega_{\mathfrak{qa}^\#}\star}-d^{\widehat{S}^\omega\star})(R^\omega):=(d^{\widehat{\nabla}^\omega_{\mathfrak{qa}^\#}\star}-d^{\widehat{S}^\omega\star})(\widehat{\Upsilon}_{\mathfrak{qa}^\#}(R^\omega))=0.
\end{equation}
Equation (\ref{ec.4.6.1}) is called the {\it right non--commutative geometrical Yang--Mills equation}. Here
\begin{equation}
    \label{ec.4.7.1}
d^{\widehat{\nabla}^\omega_{\mathfrak{qa}^\#}\star}:\Mor(\ad, \Delta_P)\otimes_B \Omega^\bullet(B)\longrightarrow  \Mor(\ad, \Delta_P)\otimes_B \Omega^\bullet(B)
\end{equation}
is the formal adjoint operator of 
\begin{equation}
    \label{ec.4.8.1}
d^{\widehat{\nabla}^\omega_{\mathfrak{qa}^\#}}:=\widehat{\Upsilon}_{\mathfrak{qa}^\#}\circ \widehat{D}^{\omega} \circ \widehat{\Upsilon}^{-1}_{\mathfrak{qa}^\#}  :\Mor(\ad, \Delta_P)\otimes_B \Omega^\bullet(B) \longrightarrow  \Mor(\ad, \Delta_P)\otimes_B \Omega^\bullet(B)
\end{equation}
with respect to  $\langle-|-\rangle^\bullet_\r$, and 
\begin{equation}
    \label{ec.4.9.1}
    d^{\widehat{S}^\omega\star}: \Mor(\ad, \Delta_P)\otimes_B\Omega^\bullet(B) \longrightarrow  \Mor(\ad, \Delta_P) \otimes_B\Omega^\bullet(B) 
\end{equation}
is the formal adjoint operator of $d^{\widehat{S}^{\omega}}$  with respect to $\langle-|-\rangle^\bullet_\r$. It is worth mentioning that, in light of Theorem 27 of reference \cite{sald2}, the operator $d^{\widehat{\nabla}^\omega_{\mathfrak{qa}^\#}\star_\r}$ is defined by
\begin{equation}
    \label{ec.4.10.1}
    d^{\widehat{\nabla}^\omega_{\mathfrak{qa}^\#}\star}:=(-1)^{k+1}(\id\otimes_B \star^{-1}_\r )\circ d^{\widehat{\nabla}^\omega_{\mathfrak{qa}^\#}} \circ (\id \otimes_B \star_\r)
\end{equation}
on elements of the space $$\Mor(\ad,\Delta_P)\otimes_B \Omega^{k+1}(B).$$ However, there is no a general form of the operator $d^{\widehat{S}^\omega\star}$ and this is why the formal adjointability of the operator $d^{\widehat{S}^\omega}$ is a condition for the theory.

Now, considering Definition \ref{3.def1} and equation (\ref{ec.3.20}), we have (\cite{sald2})
\begin{Definition}
\label{4.def2}
We define the right quantum gauge group of the Yang--Mills model as the group $$\qGG_{\YM_\r}:=\{\F \in \qGG\mid \qS_{\YM_\r}(\omega)=\qS_{\YM_\r}(\F^{\circledast} \omega)\;\mbox{ for all }\, \omega \in \mathfrak{qpc}(\zeta) \} \subseteq \qGG.$$
\end{Definition}

As noted in the previous section, in general, the quantum gauge group is larger than its counterpart in differential geometry. This motivates Definition~(\ref{4.def2}): the group $\qGG$ is sufficiently large that not all of its elements are symmetries of the non--commutative geometrical Yang--Mills actions. At first sight, this might appear to be a drawback of the theory, since \emph{not every quantum gauge transformation} is a symmetry of $\qS_{\YM_\r}$. However, as will be shown in this section, $\qGG_{\YM_\r}$ contains all the necessary elements.

The subscript $\r$ indicates that we are using the right $\Omega^\bullet(B)$--module structure of $$\Mor(\ad,\Delta_\Hor),$$ in accordance with the notation adopted in \cite{sald1,sald2}; and for the right structure of $\Mor(\ad,\Delta_\Hor)$, it is necessary to use the operators $\widehat{D}^\omega$, $\widehat{S}^\omega$ to develop the theory \cite{sald1,sald2}. Analogous definitions to those presented in this subsection can be formulated for the left $\Omega^\bullet(B)$--module structure of $\Mor(\ad,\Delta_\Hor)$; and in this situation, it is necessary to use the operators $D^\omega$, $S^\omega$ to develop the theory. See \cite{sald2} for details.

For the qpb $$\zeta_{4D},$$ the information about the curvature obtained from the right $\Omega^\bullet(B)$--module structure of $\Mor(\ad,\Delta_\Hor)$ coincides with the information about the curvature obtained from the left $\Omega^\bullet(B)$--module structure. Therefore, in this work we restrict ourselves to the right $\Omega^\bullet(B)$--module structure. However, for other corepresentations, this equivalence no longer holds, as we shall see in the next sections.

\subsection{The Model}

Consider the quantum principal $\SU(2)$--bundle $$\zeta_{4D}=(P:=B\otimes \SU(2),B,\Delta_P:=\id_B\otimes \Delta)$$ defined in equation (\ref{ec.3.50}), equipped with the differential calculus given in equation (\ref{ec.3.51}). 

It is clear that $\zeta_{4D}$ satisfy the conditions (1)--(3) listed at the beginning of this section. For the point (4), consider: 
\begin{enumerate}
    \item Point (4a) is satisfied by defining $$\dvol=dx^0\wedge dx^1\wedge dx^2\wedge dx^3.$$
    \item Point (4b) is satisfied by taking the complexification of the Minkowski metric of $\R^4$ on differential forms. For example
    $$\langle \sum^{3}_{\mu=0} f_\mu\, dx^\mu, \sum^3_{\nu=0} h_\nu\, dx^\nu\rangle_\r=\sum^3_{\mu,\nu=0}\eta^{\mu\nu}\,f^\ast_\mu\,h_\nu=f^\ast_0\,h_0-f^\ast_1\,h_1-f^\ast_2\,h_2-f^\ast_3\,h_3$$ in elements of $\Omega^1(B)$, with $\eta^{\mu\alpha}=\eta_{\mu\alpha}=\mathrm{diag}(+,-,-,-)$.
    \item Point (4c) is satisfied by defining the quantum integral as the usual integral on $\R^4$. In this case, we have $\displaystyle \Im(d)\subseteq \Ker\left(\int_B \right)$.
    \item Defining $\star_\r$ as the linear extension of the \emph{classical} Hodge star operator on $\R$--valued differential forms of $\R^4$ ensures that condition (4d) holds.
\end{enumerate}

However, verifying that point (5) is also satisfied is not trivial. In Section 5.2 of reference \cite{sald2} it is proven that the operator $d^{\widehat{S}^\omega}$ admits a formal adjoint operator $d^{\widehat{S}^\omega\star_\r}$, when the quantum base space is the Moyal--Weyl algebra. Nevertheless, taking the \emph{classical} limit ($\Lambda^{\mu\nu}\longrightarrow 0$), the Moyal--Weyl algebra turns into $B=C^\infty_\C(\R^4)$ and hence, the proof also holds for our case. Thus, point (5) is satisfied.

Let us define the inner product $$\langle-|-\rangle^\r_{\mathrm{Hom}}: \mathrm{Hom}_\C(\mathfrak{su}(2,\C)^\#_{4D},\Omega^\bullet(B))\times \mathrm{Hom}_\C(\mathfrak{su}(2,\C)^\#_{4D},\Omega^\bullet(B))\longrightarrow \C$$ given by
$$\langle L^{\tau_1} \mid L^{\tau_2} \rangle^\r_{\mathrm{Hom}}:=\sum^4_{k=1}\langle L^{\tau_1}(\theta_k) \mid L^{\tau_2}(\theta_k) \rangle_\r:=\sum^4_{k=1}\int_B \langle L^{\tau_1}(\theta_k) , L^{\tau_2}(\theta_k) \rangle_\r \, \dvol,$$
where  $\{\theta_k \}^4_{k=1}$ is any orthonormal basis of $\mathfrak{su}(2,\C)^\#_{4D}$. Here $$\mathrm{Hom}_\C(\mathfrak{su}(2,\C)^\#_{4D},\Omega^k(B))^\dagger$$ is the space of all hermitian linear maps between $\mathfrak{su}(2,\C)^\#_{4D}$ and $\Omega^k(B)$. In particular, by equation (\ref{ec.3.42.1.1}), we obtain
\begin{equation}
    \label{used8}
    F^\omega\;\in\; \mathrm{Hom}_\C(\mathfrak{su}(2,\C)^\#_{4D},\Omega^k(B))^\dagger.
\end{equation}

In accordance with Section 5.2 of reference \cite{sald2}, in our case, in terms of the map $A^\omega$, the right non--commutative geometrical Yang--Mills functional is given by 
\begin{equation*}
\begin{aligned}
    A^\omega\;\;\longmapsto \;\;-{1\over 2}\,\langle F^\omega\mid F^\omega \rangle^\r_{\mathrm{Hom}}&=-{1\over 2}\,\sum^4_{k=1}\int_B \langle F^\omega(\beta_k) , F^{\omega}(\beta_k) \rangle_\r \, \dvol
    \\
    &=-{1\over 2}\,\sum^4_{k=1} \int_B F^\omega(\beta_k)^\ast\wedge \star_\r(F^\omega(\beta_k))
    \\
    &=-{1\over 2}\,\sum^4_{k=1} \int_B F^\omega(\beta_k)\wedge \star(F^\omega(\beta_k)).
\end{aligned}
\end{equation*}
In terms of $A^j_\mu$, we have
\begin{equation}
    \label{ec.4.21.2}
    A^j_\mu\;\;\longmapsto \;\; -{1\over 4}\,\int_{\R^4}\,\sum^4_{j=1} F^{j}_{\mu\nu}\,F^{j\,\mu\nu}\,\dvol.
\end{equation}
Therefore, the right non--commutative Yang--Mills functional coincides with the {\it classical} Yang--Mills functional of the electroweak theory (\cite{notation}).

Furthermore, in light of Section 5.2 of \cite{sald2}, in order to determine the explicit form of the operator $d^{S^\omega\star}$, it is sufficient to compute the formal adjoint of the operator (see equations (\ref{ec.3.64}), (\ref{used4})) $$L^\tau\;\; \longmapsto\;\; 
\widehat{S}^{\omega}_B(L^\tau)=\wedge \circ S^{\omega}_B \circ \wedge$$ 
with respect to $\langle-|-\rangle^\r_{\mathrm{Hom}}$ in $\mathrm{Hom}_\C(\mathfrak{su}(2,\C)^\#_{4D},\Omega^k(B))^\dagger$.

It is worth mentioning that by defining the non--degenerate sesquilinear map
$$\langle-|-\rangle_\r:=\int_B \langle-,-\rangle_\r\,\dvol: \Omega^\bullet(B)\times \Omega^\bullet(B)\longrightarrow \C, $$
the following equations hold (\cite{sald2,diff1})
\begin{equation}
    \label{hodgeformula2}
    \langle \psi^\ast_2\,\psi_1\mid \star^{-1}_\r(\psi_3)\rangle_\r=(-1)^{ml}\langle \psi_1\mid \star^{-1}_\r(\psi_2\,\psi_3)\rangle_\r
\end{equation}
    for all $\psi_1$ $\in$ $\Omega^l(B)$, $\psi_2$ $\in$ $\Omega^m(B)$, $\psi_3$ $\in$ $\Omega^k(B)$ such that 
$l+m+k=n=4$. In addition, since $\Omega^\bullet(B)$ is graded--commutative, we immediately have that
\begin{equation}
    \label{hodgeformula3}
    \langle \psi_1\,\psi^\ast_3\mid \star^{-1}_\r(\psi_2)\rangle'_\r=(-1)^{mk}\langle \psi_1\mid \star^{-1}_\r(\psi_2\,\psi_3)\rangle_\r.
\end{equation}

\begin{Theorem}
    \label{the4.1}
    For the embedded differential $\Theta$ of equation (\ref{ec.2.98}), the operator $$\widehat{S}^{\omega\,\star}_B:=(-1)^{k+1}\,\star^{-1}_\r \circ \; \widehat{S}^\omega_B \circ  \star_\r $$ is the formal adjoint operator of $\widehat{S}^\omega_B$ in $\mathrm{Hom}_\C(\mathfrak{su}(2,\C)^\#_{4D},\Omega^k(B))^\dagger$.
\end{Theorem}
\begin{proof}
The proof is very easy using the base $\beta^\#_{4D}$ of equation (\ref{ec.2.89}), Proposition \ref{prop2.6} and equations (\ref{ec.2.98})--(\ref{ec.2.102}). In fact,  let $L^\tau$ $\in$ $\mathrm{Hom}_\C(\mathfrak{su}(2,\C)^\#_{4D},\Omega^{k}(B))^\dagger$. Then, by equation (\ref{used4}) 
     \begin{eqnarray*}
         S^{\omega}_B(L^\tau)(\eta_1)&=&\langle A^\omega,L^\tau\rangle(\eta_1)
-(-1)^k\langle L^\tau,A^\omega\rangle(\eta_1)
-(-1)^k[L^\tau,A^\omega](\eta_1)\\
&=& (-1)^k\,{i\,q_4\over 2} \,L^\tau(\eta_1)\wedge A^\omega(\eta_4),
     \end{eqnarray*}
     \begin{eqnarray*}
         S^{\omega}_B(L^\tau)(\eta_2)&=&\langle A^\omega,L^\tau\rangle(\eta_2)
-(-1)^k\langle L^\tau,A^\omega\rangle(\eta_2)
-(-1)^k[L^\tau,A^\omega](\eta_2)\\
&=& (-1)^k\,{i\,3\,q_4\over 2} \,L^\tau(\eta_2)\wedge A^\omega(\eta_4),
     \end{eqnarray*}
     \begin{eqnarray*}
         S^{\omega}_B(L^\tau)(\eta_3)&=&\langle A^\omega,L^\tau\rangle(\eta_3)
-(-1)^k\langle L^\tau,A^\omega\rangle(\eta_3)
-(-1)^k[L^\tau,A^\omega](\eta_3)\\
&=& (-1)^k\,i\,q_4\,L^\tau(\eta_3)\wedge A^\omega(\eta_4),
     \end{eqnarray*}
and
     \begin{eqnarray*}
         S^{\omega}_B(L^\tau)(\eta_4)&=&\langle A^\omega,L^\tau\rangle(\beta_4)
-(-1)^k\langle L^\tau,A^\omega\rangle(\beta_4)
-(-1)^k[L^\tau,A^\omega](\beta_4)\\
&=& 0.
     \end{eqnarray*}
In this way, for every element $L^\tau$ $\in$ $\mathrm{Hom}_\C(\mathfrak{su}(2,\C)^\#_{4D},\Omega^{k}(B))^\dagger$ we have $$\widehat{S}^{\omega}_B(L^\tau)=\widehat{S^{\omega}_B(\widehat{L^\tau})}=\widehat{S^{\omega}_B(L^\tau)}=\ast\circ S^{\omega}_B(L^\tau)\circ \ast;$$ so
\begin{eqnarray*}
    \widehat{S}^{\omega}_B(L^\tau)(\eta_1)=(\ast\circ S^{\omega}_B(L^\tau)\circ \ast)(\eta_1)&=&(S^{\omega}_B(L^\tau)(\eta_2))^\ast
    \\
    &=&
    -(-1)^k\,{i\,3\,q_4\over 2} \,(L^\tau(\eta_2)\wedge A^\omega(\eta_4))^\ast
    \\
    &=&
    -\,{i\,3\,q_4\over 2} \,A^\omega(\eta_4)^\ast\wedge L^\tau(\eta_2)^\ast
    \\
    &=&
    -\,{i\,3\,q_4\over 2} \,A^\omega(\eta^\ast_4)\wedge L^\tau(\eta^\ast_2)
     \\
    &=&
    -\,{i\,3\,q_4\over 2} \,A^\omega(\eta_4)\wedge L^\tau(\eta_1)
    \\
    &=&
    (-1)^{k+1}\,{i\,3\,q_4\over 2} \,L^\tau(\eta_1)\wedge A^\omega(\eta_4)
\end{eqnarray*}
and similarly, we obtain
$$\widehat{S}^{\omega}_B(L^\tau)(\eta_2)=(-1)^{k+1}\,{i\,q_4\over 2} \,L^\tau(\eta_2)\wedge A^\omega(\eta_4),\qquad \widehat{S}^{\omega}_B(L^\tau)(\eta_3)=(-1)^{k+1}\,i\,q_4\,L^\tau(\eta_3)\wedge A^\omega(\eta_4),$$  $$\widehat{S}^{\omega}_B(L^\tau)(\eta_4)=0.$$

Let $L^\sigma$ $\in$ $\mathrm{Hom}_\C(\mathfrak{su}(2,\C)^\#_{4D},\Omega^{k+1}(B))^\dagger$. For $n=4$, we have 
\begin{eqnarray*}
         \widehat{S}^{\omega\,\star}_B(L^\sigma)(\eta_1)&=&(-1)^{k+1}\star^{-1}_\r(\widehat{S}^{\omega}_B(\star_\r L^\sigma)(\eta_1) )
\\
&=& 
(-1)^{k+1}\,(-1)^{n-k}\,\,{i\,3\,q_4\over 2}\star^{-1}_\r\left(\star_\r L^\tau(\eta_1)\wedge A^\omega(\eta_4)\right)
\\
&=&
(-1)^{n+1}\,\,{i\,3\,q_4\over 2}\star^{-1}_\r\left(\star_\r L^\tau(\eta_1)\wedge A^\omega(\eta_4)\right),
\end{eqnarray*}
\begin{eqnarray*}
         \widehat{S}^{\omega\,\star}_B(L^\sigma)(\eta_2)&=&(-1)^{k+1}\star^{-1}_\r(\widehat{S}^{\omega}_B(\star_\r L^\sigma)(\eta_2) )
\\
&=& 
(-1)^{k+1}\,(-1)^{n-k}\,\,{i\,q_4\over 2}\star^{-1}_\r\left(\star_\r L^\tau(\eta_2)\wedge A^\omega(\eta_4)\right)
\\
&=&
(-1)^{n+1}\,\,{i\,q_4\over 2}\star^{-1}_\r\left(\star_\r L^\tau(\eta_2)\wedge A^\omega(\eta_4)\right),
     \end{eqnarray*}
\begin{eqnarray*}
         \widehat{S}^{\omega\,\star}_B(L^\sigma)(\eta_3)&=&(-1)^{k+1}\star^{-1}_\r(\widehat{S}^{\omega}_B(\star_\r L^\sigma)(\eta_3) )
\\
&=& 
(-1)^{k+1}\,(-1)^{n-k}\,i\,q_4\,\star^{-1}_\r\left(\star_\r L^\tau(\eta_3)\wedge A^\omega(\eta_4)\right)
\\
&=&
(-1)^{n+1}\,i\,q_4\,\star^{-1}_\r\left(\star_\r L^\tau(\eta_3)\wedge A^\omega(\eta_4)\right)
     \end{eqnarray*}
and
\begin{eqnarray*}
        \widehat{S}^{\omega\,\star}_B(L^\sigma)(\eta_4)&=&(-1)^{k+1}\star^{-1}_\r(\widehat{S}^{\omega}_B(\star_\r L^\sigma)(\eta_4) )=0.
     \end{eqnarray*} 

\noindent Thus, by equation (\ref{hodgeformula3})
\begin{eqnarray*}
    \langle L^\tau \mid \widehat{S}^{\omega\,\star}_B(L^\sigma)\rangle^\r_{\mathrm{Hom}}&=&\sum^4_{j=1} \langle L^\tau(\beta_j)\mid \widehat{S}^{\omega\,\star}_B(L^\sigma)(\beta_j)\rangle_\r
    \\
    &=&
    (-1)^{n+1}\,\,{i\,3\,q_4\over 2}\langle L^\tau(\beta_1)\mid \star^{-1}_\r\left(\star_\r L^\sigma(\eta_1)\wedge A^\omega(\eta_4)\right)\rangle_\r
    \\
    &+&
    (-1)^{n+1}\,\,{i\,q_4\over 2}\langle L^\tau(\beta_2)\mid \star^{-1}_\r\left(\star_\r L^\sigma(\eta_2)\wedge A^\omega(\eta_4)\right)\rangle_\r
    \\
    &+&
    (-1)^{n+1}\,i\,q_4\,\langle L^\tau(\beta_2)\mid \star^{-1}_\r\left(\star_\r L^\sigma(\eta_3)\wedge A^\omega(\eta_4)\right)\rangle_\r
    \\
    &=&
    (-1)^{k}\,{i\,3\,q_4\over 2}\langle L^\tau(\beta_1)\wedge A^\omega(\eta_4)^\ast\mid \star^{-1}_\r\left(\star_\r L^\sigma(\eta_1)\right)\rangle_\r
    \\
    &+&
    (-1)^{k}\,{i\,q_4\over 2}\langle L^\tau(\beta_2)\wedge A^\omega(\eta_4)^\ast\mid \star^{-1}_\r\left(\star_\r L^\sigma(\eta_2)\right)\rangle_\r
    \\
    &+&
    (-1)^{k}\,i\,q_4\,\langle L^\tau(\beta_3)\wedge A^\omega(\eta_4)^\ast\mid \star^{-1}_\r\left(\star_\r L^\sigma(\eta_3)\right)\rangle_\r
    \\
    &=&
    \langle (-1)^{k+1}\,{i\,3\,q_4\over 2} L^\tau(\beta_1)\wedge A^\omega(\eta_4)\mid  L^\sigma(\eta_1) \rangle_\r
    \\
    &+&
    \langle(-1)^{k+1}\,{i\,q_4\over 2} L^\tau(\beta_2)\wedge A^\omega(\eta_4)\mid   L^\sigma(\eta_2) \rangle_\r
    \\
    &+&
    \langle (-1)^{k+1}\,i\,q_4 L^\tau(\beta_3)\wedge A^\omega(\eta_4)\mid   L^\sigma(\eta_3) \rangle_\r
    \\
    &=&
    \sum^4_{j=1} \langle \widehat{S}^{\omega}_B(L^\tau)(\beta_j)\mid L^\sigma(\beta_j)\rangle_\r
    \\
    &=&
    \langle \widehat{S}^{\omega}(L^\tau)\mid L^\sigma\rangle^\r_{\mathrm{Hom}}
\end{eqnarray*}
\end{proof}

Since $\widehat{D}^\omega=\wedge\circ D^\omega\circ \wedge$, by equation (\ref{ec.3.62}), it follows that
$$\widehat{D}^\omega(\tau)=(\widehat{D}^\omega_B(L^\tau)\otimes \id_{\SU(2)})\qquad \mbox{ with }\qquad \widehat{D}^\omega_B=\wedge\circ D^\omega_B\circ \wedge $$ and hence,  the operator of equation (\ref{ec.4.10.1}) in terms of the space $\mathrm{Hom}_\C(\mathfrak{su}(2,\C)^\#_{4D},\Omega^\bullet(B))^\dagger$ is given by (see equation (123) of reference \cite{sald2})
\begin{equation}
    \label{ec.4.23}
   L^\tau\;\;\longmapsto \;\; \widehat{D}^{\omega\star}_B(L^\tau):=(-1)^{k+1}\, \star^{-1}_\r\circ\, \widehat{D}^\omega_B\circ \star_\r.
\end{equation}

In this way, we define 
\begin{equation}
    \label{used5}
\widehat{DS}^\omega_B:=\widehat{D}^\omega_B-\widehat{S}^\omega_B:\mathrm{Hom}_\C(\mathfrak{su}(2,\C)^\#_{4D},\Omega^\bullet(B))\longrightarrow \mathrm{Hom}_\C(\mathfrak{su}(2,\C)^\#_{4D},\Omega^\bullet(B))
\end{equation}
and by Theorem \ref{the4.1} and Corollary 5.8 of reference \cite{sald2}, we get that in terms of the maps $A^\omega$, $F^\omega$, the equation 
\begin{equation}
    \label{ec.4.6.1.1}
    (d^{\widehat{\nabla}^\omega_{\mathfrak{qa}^\#}\star}-d^{\widehat{S}^\omega\star})(R^\omega)=0.
\end{equation}
turns into 
 \begin{equation}
     \label{used6}
     0=\widehat{DS}^{\omega\star}_B(F^\omega)
 \end{equation}
 with 
 \begin{equation}
     \label{used9}
     \widehat{DS}^{\omega\star}_B:= \widehat{D}^{\omega\star}_B-\widehat{S}^{\omega\star}_B=(-1)^{k+1}\star^{-1}_\r\circ \,\widehat{DS}^\omega_B\circ \star_\r. 
 \end{equation}
 However, recalling that $DS^\omega=\widehat{DS}^\omega$ in elements of $\Mor(\ad,\Delta_\Hor)^\dagger$ (\cite{sald2}), it follows that 
 \begin{equation}
     \label{used10}
     DS^\omega_B=\widehat{DS}^\omega_B
 \end{equation}
  in elements of $\mathrm{Hom}_\C(\mathfrak{su}(2,\C)^\#_{4D},\Omega^\bullet(B))^\dagger$. In addition, it is clear that $$\star_\r (L^\tau)\;\in \;\mathrm{Hom}_\C(\mathfrak{su}(2,\C)^\#_{4D},\Omega^\bullet(B))^\dagger$$ for every $L^\tau$ $\in$ $\mathrm{Hom}_\C(\mathfrak{su}(2,\C)^\#_{4D},\Omega^\bullet(B))^\dagger$. So, equation (\ref{used6}) can also be written by (see equation (\ref{twisted.13}))
 \begin{equation}
     \label{used7}
     0=DS^{\omega\star}_B(F^\omega) \qquad \mbox{ with }\qquad DS^{\omega\star}_B:=(-1)^{k+1}\star^{-1}_\r\circ \,DS^\omega_B\circ \star_\r.
 \end{equation}
Explicitly, equation (\ref{used6}) (and equation (\ref{used7})) implies that
\begin{equation}
    \label{ec.yang-mills}
    0=d\star_\r F^\omega-\langle A^\omega,\star_\r F^\omega\rangle+\langle \star_\r F^\omega,A^\omega\rangle.
\end{equation}
This is the (right and left) non--commutative geometrical Yang--Mills equation. At first glance, equation~(\ref{ec.yang-mills}) is not the \emph{classical} Yang--Mills equation of the electroweak theory. 
\begin{Proposition}
    \label{prop4.7}
     For the basis $\widetilde{\beta}^\#_{4D}$, the non--commutative geometrical Yang--Mills equation is given by $$d\star_\r F^\omega(\beta_1)-q_w\,(\star_\r F^\omega(\beta_2)\wedge A^\omega(\beta_3)-\star_\r F^\omega(\beta_3)\wedge A^\omega(\beta_2))=0,$$ $$d\star_\r F^\omega(\beta_2)-q_w\,(\star_\r F^\omega(\beta_3)\wedge A^\omega(\beta_1)-\star_\r F^\omega(\beta_1)\wedge A^\omega(\beta_3))=0,$$ $$d\star_\r F^\omega(\beta_3)-q_w\,(\star_\r F^\omega(\beta_1)\wedge A^\omega(\beta_2)-\star_\r F^\omega(\beta_2)\wedge A^\omega(\beta_1))=0,$$ $$d \star_\r F^\omega(\beta_4)=0.$$ Here, $\wedge$ denotes the usual wedge product of $\Omega^\bullet(B)$.
\end{Proposition}

\begin{proof}
    The statement easily follows from Proposition~\ref{prop2.3.2} together with the fact that $\Omega^\bullet(B)$ is graded--commutative. For example,
    \begin{eqnarray*}
        0&=&(d\star_\r F^\omega-\langle A^\omega,\star_\r F^\omega\rangle+\langle \star_\r F^\omega,A^\omega\rangle)(\beta_1)
        \\
        &=& 
        d\star_\r F^\omega(\beta_1)-\langle A^\omega,\star_\r F^\omega\rangle(\beta_1)+\langle \star_\r F^\omega,A^\omega\rangle(\beta_1)
        \\
        &=& 
        d\star_\r F^\omega(\beta_1)
        \\
        &-&
        {i\,q_4\over 2}\,(A^\omega(\beta_1)\wedge \star_\r F^\omega(\beta_4)+A^\omega(\beta_4)\wedge \star_\r F^\omega(\beta_1))
        \\
        &+&
        {q_w\over 2}\,(A^\omega(\beta_2)\wedge \star_\r F^\omega(\beta_3)- A^\omega(\beta_3)\wedge\star_\r F^\omega(\beta_2))
        \\
        &+&
        {i\,q_4\over 2}\,(\star_\r F^\omega(\beta_1)\wedge A^\omega(\beta_4)+\star_\r F^\omega(\beta_4)\wedge A^\omega(\beta_1))
        \\
        &-&
        {q_w\over 2}\,(\star_\r F^\omega(\beta_2)\wedge A^\omega(\beta_3)- \star_\r F^\omega(\beta_3)\wedge A^\omega(\beta_2))
        \\
        &=&
        d\star_\r F^\omega(\beta_1)-q_w\,(\star_\r F^\omega(\beta_2)\wedge A^\omega(\beta_3)- \star_\r F^\omega(\beta_3)\wedge A^\omega(\beta_2))
    \end{eqnarray*}
\end{proof}

The next theorem is the purpose of this section

\begin{Theorem}
    \label{4.prop2}
    The equations of Proposition \ref{prop4.7} can be written as (see equation (\ref{ec.3.85})) 
    \begin{equation}
        \label{ec.4.16}
        0=DS_\mu\,(F^{j\mu\nu})=\partial_\mu F^{j\,\mu\nu} + q_w\,\sum^4_{k,l=1} \,f^{jkl}A^k_\mu\,F^{l\,\mu\nu},
    \end{equation}
    where 
    \begin{equation}
    \label{ec.4.18}
    F^{j\,\alpha\beta}= \eta^{\mu\alpha}\eta^{\nu\beta} \,F^j_{\mu\nu},
\end{equation}
with $\eta^{\mu\alpha}=\eta_{\mu\alpha}=\mathrm{diag}(+,-,-,-)$ for $\mu,\nu=0,1,2,3$ and $j=1,2,3,4$.
\end{Theorem}

\begin{proof}
   Let $$A^\omega: \mathfrak{su}(2,\C)^\#_{4D}\longrightarrow \Omega^1(B)$$ be a quantum gauge potential. Then, we know that 
   \begin{equation*}
A^\omega(\beta_j)=A_0\,dx^0+A_1\,dx^1+A_2\,dx^2+A_3\,dx^3+A_4\,dx^4\;\;\longleftrightarrow \;\; A^j_\mu=(A^j_0,A^j_1,A^j_2,A^j_3)
   \end{equation*}
   and 
    \begin{equation*}
    F^\omega(\beta_j)=\sum^3_{\mu, \nu=0} {1\over 2}\,F^j_{\mu\nu}\,dx^\mu\wedge dx^\nu,
\end{equation*}
where  $$F^j_{\mu\nu}=\partial_\mu A^j_\nu-\partial_\nu A^j_\mu+q_w\,\sum^4_{k,l=1} \,f^{jkl}\,A^k_\mu\,A^l_\nu.$$ 

On the other hand, the action of the star Hodge operator in $\Omega^2(B)$ is completely defined by
\begin{equation*}
    \begin{aligned}
    \star_\r (dx^0 \wedge dx^1)=-dx^2\wedge dx^3, \quad \star_\r (dx^0 \wedge dx^2)=dx^1\wedge dx^3, \quad \star_\r (dx^0 \wedge dx^3)=-dx^1\wedge dx^2,\\
    \star_\r (dx^1 \wedge dx^2)=dx^0\wedge dx^3, \quad \star_\r (dx^1 \wedge dx^3)=-dx^0\wedge dx^2, \quad \star_\r (dx^2 \wedge dx^3)=dx^0\wedge dx^1.
    \end{aligned}
\end{equation*}
Thus 
\begin{equation}
    \label{ec.4.19}
    \star_\r F^\omega(\beta_j)=\sum^3_{\mu, \nu=0} {1\over 2}\,F^j_{\mu\nu}\,\star_\r(dx^\mu\wedge dx^\nu).
\end{equation}
A long but simple substitution on the equations of Proposition \ref{prop4.7} shows that
\begin{equation}
    \label{ec.4.19.1}
    \begin{aligned}
        -DS_\mu\,( F^{j\,\mu0})\,dx^1\wedge dx^2\wedge dx^3+DS_\mu\,(  F^{j\,\mu1})\,dx^0\wedge dx^2\wedge dx^3& &\\ -\;DS_\mu\,(F^{j\,\mu2})\,dx^0\wedge dx^1\wedge dx^3-DS_\mu( F^{j\,\mu3})\,dx^0\wedge dx^1\wedge dx^2&=0,
    \end{aligned}
\end{equation}
where (see equation (\ref{ec.3.85}))
\begin{equation*}
    DS_\mu\,(F^{j\,\mu\nu})=\partial_\mu F^{j\,\mu\nu}+q_w\,\sum^4_{k,l=1} \,f^{jkl}\,A^k_\mu\,F^{l\,\mu\nu}
\end{equation*}
for $\mu,\nu=0,1,2,3$ and $j=1,2,3,4$. Therefore
$$DS_\mu\,(F^{j\,\mu\nu})=0.$$ 
\end{proof}

There are five remarks to be made concerning the last result.
\begin{Remark}
    \label{rema4.0}
    By equation (\ref{used10}), we can also write 
    \begin{equation}
    \label{ec.3.85.1}
    \widehat{DS}_\mu\,(F^{j\,\mu\nu}):=\partial_\mu F^{j\,\mu\nu}+q_w\,\sum^4_{k,l=1} \,f^{jkl}\,A^k_\mu\,F^{l\,\mu\nu}
\end{equation}
and hence $$\widehat{DS}_\mu\,(F^{j\,\mu\nu})=0. $$
\end{Remark}

\begin{Remark}
    \label{rema.0.5}
    As proved in Theorem 5.16 of reference \cite{saldcon}, equation (\ref{ec.4.6.1}) characterizes all the critical points of $\qS_{\YM_\r}$ and was obtained by considering the variation $$\omega\;\longmapsto\;\omega+t\lambda,$$ in analogy with the classical case. As we have seen above, the equation (\ref{ec.4.6.1}) turns into $$DS^{\omega\star}_B(F^\omega)=0.$$

    Therefore, it is incorrect to assume that the Yang--Mills equation in the non--commutative geometric setting is
\begin{equation}
\label{falseym}
D^{\omega\star}_B F^\omega = 0,
\end{equation}
as in the classical case. Indeed, not every critical point $\omega$ satisfies equation~(\ref{falseym}); rather, equation~(\ref{ec.4.6.1}) is the equation satisfied by every critical point $\omega$ (\cite{saldcon}). This is important since in the author's opinion, most of the researchers of the area still believe that equation (\ref{falseym}) is the correct one in non--commutative geometry, but this is not the case.
\end{Remark}

\begin{Remark}
    \label{rema4.1}
    It is worth emphasizing that, when expressed in terms of the maps $A^\omega$ and $F^\omega$, the non--commutative geometrical Yang--Mills equation coincides with the  classical Yang--Mills equation of the electroweak interaction  \cite{notation}. Theorem~\ref{4.prop2} shows that, for the quantum model of the electroweak theory that we aim to construct, the dynamical equations of free--interaction gauge boson fields are the correct ones. 
\end{Remark}
 
\begin{Remark}
    \label{rema4.2}
As in the previous subsection,   the most striking part is the fact that the action of $d^{\widehat{\nabla}^\omega\star}-d^{\widehat{S}^\omega\star}$ on $R^\omega$ in $\zeta_{4D}$ in terms of the maps $A^\omega$, $F^{\omega}$, 
coincides with the classical Yang--Mills equation of the electroweak theory.

Additionally, it follows from the classical case that the equation 
\begin{equation}
    \label{ec.4.26}
    \widehat{DS}_\mu\,(\widehat{DS}_\mu\,(F^{j\,\mu\nu}))=0
\end{equation}
is satisfied and it provides the corresponding conservation laws.
\end{Remark}

\begin{Remark}
    \label{rema4.3} 
    Let us consider the quantum gauge group of the non--commutative geometrical Yang--Mills functional, denoted by $\qGG_{\YM_\r}\subseteq \qGG$. By Theorem \ref{prop3.6} we know that $\mathfrak{GG}\subseteq \qGG$ and by Proposition \ref{propgauge}, the action of $\mathfrak{qGG}$ on $\mathfrak{qpc}(\zeta_{4D})$ coincides with the action of $\mathfrak{GG}$ on the space of classical principal connection of $$\pi_\class:\R^4\times G\longrightarrow \R^4,\qquad G=SU(2)\times U(1).$$ In addition, by equation  (\ref{ec.4.21.2}), we know that the non--commutative geometrical Yang--Mills functional agrees with its classical counterpart. Finally, since every element of $\mathfrak{GG}$ is a symmetry of the classical Yang--Mills functional, it follows that
\begin{equation}
    \label{ec.4.29}
    \mathfrak{GG}\subseteq \qGG_{\YM_\r}.
\end{equation}
In other words, every element of $\mathfrak{GG}$ (via the embedding $\iota$ of Theorem \ref{prop3.6}) is a symmetry of the non--commutative geometrical Yang--Mills functional.

Of course, the inclusion in equation~(\ref{ec.4.29}) may be proper. However, if there exists an element $\F\in\qGG$ such that $\F\in\qGG_{\YM_\r}$ but $\F\notin\mathfrak{GG}$, then $\F$ does not arise from a $G=SU(2)\times U(1)$ symmetry. That is, $\F$ does not correspond to a principal bundle automorphism of $\pi_{\class}:\mathbb{R}^4\times G\longrightarrow \mathbb{R}^4$ and is therefore irrelevant in the context of differential geometry and hence, irrelevant under the usual physical interpretation.
\end{Remark}

\begin{Remark}
    \label{rema4.5}
    There is a technical point that needs to be addressed. In physics, all fields are required to vanish at infinity in order to ensure finite energy. Mathematically, this requirement translates into the fact that all relevant maps must be $L^2$--integrable. In this section and in the previous one, we have omitted this detail and worked on the full space $\Omega^\bullet(B)$ in order to avoid introducing additional notation.

In physics, for instance in electromagnetic theory, it is common to use non--$L^2$ functions as toy models for pedagogical purposes. Moreover, in a restricted region of space, such as a laboratory, it is also standard practice to employ non--$L^2$ functions to model physical fields. For example, one often assumes that the electromagnetic field is constant in all the space (although, in reality, it is only constant in a region in the laboratory) and carries out all calculations under this assumption. Thus, non--$L^2$ functions can indeed be useful in appropriate contexts. However, in a physically meaningful theory, all fields are required to be $L^2$--integrable.

In our framework, this issue can be easily addressed by restricting the image of $A^\omega$ (who is a linear map with finite--dimensional domain) to the $L^2$--integrable elements of $\Omega^\bullet(B)$ and mutatis mutandis, all the results of the theory remain valid under this restriction. In other words, it suffices to require that $A^\omega(\beta_j)$ be a $\C$--valued differential $1$--form vanishing at infinity and the condition stated in this remark is satisfied, exactly as in differential geometry.
\end{Remark} 

Summarizing the results of Section 3 and the present section, we conclude that the quantum principal $\SU(2)$--bundle $\zeta_{4D}$ of equation (\ref{ec.3.50}) together with the differential calculus of equation (\ref{ec.3.51}), accurately reproduces all the physics of the pure gauge boson sector of the electroweak interaction.

\section{Yang--Mills--Higgs Theory}

In the \emph{classical} case, the Higgs field is a real scalar field whose interaction with the gauge boson fields $A^j_\mu$ gives them their mass and breaks the symmetry. This field is a section of the associated vector bundle of $\pi_\class:\R^4\times G\longrightarrow \R^4$ for a certain representation of $G$.

\subsection{General Theory}

We will start this section introducing the basics of the theory of associated quantum vector bundles (abbreviated associated qvb) in Durdevich's formulation (\cite{sald1,sald2}). In this section, we shall restrict our attention exclusively to qpb’s for which the quantum base space $B$ is a $C^\ast$--algebra or can be completed into a one. As the reader can verify in \cite{sald1,saldfun}, under this hypothesis it can be proven that every qpb is strong in the sense of Section 5.4.2 of reference \cite{libro}. In other words, under this hypothesis, every quantum principal connection of the bundle is strong (\cite{libro}).     

Let $$\zeta=(P,B,\Delta_P)$$ be a quantum principal $A$--bundle with a differential calculus. Also, let $$\delta^V:V\longrightarrow V\otimes A$$ be a (unitary) finite--dimensional corepresentation of $A$. Then, according to Section 3 of reference \cite{sald1}, the space 
\begin{equation}
    \label{ec.5.1}
    \Mor(\delta^V,\Delta_P):=\{T:V \longrightarrow  P\mid T \mbox{ is linear and }
    (T\otimes \id_A)\circ \delta^V=\Delta_P \circ T\}.
\end{equation}
is a finitely generated projective right $B$--module. Hence, in accordance with the Serre--Swan theorem, we will refer to 
\begin{equation}
    \label{ec.5.2.1}
    E^V_\r:=\Mor(\delta^V,\Delta_P)
\end{equation}
as the associated right quantum vector bundle of $\zeta$ with respect to $\delta^V$ (abbreviated associated right qvb). 

On the other hand, there exists a right $\Omega^\bullet(B)$--module isomorphism (see Section 3 of reference \cite{sald1})
\begin{equation}
    \label{ec.5.3}
 \widehat{\Upsilon}_{V}: \Mor(\delta^V,\Delta_\Hor) \longrightarrow E^V_\r \otimes_B \Omega^\bullet(B).
\end{equation}

The space $$\Mor(\delta^V,\Delta_\Hor):=\{\tau:V \longrightarrow  \Hor^\bullet\,P\mid \tau \mbox{ is linear and }
    (\tau\otimes \id_A)\circ \delta^V=\Delta_\Hor \circ \tau\}$$  will be interpreted as basic quantum differential forms of type $\delta^V$ of $P$, and the space $$E^V_\r \otimes_B \Omega^\bullet(B)$$ will be interpreted as right qvb--valued differential forms of $B$.

Let $\omega$ be a qpc. According to \cite{stheve}, the dual covariant derivative satisfies \begin{equation*}
    \begin{aligned}
        \widehat{D}^\omega  \,\in\, \Mor(\Delta_\Hor, \Delta_\Hor):=\{f:\; &\Hor^\bullet\,P\longrightarrow \Hor^\bullet\,P\\
        &\mid f \mbox{ is linear and }
    (f\otimes \id_A)\circ \Delta_\Hor=\Delta_\Hor \circ f \}.
    \end{aligned}
\end{equation*}
and hence, we have
\begin{equation}
    \label{ec.5.5.1}
\widehat{D}^\omega:\Mor(\delta^V,\Delta_\Hor)\longrightarrow \Mor(\delta^V,\Delta_\Hor), \qquad \tau\longmapsto \widehat{D}^\omega(\tau),
\end{equation}
where $\widehat{D}^\omega(\tau)(v)=\widehat{D}^\omega(\tau(v))=(D^\omega(\tau(v)^\ast))^\ast$ for all $v$ $\in$ $V$.

In this way, we define the gauge quantum linear connection on $E^V_\r$ as 
\begin{equation}
    \label{ec.5.6.1}
    \widehat{\nabla}^\omega_V: E^V_\r\longrightarrow  E^V_\r\otimes_B\Omega^1(B) , \qquad T\longmapsto \widehat{\Upsilon}_V(\widehat{D}^\omega(T)).
\end{equation}
It is clear that $\widehat{\nabla}^\omega_V$ is linear and by equation (\ref{2.f32.1}), it satisfies the right Leibniz rule.

Furthermore, we can extend $\widehat{\nabla}^\omega_V$ to 
\begin{equation}
    \label{ec.5.7.1}
    d^{\widehat{\nabla}^\omega_V}: E^V_\r\otimes_B\Omega^1(B) \longrightarrow E^V_\r\otimes_B\Omega^1(B) 
\end{equation}
  by $$d^{\widehat{\nabla}^\omega_V}(T\otimes_B \mu)=\widehat{\nabla}^\omega_V(T)\,\mu +T\otimes_B d\mu\, $$ for all $\mu$ $\in$ $\Omega^k(B)$, and according to Section 3 of reference \cite{sald1}, we have
\begin{equation}
    \label{ec.5.8.2}
    d^{\widehat{\nabla}^\omega_V}=\widehat{\Upsilon}_V\circ \widehat{D}^\omega\circ \widehat{\Upsilon}^{-1}_V.
\end{equation}

According to Section 3.2 of reference \cite{sald1}, the map  $$(-,-)_\r:E^V_\r\times E^V_\r\longrightarrow B, \qquad (T_1,T_2)\longmapsto \sum^n_{k=1}T_1(v_k)^\ast\,T(v_k),$$
with $\{ v_k\}^m_{k=1}$ an orthonormal basis of $V$, is a $B$--valued inner product. Therefore, there is non--degenerated sesquilinear form
\begin{equation}
    \label{ec.5.8.9.1}
    \langle-|-\rangle^\bullet_\r:(E^V_\r\otimes_B \Omega^\bullet(B) )\times (E^V_\r\otimes_B \Omega^\bullet(B) )\longrightarrow \C,
\end{equation}
given by $$\langle T_1\otimes_B \mu_1\mid T_2\otimes_B \mu_2\rangle^\bullet_\r=\int_B \langle\mu_1,(T_1,T_2)_\r\,\mu_2\rangle_\r\,\dvol.$$

As in the previous section, we will follow the theory presented in reference \cite{sald2} for the non--commutative geometrical Yang--Mills--Higgs theory. However, in order to perfectly coupled the \emph{quantum} case with the \emph{classical} case, we will change a little the theory of \cite{sald2} for the potential $\mathcal{V}$.

\begin{Definition}
\label{6.1.12} 
Consider the map $${\mathcal V}:B\longrightarrow B,\qquad t\longmapsto -a_1\,b+a_2\,b^2$$ for $a_1,a_2\geq 0$ and $\delta^V$ be a (unitary) finite--dimensional $A$--corepresentation. We define the right non--commutative geometrical Yang--Mills scalar matter action as the association \begin{equation*}
\qS_{\YMSM_\r}: \mathfrak{qpc}(\zeta) \times E^V_\r \longrightarrow \R
\end{equation*}
given by $$\qS_{\YMSM_\r}(\omega,T)=\qS_{\YM_\r}(\omega)+\qS_{\SM_\r}(\omega,T),$$ where
\begin{equation*}
\qS_{\SM_\r}(\omega,T)=|| \widehat{\nabla}^{\omega}_{V}T||^{\bullet \,2}_\r-\int_B{\mathcal V}_\r(T)\,\dvol,
\end{equation*}
with ${\mathcal V}_\r(T):={\mathcal V}\circ (T,T)_\r$. 
\end{Definition}

In total analogy with the \emph{classical} case, we get

\begin{Definition}
\label{6.1.14}
A  stationary point of $\qS_{\YMSM_\r}$ is a pair $(\omega,T)$ $\in$ $\mathfrak{qpc}(\zeta)\times E^V_\r$ such that for any $(\lambda, U)$ $\in$ $\overrightarrow{\mathfrak{qpc}(\zeta)} \times E^V_\r$, we have  $$\left.\dfrac{d}{d t}\right|_{t=0}\qS_{\YMSM_\r}(\omega+t\,\lambda ,T)=\left.\dfrac{\partial}{\partial z}\right|_{z=0}\qS_{\YMSM_\r}(\omega,T+z\,U)=0.$$ Stationary points are also called right non--commutative geometrical Yang--Mills Scalar Matter fields and in terms of a traditional physical interpretation, they can be interpreted as {\it scalar matter (or antimatter fields) coupled to gauge boson fields in the presence of the potential ${\mathcal V}$}.
\end{Definition}

As in the previous section, $\qS_\YMSM(\omega+t\,\lambda,T_1,T_2)$ depends polynomially on $t$, and in the same way, $\qS_\YMSM(\omega,T_1+z\,U_1,T_2+z\,U_2)$ depends polynomially on $z$. Hence, we interpret both derivatives of $\qS_\YMSM$ as formal derivatives in the obvious way.

Consider 
\begin{equation}
    \label{ec.5.8.1}
    \mathcal{V}':B\longrightarrow B,\qquad b\longmapsto -a_1\,\mathbbm{1}+2\,a_2\,b
\end{equation}
the formal derivative of $\mathcal{V}$. Then, \emph{mutatis mutandis}, in accordance with Theorem 4.13 of reference \cite{sald2},  $(\omega,T)$ is a critical point of $\qS_{\YMSM_\r}$ if and only if the following equations hold
\begin{equation}
    \label{ec.5.9.2}
\mathrm{Re}\left(\langle(d^{\widehat{\nabla}^{\omega}_{\mathfrak{qa}^\#}\star}-d^{\widehat{S}^{\omega}\star})R^{\omega}\mid  \widehat{\Upsilon}_{\mathfrak{qa}^\#}(\lambda))\rangle^\bullet_\r\right)=\mathrm{Re}\left(\langle \widehat{\Upsilon}_{V}(\widehat{K}^{\lambda}(T))\,|\,\widehat{\nabla}^{\omega}_{V}T\rangle^\bullet_\r\right)
\end{equation}
for all $\lambda$ $\in$ $\mathfrak{qpc}(\zeta)$, and
\begin{equation}
    \label{ec.5.10.2}
\widehat{\nabla}^{\omega\,\star}_{V}\left(\widehat{\nabla}^{\omega}_{V}\,T\right)-{\mathcal V}'_\r(T)\,T=0.
\end{equation}
Here, 
    \begin{equation}
    \label{ec.5.11.1.2}
          \widehat{K}^{\lambda}: \Mor(\delta^V,\Delta_\Hor)\longrightarrow  \Mor(\delta^V,\Delta_\Hor), \qquad
          \tau  \longmapsto \widehat{K}^{\lambda}(\tau)
    \end{equation}
    is given by $$\widehat{K}^{\lambda}(\tau)(v)= (K^\lambda(\tau(v)^\ast))^\ast,\qquad K^\lambda(\tau(v))=-(-1)^k (\tau^{(0)}(v))\,\lambda(\pi(\tau^{(1)}(v))) $$ for all $v$ $\in$ $V$, if $\Im(\tau)\subseteq \Hor^k P$ and $\Delta_\Hor(\tau(v))=\tau^{(0)}(v)\otimes \tau^{(1)}(v)$. Moreover, 
    \begin{equation}
        \label{ec.5.11.2}
        \widehat{\nabla}^{\omega\,\star}_{V}=-(\id \otimes \star^{-1}_\r)\circ d^{\nabla^\omega_V}\circ (\id\otimes \star_\r):E^V_\r\otimes_B \Omega^1(B) \longrightarrow E^V_\r 
    \end{equation}
    is the formal adjoint operator of $\widehat{\nabla}^\omega_V$ with respect to $\langle-|-\rangle^\bullet_\r$. Equations (\ref{ec.5.9.2}), (\ref{ec.5.10.2}) are called the right non--commutative geometrical Yang--Mills scalar matter equations.

Recalling that $\mathfrak{qpc}(\zeta)$ is an affine space modeled by $\overrightarrow{\mathfrak{qpc}(\zeta)}$, for every $\omega$ $\in$ $\mathfrak{qpc}(\zeta)$ and every $\lambda$ $\in$ $\overrightarrow{\mathfrak{qpc}(\zeta)}$, we obtain $\omega+\lambda$ $\in$ $\mathfrak{qpc}(\zeta)$ and hence
$$ \widehat{D}^{\omega+\lambda}=\widehat{D}^\omega+\widehat{K}^\lambda;$$ which implies

\begin{equation}
    \label{2.f30.1.2}  \widehat{\nabla}^{\omega+\lambda}_V=\widehat{\nabla}^\omega_V+\widehat{\Upsilon}_V(\widehat{K}^\lambda).
\end{equation}
Finally, we have (\cite{sald2})
\begin{Definition}
    \label{defgauge2right}
    We define the quantum gauge group of the right non--commutative geometrical Yang-Mills--Matter model as the group $$\qGG_{\YMSM_\r}:=\{\F \in \qGG\mid \qS_{\YMSM_\r}(\omega,T)=\qS_{\YMSM_\r}(\F^\circledast \omega,\widehat{\mathbf{A}}_\F(T))\;\; \mbox{ for all } $$ $$(\omega,T)\, \in \, \mathfrak{qpc}(\zeta)\times E^V_\r \} \subseteq \qGG,$$ where 
    \begin{equation}
        \label{something}
        \widehat{\mathbf{A}}_\F: E^{V}_\r\longrightarrow E^{V}_\r,\qquad \widehat{\mathbf{A}}_\F(T)(v)=\F(T(v)^\ast)^\ast.
    \end{equation}
\end{Definition}

\subsection{The Model}

Consider the quantum principal $\SU(2)$--bundle $$\zeta_{4D}=(P:=B\otimes \SU(2),B,\Delta_P:=\id_B\otimes \Delta)$$ defined in equation (\ref{ec.3.50}), equipped with the differential calculus given in equation (\ref{ec.3.51}).

Consider the Hilbert space $$(W:=\C^2,\langle-|-\rangle_W),$$ with $\langle-|-\rangle_W$ the canonical inner product of $W$ (antilinear in the first coordinate). The standard $\SU(2)$--representation $\delta_W$ on $W$, i.e., the map
\begin{equation}
    \label{ec.5.26}
    \delta_W: SU(2)\times W\longrightarrow W, \qquad (A,\bar{w})\longmapsto A\cdot \bar{w}
\end{equation}
 induces by means of the pull--back, the (right) $\SU(2)$--corepresentation
\begin{equation}
    \label{ec.5.32}
    \delta^W:W\longrightarrow W\otimes \SU(2)
\end{equation}
given by $$\delta^W(\bar{e}_1)=\bar{e}_1\otimes \alpha+\bar{e}_2\otimes \gamma, \qquad  \delta^W(\bar{e}_2)=-\bar{e}_1\otimes \gamma^\ast +\bar{e}_2\otimes \alpha^\ast.$$ 
Here, $$\beta_W:=\{\bar{e}_1=(1,0), \,\bar{e}_2=(0,1)\} $$ is the canonical linear basis of $W$ and notice that the inner product $\langle-|-\rangle_W$ makes $\delta^W$ unitary, since $\delta_W$ is unitary.

According to Section 5 of reference \cite{sald1}, the right associated qvb (see equation (\ref{ec.5.2.1}))
\begin{equation}
    \label{ec.5.32.1}
    E^W_\r:=\Mor(\delta^W,\Delta_P)=\{\Phi:W \longrightarrow  P\mid \Phi \mbox{ is linear and }
    (\Phi\otimes \id_{\SU(2)})\circ \delta^W=\Delta_P \circ \Phi\}
\end{equation}
is trivial, i.e., it is a free finite--dimensional right $B$--module. In addition, by defining
\begin{equation}
    \label{ec.5.32.2}
    g_{11}:=\alpha,\qquad g_{21}:=\gamma,\qquad g_{12}:=-\gamma^\ast,\qquad g_{22}:=\alpha^\ast,
\end{equation}
the linear maps 
\begin{equation}
    \label{ec.5.32.3}
    \Phi^H_i:W\longrightarrow P\qquad \mbox{ such that }\qquad \Phi^H_i(\bar{e}_j)=\mathbbm{1}\otimes g_{ij}
\end{equation}
for $i=1,2$, forms a right $B$--basis of $E^W_\r$, as it is proven in Proposition 5.1 of reference \cite{sald1}. For now on, elements of $E^W_\r$ will be called non--commutative geometrical Higgs fields. Furthermore,   we will refer to the right non--commutative geometrical Yang--Mills scalar matter action of $\zeta_{4D}$ for $\delta^W$ (see Definition \ref{6.1.12}) as the  non--commutative geometrical Yang--Mills-Higgs action.

Thus, every non--commutative geometrical Higgs field $\Phi$ $\in$ $E^W_\r$ is of the form
\begin{equation}
    \label{ec.5.41}
    \Phi=\sum^{2}_{k=1} \phi_k \,\Phi^H_k\quad \mbox{ with }\quad \phi_1,\,\phi_2 \;\in\; B=C^\infty_\C(\R^4).
\end{equation}
Furthermore, in light of Section 3 of reference \cite{sald1}, 
the right $\Omega^\bullet(B)$--module isomorphism  of equation (\ref{ec.5.3}) for $\delta^W$ is given by
\begin{equation}
    \label{ec.5.42}
    \widehat{\Upsilon}_W(\Lambda)= \sum^2_{k=1} \Phi^H_k\otimes_B \xi^\Lambda_k \qquad \mbox{ with }\qquad \xi^\Lambda_k:=\sum^2_{i=1} \Lambda(\bar{e}_i)(\mathbbm{1}\otimes g^\ast_{ki})\,\in\,\Omega^\bullet(B)
\end{equation}
for all $$\Lambda\, \in\, \Mor(\delta^W,\Delta_\H):=\{\Lambda: W\longrightarrow \Hor^\bullet\,P\mid \Lambda \mbox{ is linear and }
    (\Lambda\otimes \id_{\SU(2)})\circ \delta^W=\Delta_\Hor \circ \Lambda \}.$$

\begin{Theorem}
    \label{teorema1}
    The gauge qlc $\widehat{\nabla}^\omega_W$ is given by (see equation (\ref{ec.5.6.1})) $$\widehat{\nabla}^\omega_W(\Phi)=\widehat{\Upsilon}_W(\widehat{D}^\omega(\Phi))=\sum^2_{k=1} \Phi^H_k\otimes_B \xi^{\widehat{D}^\omega\Phi}_k$$ for all $\Phi$ $\in$ $E^W_\r$, where
    \begin{equation}
    \label{ec.5.32.5}
    \begin{aligned}
        \xi^{\widehat{D}^\omega\Phi}_1&=
    d\phi_1-{i\,q_w\over 2}A^\omega(\beta_1)\,\phi_2-{q_w\over 2}A^\omega(\beta_2)\,\phi_2-{i\,q_w\over 2}A^\omega(\beta_3)\,\phi_1-{i\,q_4\over 2}A^\omega(\beta_4)\,\phi_1,
    \\
    \xi^{\widehat{D}^\omega\Phi}_2&=
    d\phi_2 -{i\,q_w\over 2}A^\omega(\beta_1)\,\phi_1+{q_w\over 2}A^\omega(\beta_2)\,\phi_1+{i\,q_w\over 2}A^\omega(\beta_3)\,\phi_2-{i\,q_4\over 2}A^\omega(\beta_4)\,\phi_2. 
    \end{aligned}
\end{equation}
\end{Theorem}

\begin{proof}
Let $\omega$ be a qpc and consider its dual covariant derivative $\widehat{D}^\omega$ (see equation (\ref{ec.3.8.1}). Since  $dS(g)=-\pi(g^{(1)})\,S(g^{(2)})$ for all $g$ $\in$ $\G$ (\cite{stheve}), and $S\circ S=\id_{\SU(2)}$, it follows that
$$dg=-\pi(S^{-1}(g^{(2)}))\,g^{(1)}$$ and
we obtain 
\begin{eqnarray*}
    \widehat{D}^{\omega}(\phi\otimes g)&=&d(\phi\otimes g)+\omega(\pi(S^{-1}(g^{(2)})))\,(\phi\otimes g^{(1)})
    \\
    &=&
    d\phi\otimes g+\phi\otimes dg+[(A^\omega\otimes \id_{\SU(2)})\ad(\pi(S^{-1}(g^{(2)})))
    \\
    &+&
    \mathbbm{1}\otimes \pi(S^{-1}(g^{(2)}))]\,(\phi\otimes g^{(1)})
    \\
    &=&
    d\phi\otimes g+[(A^\omega\otimes \id_{\SU(2)})\ad(\pi(S^{-1}(g^{(2)})))]\,(\phi\otimes g^{(1)})
    \\
    &=&
    d\phi\otimes g+[((A^\omega\circ \pi)\otimes \id_{\SU(2)})\Ad(S^{-1}(g^{(2)}))]\,(\phi\otimes g^{(1)})
    \\
    &=&
    d\phi\otimes g+[(A^\omega(\pi(S^{-1}(g^{(3)})))\otimes g^{(4)}S^{-1}(g^{(2)})]\,(\phi\otimes g^{(1)})
    \\
    &=&
    d\phi\otimes g+A^{\omega}(S^{-1}(\pi(g^{(1)})))\,\phi\otimes g^{(2)}
\end{eqnarray*}
for all $\phi$ $\in$ $B$.

For a right non--commutative geometrical Higgs field $\Phi$, notice that 
\begin{equation}
    \label{ec.5.need}
\Delta_P(\Phi(\bar{e}_i))=\sum^{2}_{k=1} \phi_k \,\Delta_P(\Phi^H_k(\bar{e}_i))=\sum^{2}_{k=1} \phi_k \otimes\,\Delta(g_{ki})=\sum^{2}_{k,l=1} (\phi_k \otimes\,g_{kl})\otimes g_{li};
\end{equation}
hence 
\begin{eqnarray*}
\widehat{D}^\omega(\Phi(\bar{e}_1))=\sum^2_{k=1}\widehat{D}^\omega(\phi_k\,\Phi^H_k(\bar{e}_1))&=&\sum^2_{k=1}\widehat{D}^\omega(\phi_k\,\Phi^H_k(\bar{e}_1))
\\
&=&\sum^2_{k=1}\widehat{D}^\omega(\phi_k\otimes g_{k1})
\\
&=&
\sum^2_{k=1} d\phi_k\otimes g_{k1}+A^{\omega}(\pi(S^{-1}(g^{(1)}_{k1})))\,\phi_k\otimes g^{(2)}_{k1}
\\
&=&
\sum^2_{k=1} d\phi_k\otimes g_{k1}+\sum^2_{k,l=1}A^{\omega}(\pi(S^{-1}(g_{kl})))\,\phi_k\otimes g_{l1}
\\
&=&
\sum^2_{k=1} [d\phi_k+\sum^2_{l=1}A^{\omega}(\pi(S^{-1}(g_{lk})))\,\phi_l]\otimes g_{k1}
\\
&=&
[d\phi_1+A^{\omega}(\pi(S^{-1}(g_{11})))\,\phi_1+A^{\omega}(\pi(S^{-1}(g_{21})))\,\phi_2]\otimes g_{11}
\\
&+&
[d\phi_2+A^{\omega}(\pi(S^{-1}(g_{12})))\,\phi_1+A^{\omega}(\pi(S^{-1}(g_{22})))\,\phi_2]\otimes g_{21}.
\end{eqnarray*}
Similarly, we have
\begin{eqnarray*}
    \widehat{D}^\omega(\Phi(\bar{e}_2))&=&\sum^2_{k=1} [d\phi_k+\sum^2_{l=1}A^{\omega}(\pi(S^{-1}(g_{lk})))\,\phi_l]\otimes g_{k2}
    \\
    &=&
    [d\phi_1+A^{\omega}(\pi(S^{-1}(g_{11})))\,\phi_1+A^{\omega}(\pi(S^{-1}(g_{21})))\,\phi_2]\otimes g_{12}
    \\
    &+&
    [d\phi_2+A^{\omega}(\pi(S^{-1}(g_{12})))\,\phi_1+A^{\omega}(\pi(S^{-1}(g_{22})))\,\phi_2]\otimes g_{22}.
\end{eqnarray*}
Thus,  by equation (\ref{ec.5.42}) we get
\begin{eqnarray*}
\xi^{\widehat{D}^\omega\Phi}_1=\sum^2_{k=1}\widehat{D}^\omega(\Phi(\bar{e}_k))(\mathbbm{1}\otimes g^\ast_{1k})=
    d\phi_1+A^{\omega}(\pi(S^{-1}(g_{11})))\,\phi_1+A^{\omega}(\pi(S^{-1}(g_{21})))\,\phi_2\;\in\; \Omega^1(B)
\end{eqnarray*}
and
$$\xi^{\widehat{D}^\omega\Phi}_2=\sum^2_{k=1}\widehat{D}^\omega(\Phi(\bar{e}_k))(\mathbbm{1}\otimes g^\ast_{2k})=d\phi_2+A^{\omega}(\pi(S^{-1}(g_{12})))\,\phi_1+A^{\omega}(\pi(S^{-1}(g_{22})))\,\phi_2 \;\in\; \Omega^1(B).$$ 
According to Section 2.1, we have
\begin{equation*}
    \begin{aligned}
        \pi(S^{-1}(g_{11}))&=\pi(S^{-1}(\alpha))=\pi(\alpha^\ast)=-{i\,q_w\over 2}\beta_3-{i\,q_4\over 2}\beta_4,
    \\
    \pi(S^{-1}(g_{21})) &= \pi(S^{-1}(\gamma))=-\pi(\gamma)=-{i\,q_w\over 2}\beta_1-{q_w\over 2}\beta_2,
    \\
   \pi(S^{-1}(g_{12})) &= -\pi(S^{-1}(\gamma^\ast))=\pi(\gamma^\ast)=-{i\,q_w\over 2}\beta_1+{q_w\over 2}\beta_2,
   \\
   \pi(S^{-1}(g_{22}))&=\pi(S^{-1}(\alpha^\ast))=\pi(\alpha)={i\,q_w\over 2}\beta_3-{i\,q_4\over 2}\beta_4 
    \end{aligned}
\end{equation*}
and therefore
\begin{equation*}
    \begin{aligned}
        \xi^{\widehat{D}^\omega\Phi}_1&=
    d\phi_1-{i\,q_w\over 2}A^\omega(\beta_1)\,\phi_2-{q_w\over 2}A^\omega(\beta_2)\,\phi_2-{i\,q_w\over 2}A^\omega(\beta_3)\,\phi_1-{i\,q_4\over 2}A^\omega(\beta_4)\,\phi_1,
    \\
    \xi^{\widehat{D}^\omega\Phi}_2&=
    d\phi_2 -{i\,q_w\over 2}A^\omega(\beta_1)\,\phi_1+{q_w\over 2}A^\omega(\beta_2)\,\phi_1+{i\,q_w\over 2}A^\omega(\beta_3)\,\phi_2-{i\,q_4\over 2}A^\omega(\beta_4)\,\phi_2. 
    \end{aligned}
\end{equation*}
\end{proof}

Let us define
\begin{equation}
    \label{operators}
    T_1:={\sigma_1\over 2}, \qquad T_2:={\sigma_2\over 2},\qquad  T_3:={\sigma_1\over 2},\qquad Y:={\Id_2\over 2}={1\over 2}\,\begin{pmatrix}
1 & 0  \\
0 & 1  \\
\end{pmatrix},
\end{equation}
where $\sigma_1$, $\sigma_2$, $\sigma_3$ are the Pauli matrices. Hence, in index notation, we define
\begin{equation}
    \label{ec.cova}
    \begin{aligned}
        \widehat{D}_\mu\phi_1&:=\partial_\mu\phi_1-{i\,q_w\over 2}\,A^1_\mu\,\phi_2-{q_w\over 2}\,A^2_\mu\,\phi_2-{i\,q_w\over 2}\,A^3_\mu\,\phi_1-{i\,q_4\over 2}\,A^4_\mu\,\phi_1,
        \\
       \widehat{D}_\mu\phi_2 &:=\partial_\mu\phi_2-{i\,q_w\over 2}\,A^1_\mu\,\phi_1+{q_w\over 2}\,A^2_\mu\,\phi_1+{i\,q_w\over 2}\,A^3_\mu\,\phi_2-{i\,q_4\over 2}\,A^4_\mu\,\phi_2
    \end{aligned}
\end{equation}
for $\mu=0,1,2,3$; which can be summarized as
\begin{equation}
    \label{ec.5.29}
    \widehat{D}_\mu \phi:=\partial_\mu \phi- i\,q_w\,A^1_\mu\,T_1\,\phi- i\,q_w\,A^2_\mu\,T_2\,\phi- i\,q_w\,A^3_\mu\,T_3\,\phi-i\,q_4\,A^4_\mu\,Y\,\phi,
\end{equation}
where  
$$\phi:=\begin{pmatrix}
    \phi_1  \\
    \phi_2
\end{pmatrix}. $$

Let us consider the principal bundle $$\pi_\class: \R^4\times G\longrightarrow \R^4, \qquad G=SU(2)\times U(1)$$ and a principal connection $\omega_\class$. Let us define the $G$--representation 
\begin{equation}
    \label{ec.5.13}
    \delta^H_W: SU(2)\times U(1)\times W\longrightarrow W,\qquad (A,\mathrm{e}^{it},\bar{w})\longmapsto A\cdot \begin{pmatrix} 
\mathrm{e}^{it} & 0 \\
0 & \mathrm{e}^{it}   
\end{pmatrix} \cdot \bar{w}.
\end{equation}
 and consider the space of smooth sections $\Gamma(W^H\R^4)$ of the associated vector bundle $$W^H\R^4\cong\R^4\times W$$ of $\pi_\class$ with respect to $\delta^H_W$. Elements of $\Gamma(W^H\R^4)$ are called Higgs fields (\cite{gtvp}), and since 
\begin{equation}
    \label{ec.5.13.1}
    \{s_i: \R^4\longrightarrow  W^H\R^4, \qquad x\longmapsto (x,\bar{e}_i) \}^2_{i=1}
\end{equation}
is a $B$--basis of $\Gamma(W^H\R^4)$, it follows that every Higgs field is of the form 
\begin{equation}
    \label{ec.5.13.2}
   \phi=\sum^2_{i=1}\phi_i\,s_i=\begin{pmatrix}
    \phi_1  \\
    \phi_2
\end{pmatrix}, \qquad \phi_1,\,\phi_2\,\in\,B=C^\infty_\C(\R^4). 
\end{equation}
In this way, we identify a non--commutative geometrical Higgs field $\Phi$ with a \emph{classical} Higgs field $\phi$, i.e., 
\begin{equation}
    \label{ec.5.48}
\Phi=\sum_{i=1}\phi_i\,\Phi^H_i\;\;\longleftrightarrow\;\;\phi=\sum^2_{i=1}\phi_i\,s_i \;\;\longleftrightarrow\;\; \phi=\begin{pmatrix}
    \phi_1  \\
    \phi_2
\end{pmatrix}
\end{equation}
and
\begin{equation}
    \label{ec.5.49}
    \widehat{\nabla}^\omega_W\Phi\;\;\longleftrightarrow \;\; \widehat{D}_\mu\phi,\qquad \xi^{\widehat{D}^\omega\Phi}_1\;\;\longleftrightarrow\;\; \widehat{D}_\mu\phi_1, \qquad \xi^{\widehat{D}^\omega\Phi}_2\;\;\longleftrightarrow\;\; \widehat{D}_\mu\phi_2.
\end{equation}

\begin{Remark}
    \label{remacovaright}
     The most striking feature is that $\widehat{\nabla}^\omega_W\Phi$ coincides exactly with the action of the gauge covariant derivative (also called the gauge linear connection) on the Higgs field $\phi$ in the  electroweak theory (\cite{notation}), including of course, the correct corresponding value of the hypercharge $$ Y={1\over 2}.$$  
\end{Remark}

Furthermore, as in the proof of Proposition A.1 of reference \cite{sald2}, we have
\begin{eqnarray*}
    ||\widehat{\nabla}^\omega_W\Phi||^{\bullet 2}_\r=\sum^2_{k,l=1}\int_B \langle \xi^{\widehat{D}^\omega\Phi}_k \,(\Phi_k,\Phi_l), \xi^{\widehat{D}^\omega\Phi}_j\rangle_\r\,\dvol=
    \sum^2_{k=1}\int_B \langle \xi^{\widehat{D}^\omega\Phi}_k, \xi^{\widehat{D}^\omega\Phi}_k\rangle_\r\,\dvol
\end{eqnarray*}
and since $\langle-,-\rangle_\r$ is the Minkowski metric on $\C$--valued differential forms (antilinear in the first coordinate), we get
\begin{equation}
    \label{ec.5.50}
    ||\widehat{\nabla}^\omega_W\Phi||^{\bullet 2}_\r= \int_{\R^4} (\widehat{D}^\mu\phi)^\dagger\,(\widehat{D}_\mu\phi)\,\dvol,
\end{equation}
where $\dagger$ indicates the complex transpose. In the same way, we have
\begin{equation}
    \label{ec.5.51}
    \int_B \mathcal{V}_\r(\Phi)\,\dvol= \sum^2_{k=1}\int_B \mathcal{V}(\phi^\dagger_k\,\phi_k)\,\dvol=:\int_{\R^4}\mathcal{V}(\phi)\,\dvol.
\end{equation}

\begin{Remark}
    \label{rema2}
    Equations (\ref{ec.5.50}), (\ref{ec.5.51}), together with equation (\ref{ec.4.21.2}) shows that the right non--commutative geometrical Yang--Mills--Higgs action is exactly the classical Yang--Mills--Higgs action (\cite{notation}). 
\end{Remark}

 Since the gauge qlc $\widehat{\nabla}^\omega_W$ coincides with the \emph{classical} gauge linear connection of a Higgs field and the formal adjoint operator $\widehat{\nabla}^{\omega\,\star}_W$ is defined exactly as the formal adjoint operator of the \emph{classical} gauge linear connection (\cite{diff1}), it follows that
\begin{equation}
    \label{ec.5.52}
    \widehat{\nabla}^{\omega\,\star}_W(\widehat{\nabla}^\omega_W \Phi)\;\;\longleftrightarrow \;\; \;\;\widehat{D}^\mu\,\widehat{D}_\mu\phi.
\end{equation}
This shows that 
\begin{equation}
    \label{ec.5.53}
    \widehat{\nabla}^{\omega\,\star}_W(\widehat{\nabla}^\omega_W \Phi)=\mathcal{V}'_\r(\Phi)\Phi \;\; \longleftrightarrow \;\;\widehat{D}^\mu\,\widehat{D}_\mu\phi=\mathcal{V}'(\phi)\phi.
\end{equation}
In other words, equation (\ref{ec.5.10.2}) for $\zeta_{4D}$ and $\delta^W$ correctly reproduces the equation of motion of the \emph{classical} Higgs field in the electroweak theory (\cite{diff1,notation}).

Consider the principal bundle $$\pi_\class:\R^4\times G\longrightarrow \R^4,\qquad G=SU(2)\times U(1)$$ and a principal connection $\omega_\class$. Then, for every basic differential $1$--form $\lambda_\class$ of type $\ad^\class$, we know that  $\omega_\class+\lambda_\class$ is again a principal connection, and the gauge covariant derivative on the associated vector bundle $W^H\R^4\cong \R^4\times W$ of $\pi_\class$ with respect to $\delta^H_W$, satisfies 
\begin{equation}
    \label{clas}
    \nabla^{\omega_\class+\lambda_\class} (\phi) =\nabla^{\omega_\class}(\phi)+\delta^{H'}_W(\lambda_\class)\phi
\end{equation}
for all smooth sections $\phi$ of $W^H\R^4$,  where  $$\delta^{H'}_W(\lambda)(\phi)(X_p)=\left.{d\over dt}\right|_{t=0}\delta^{H}_W(\mathrm{exp}(t\,\lambda_\class(Y_x)))\phi(p)$$
for all $X_p$ $\in$ $T_p \R^4$, with $Y_x$ $\in$ $T_x(\R^4\times G)$ such that $(d\pi_\class)_x(Y_x)=X_p$ \cite{gtvp}. According to Theorem \ref{teorema1}, we have $$\widehat{\nabla}^{\omega+\lambda}_W(\Phi)\;\;\longleftrightarrow\;\; \nabla^{\omega_\class+\lambda_\class} (\phi),\qquad   \widehat{\nabla}^{\omega}_W(\Phi)\;\;\longleftrightarrow\;\; \nabla^{\omega_\class} (\phi)$$ and hence, by equations (\ref{2.f30.1.2}), (\ref{clas}), we obtain  
\begin{equation}
    \label{ec.k1}
\widehat{\Upsilon}_V(\widehat{K}^\lambda(\Phi))\;\;\longleftrightarrow \;\;\delta^{H'}_W(\lambda_\class)\phi.
\end{equation}
In this way, according to Section 7.2 of reference \cite{gtvp}, 
Theorems \ref{4.prop2}, \ref{teorema1} and equation (\ref{ec.k1}), the equation (\ref{ec.5.9.2}) coincides with its \emph{classical} counterpart and therefore, in index notation, it turns into 
\begin{equation}
    \label{ec.5.46}
    \widehat{DS}_\mu\,(F^{j\,\mu\nu})={\mathcal{J}}^{j\,\nu}.
\end{equation}
where ${\mathcal{J}}^{j\,\nu}$ $\in$ $B$, for $\mu,\nu=0,1,2,3$ and $j=1,2,3,4$ \cite{gtvp,notation}. The tensor ${\mathcal{J}}^{j\,\nu}$ is called \emph{the current} and in accordance with reference \cite{notation}, it is given by
\begin{equation}
    \label{ec.5.47}
    {\mathcal{J}}^{j\,\nu}=-{i\,q_w\over 2}\,(\phi^\dagger\,\sigma^j\,(\widehat{D}^\mu\phi)-(\widehat{D}^\mu\phi)^\dagger\,\sigma^j\,\phi )
\end{equation}
for $j=1,2,3$, where $\sigma^1$, $\sigma^2$, $\sigma^3$ are the Pauli matrices and $$\widehat{D}^\mu=\eta^{\mu\nu}\,\widehat{D}_\nu$$ with $\eta^{\mu\nu}=\mathrm{diag}(+,-,-,-)$. In addition (\cite{notation})
\begin{equation}
    \label{ec.5.48.1}
    {\mathcal{J}}^{4\,\nu}=-{i\,q_4\over 2}\,(\phi^\dagger\,(\widehat{D}^\mu\phi)-(\widehat{D}^\mu\phi)^\dagger\,\phi).
\end{equation}

\begin{Remark}
    \label{rema4.3.5}
    Equations (\ref{ec.5.53}), (\ref{ec.5.46}) show that right non--commutative geometrical Yang--Mills--Higgs fields are actually the classical Yang--Mills--Higgs fields.
\end{Remark}

In the \emph{classical} case, the gauge group $\mathfrak{GG}$ of the principal bundle $$\pi_\class: \R^4\times G\longrightarrow \R^4,\qquad G=SU(2)\times U(1) $$ acts on right of $\Gamma(W^H\,\R^4)$ by means of 
\begin{equation}
    \label{gaugeactionsections1}
    {\bf{A}}_{F_\class} : \Gamma(W^H\R^4)\longrightarrow \Gamma(W^H\R^4),\quad \phi=\sum^2_{i=1}\phi_i\,s_i\longmapsto {\bf{A}}_{F_\class}(\phi)=\sum^2_{i=1}\phi_i\,A_{F_\class}(s_i),
\end{equation}
where $${\bf{A}}_{F_\class}(s_i):\R^4\longrightarrow  W^H\R^4,\qquad x\longmapsto (x,\overline{e}_i\cdot \psi(x))$$ and $\phi_1$, $\phi_2$ $\in$ $B=C^\infty_\C(\R^4)$, for a gauge transformation $$F_\class: \R^4\times G\longrightarrow \R^4\times G,\qquad (x,C)\longmapsto (x,\psi(x)\,C).$$ In a more \emph{traditional} notation, we get 
\begin{equation}
    \label{gaugeactionsections2}
    (\phi,\psi)\longmapsto \phi\cdot \psi \qquad \mbox{ with }\qquad \phi=\begin{pmatrix}
    \phi_1  \\
    \phi_2
\end{pmatrix}.
\end{equation}

\begin{Proposition}
    \label{gauge2}
    Consider the map $\F_\psi$ of equation (\ref{ec.3.100}). Then, the action of $\F_\psi$ on $E^W_\r$ by means of equation (\ref{something}) coincides with  ${\bf{A}}_{F_\class}$ by identifying $s_i$ (a $B$--basis of $\Gamma(W^H\R^4)$) with $\Phi^H_i$ (a $B$--basis of $E^H_\r$), for $i=1,2$.
\end{Proposition}

\begin{proof}
    Consider (see equations (\ref{ec.5.13.1}), (\ref{ec.5.13.1})) $$\phi=\sum^2_{i=1}\phi_i\,s_i=\begin{pmatrix}
    \phi_1  \\
    \phi_2
\end{pmatrix} \,\in\, \Gamma(W^H\R^4), \qquad \phi_1,\,\phi_2\,\in\,B=C^\infty_\C(\R^4).$$  Denoting $$\psi(x)=\begin{pmatrix}
    \psi_{11}(x) & \psi_{12}(x)  \\
    \psi_{21}(x) & \psi_{22}(x)
\end{pmatrix} \; \in\; SU(2)\times U(1),\qquad x\in \R^4,$$ we get 
\begin{equation*}
    \begin{aligned}
        {\bf{A}}_{F_\class}(s_1)(x)&=\psi_{11}(x)\,s_1(x)+\psi_{12}(x)\,s_2(x)\\
        &= \f_\psi(g_{11}\otimes z)\,s_1(x)+\f_\psi(g_{12}\otimes z)\,s_2(x)
    \end{aligned}
\end{equation*}
and 
\begin{equation*}
    \begin{aligned}
        {\bf{A}}_{F_\class}(s_2)(x)&=\psi_{21}(x)\,s_1(x)+\psi_{22}(x)\,s_2(x)\\
        &= \f_\psi(g_{21}\otimes z)\,s_1(x)+\f_\psi(g_{22}\otimes z)\,s_2(x),
    \end{aligned}
\end{equation*}
recalling that $\f_\psi$ denotes the pull--back of $\psi$. Hence $${\bf{A}}_{F_\class}(\phi)=  (\phi_1\, \f_\psi(g_{11}\otimes z)+\phi_2\, \f_\psi(g_{21}\otimes z))\,s_1 + (\phi_1\, \f_\psi(g_{12}\otimes z)+\phi_2\, \f_\psi(g_{22}\otimes z))\,s_2.$$

On the other hand, by the definition of $\widehat{\f}_\psi$, we obtain
\begin{eqnarray*}
    \widehat{{\bf A}}_{\F_\psi}(\Phi^H_1)(\overline{e}_i)=\F_\psi(\Phi^{H\ast}_1(\overline{e}_i))^\ast=\F_\psi(\mathbbm{1}\otimes g^\ast_{1i})^\ast&=&\sum^2_{k=1} \widehat{\f}_\psi(g^\ast_{1k})^\ast\otimes g_{ki}\\
        &=&\sum^2_{k=1} \f_\psi(g^\ast_{1k}\otimes z^\ast)^\ast\otimes g_{ki}\\
        &=& \sum^2_{k=1} \f_\psi(g_{1k}\otimes z)\otimes g_{ki}\\
        &=& \f_\psi(g_{11}\otimes z)\otimes g_{1i}+\f_\psi(g_{12}\otimes z)\otimes g_{2i}
        \\
        &=&
        \f_\psi(g_{11}\otimes z)\, \Phi^H_1(\overline{e}_i)+\f_\psi(g_{12}\otimes z)\, \Phi^H_2(\overline{e}_i)
\end{eqnarray*}
and similarly, we have $$\widehat{{\bf A}}_{\F_\psi}(\Phi^H_2)(\overline{e}_i)=\f_\psi(g_{21}\otimes z)\, \Phi^H_1(\overline{e}_i)+\f_\psi(g_{22}\otimes z)\, \Phi^H_2(\overline{e}_i).$$ Therefore $$\widehat{{\bf A}}_{\F_\psi}(\Phi)=(\phi_1\,\f_\psi(g_{11}\otimes z)+\phi_2 \,\f_\psi(g_{21}\otimes z))\, \Phi^H_1 + (\phi_1\,\f_\psi(g_{12}\otimes z)+\phi_2 \,\f_\psi(g_{22}\otimes z))\, \Phi^H_2$$ and the proposition follows by identifying $s_i$ with $\Phi^H_i$.
\end{proof}

\begin{Remark}
    \label{rematubote}
    By the previous proposition, the fact that the non--commutative geometrical Yang--Mills--Higgs functional is equal to its \emph{classical} counterpart, the fact that (see Remark \ref{rema4.3}) $$\mathfrak{GG} \subseteq \qGG_{\YM_\r}$$ and Definition \ref{defgauge2right}, we conclude that 
    \begin{eqnarray}
        \label{rematubote2}
        \mathfrak{GG} \subseteq \qGG_{\YMSM_\r}.
    \end{eqnarray}
    As in the previous section,  the inclusion in equation~(\ref{rematubote2}) may be proper. However, if there exists an element $\F\in\qGG$ such that $\F\in\qGG_{\YMSM_\r}$ but $\F\notin\mathfrak{GG}$, then $\F$ does not arise from a $G=SU(2)\times U(1)$ symmetry. That is, $\F$ does not correspond to a principal bundle automorphism of $\pi_{\class}:\mathbb{R}^4\times G\longrightarrow \mathbb{R}^4$ and is therefore irrelevant in the context of  differential geometry and hence, irrelevant under the usual physical interpretation.
\end{Remark}

In summary, Remarks \ref{rema2}, \ref{rema4.3.5},  \ref{rematubote} and equation (\ref{ec.5.46}) show that the non--commutative geometrical Yang--Mills--Higgs theory for right structures of $\zeta_{4D}$ coincides with the \emph{classical} Yang--Mills--Higgs theory of the electroweak interaction.

\subsection{The Higgs Mechanism}

Since all the operators, symmetries and actions coincide with those of the \emph{classical} case, it follows that the Higgs mechanism holds in our model. However, given its importance, we will nevertheless develop this part of the theory.

\begin{Definition}
    \label{def6.1}
    We have
    \begin{enumerate}
        \item A vector $\bar{w}_0$ $\in$ $W:=\C^2$ is called a vacuum vector if it is a non--zero minimum of $$\mathcal{V}: W\longrightarrow \R,\qquad \bar{w}\longmapsto \mathcal{V}(\bar{w}):=-a_1\,||\bar{w}||^2+a_2\,||\bar{w}||^4.$$ The set of vacuum vectors is called the vacuum manifold of $W$.
        \item A vacuum gauge is an element $\Phi_0$ $\in$ $E^W $ such that $$\Phi_0=w_1\,\Phi_1+w_2\,\Phi_2$$ with $\bar{w}_0=(w_1,w_2)$ $\in$ $W=\C^2$ a vacuum vector. Here, $\{\Phi_k \}^2_{k=1}$ is the right $B$--basis of $E^W $ of equation (\ref{ec.5.32.3}).
        \item Given vacuum gauge $\Phi_0$, the pair $$(\omega^\triv,\Phi_0)$$ is called a \emph{vacuum configuration}. 
    \end{enumerate}
\end{Definition}
Since $A^{\omega^\triv}=0$, a vacuum configuration is a  non--commutative geometrical Yang--Mills--Higgs field, i.e., a vacuum configuration is a classical  Yang--Mills--Higgs fields.

We assume that, moments after the Big Bang, the symmetry was spontaneously broken; that is, the  non--commutative geometrical Higgs field $$\Phi_0=w_1\,\Phi_1+w_2\,\Phi_2$$ with $\bar{w}=(w_1,w_2)$ transitions spontaneously from the unstable value\footnote{Unstable with respect to the potential $\mathcal{V}$.} $$\bar{w}=0$$ to the vacuum vector $$\bar{w}_0\neq 0.$$  

It is easy to check that a vector $\bar{w}_0$ $\in$ $W$ is a vacuum vector if and only if (\cite{diff1})
\begin{equation}
    \label{ec.6.1}
    ||\bar{w}_0||=\sqrt{a_1\over 2\,a_2 }. 
\end{equation}
It follows that the vacuum manifold is a $3$--sphere in $W=\C^2$, centered at $0$ and with radius $\displaystyle \sqrt{a_1\over 2\,a_2 }$. The \emph{traditional} choice is  (\cite{diff1})
\begin{equation}
    \label{ec.6.2}
    \bar{w}_0:= \begin{pmatrix}
    0  \\
    \sqrt{a_1\over 2\,a_2 }
\end{pmatrix}={1\over \sqrt{2}}\begin{pmatrix}
    0  \\
    \v
\end{pmatrix} \quad \mbox{ with }\quad \v=\sqrt{2}\,||\bar{w}_0||=\sqrt{a_1\over a_2}.
\end{equation}
Hence 
\begin{equation}
    \label{ec.6.3}
   (\omega^\triv,\Phi_0=0\,\Phi_1+ {1\over \sqrt{2}}\,\v\,\Phi_2)
\end{equation}
is a vacuum configuration. By taking a \emph{small} variation around $\Phi_0$, in a unitary gauge (\cite{diff1}), we get
\begin{equation}
    \label{ec.6.4}
    (\omega,\Phi=\Phi_0+\widehat{\Phi}), \qquad \widehat{\Phi}:=0\,\Phi_1+ {1\over \sqrt{2}}\,  \h (x)\,\Phi_2 \quad \mbox{ with }\quad  \h(x)\,\in\,B=C^\infty_\C(\R^4).
\end{equation}
Under the identification
\begin{equation}
    \label{higgsfield}
    \Phi\;\;\longleftrightarrow\;\;\phi:={1\over \sqrt{2}}\begin{pmatrix}
    0  \\
    \v+ \h
\end{pmatrix}
\end{equation}
we have 
\begin{equation*}
\begin{aligned}
    \displaystyle\widehat{D}_\mu\phi&=\begin{pmatrix}
    -{i\,q_w\over 2}\,A^1_\mu \,{1\over \sqrt{2}}\,(\v+ \h)-{q_w\over 2}\,A^2_\mu\, {1\over \sqrt{2}}\,(\v+ \h)  \\\\
    \partial_\mu\, {1\over \sqrt{2}}\, \h+{i\,q_w\over 2}\,A^3_\mu\,{1\over \sqrt{2}}\,(\v+ \h)-{i\,q_4\over 2}\,A^4_\mu\,{1\over \sqrt{2}}\,(\v+ \h)
\end{pmatrix}\\\\
&=\begin{pmatrix}
    -{i\,q_w\over 2}\,(A^1_\mu-i\,A^2_\mu)\,{1\over \sqrt{2}}\,(\v+ \h) \\\\
    \partial_\mu \,{1\over \sqrt{2}}\, \h+{i\over 2}\,(q_w\,A^3_\mu-q_4\,A^4_\mu)\,{1\over \sqrt{2}}\,(\v+ \h)
\end{pmatrix}.
\end{aligned}
\end{equation*}
Consider now the linear basis  of $\mathfrak{su}(2,\C)^\#_{4D}$ given by
\begin{equation}
    \label{ec.6.5}
    \alpha^\#_{4D}:=\{\alpha_1:={\beta_1-i\beta_2\over \sqrt{2}}, \quad \alpha_2:={\beta_1+i\beta_2\over \sqrt{2}},\quad \alpha_3:={q_3\,\beta_3-q_4\,\beta_4\over \sqrt{q^2_w+q^2_4}}, \quad \alpha_4:={q_4\,\beta_3+q_w\,\beta_4\over \sqrt{q^2_w+q^2_4}}\}. 
\end{equation}
It is worth mentioning that $$\mathrm{span}_\C\{\alpha_1,\,\alpha_2 \}=\mathrm{span}_\C\{\beta_1,\,\beta_2 \},\qquad \mathrm{span}_\C\{\alpha_3,\,\alpha_4 \}=\mathrm{span}_\C\{\beta_3,\,\beta_4 \}.$$ The linear independent set $\{\beta_1,\beta_2\}$ transforms into $\{\alpha_1,\alpha_2\}$ by means of $${1\over \sqrt{2}}\begin{pmatrix}
1 & -i  \\
1 & i
\end{pmatrix}\begin{pmatrix}
\beta_1  \\
\beta_2 
\end{pmatrix}=\begin{pmatrix}
\alpha_1  \\
\alpha_2 
\end{pmatrix};$$ while $\{\beta_3,\beta_4\}$ transforms into $\{\alpha_3,\alpha_4\}$ by means of the rotation by 
\begin{equation}
    \label{ec.5.44}
    \theta_W:=\mathrm{tan}^{-1}\left({q_4\over q_w}\right), \qquad \theta_W\,\in\,[0,{\pi\over 2}).
\end{equation}
In other words  
\begin{equation}
    \label{ec.5.45}
    \begin{pmatrix}
\mathrm{cos}(\theta_W) & -\mathrm{sin}(\theta_W) \\
\mathrm{sin}(\theta_W) & \mathrm{cos}(\theta_W)
\end{pmatrix}\begin{pmatrix}
\beta_3  \\
\beta_4 
\end{pmatrix}=\begin{pmatrix}
\alpha_3  \\
\alpha_4 
\end{pmatrix}.
\end{equation}
Therefore, the linear basis $\alpha^\#_{4D}$ is orthonormal with respect to the inner product $\langle-|-\rangle_{4D}$ (see Section 2). We have the following identifications of the quantum gauge potential $A^\omega$ 
\begin{equation}
    \label{ec.6.6}
    \begin{aligned}
        A^\omega(\alpha_1)&={A^\omega(\beta_1)-iA^2(\beta_2)\over \sqrt{2}}\quad\longleftrightarrow \quad \W^+_\mu:={A^1_\mu-iA^2_\mu\over \sqrt{2}},\\
        A^\omega(\alpha_2)&={A^\omega(\beta_1)+iA^2(\beta_2)\over \sqrt{2}}\quad\longleftrightarrow \quad \W^-_\mu:={A^1_\mu+iA^2_\mu\over \sqrt{2}},
        \\
        A^\omega(\alpha_3)&={q_3\,A^\omega(\beta_3)-q_4\,A^\omega(\beta_4)\over \sqrt{q^2_w+q^2_4}}\quad \longleftrightarrow \quad \Z^{\circ}_\mu:={ q_w\,A^3_\mu-q_4\,A^4_\mu\over \sqrt{q^2_w+q^2_4}},
        \\
        A^\omega(\alpha_4)&={q_4\,A^\omega(\beta_3)+q_3\,A^\omega(\beta_4)\over \sqrt{q^2_w+q^2_4}}\quad \longleftrightarrow \quad \A_\mu:={ q_4\,A^3_\mu+q_w\,A^4_\mu\over \sqrt{q^2_w+q^2_4}}.
    \end{aligned}
\end{equation}
for $\mu=0,1,2,3$, and under this change of basis, we obtain
\begin{equation*}
\widehat{D}_\mu\phi= \begin{pmatrix}
    -{i\,q_w\over \sqrt{2}}\,\W^+_\mu\,{1\over \sqrt{2}}\,(\v+ \h) \\\\
    \partial_\mu {1\over \sqrt{2}}\, \h+{i\over 2}\,\sqrt{q^2_w+q^2_4}\,\,\Z^\circ_\mu \,{1\over \sqrt{2}}\,(\v+ \h)
\end{pmatrix}.
\end{equation*}
It is worth mentioning that there is no interaction between $\A_4$ and the field $\phi$. In this way, we get
\begin{equation*}
    \begin{aligned}
        (\widehat{D}^\mu\phi)^\dagger (\displaystyle\widehat{D}_\mu\phi)&=\left| \begin{pmatrix}
    -{i\,q_w\over \sqrt{2}}\,\W^+_\mu\,{1\over \sqrt{2}}\,(\v+\widehat{\h}) \\\\
    \partial_\mu {1\over \sqrt{2}}\, \h+{i\over 2}\,\sqrt{q^2_w+q^2_4}\,\,\Z^\circ_\mu \,{1\over \sqrt{2}}\,(\v+ \h)
\end{pmatrix} \right|^2\\
&={1\over 2}\,(\partial^\mu\, \h)^\dagger\,(\partial_\mu\, \h)+ {1\over 4}\,\left(q^2_w\,(\W^+_\mu)^\dagger\,\W^{+\mu}+{q^2_w+q^2_4\over 2}\,(\Z^\circ_\mu)^\dagger\,\Z^{\circ \mu} \right)\,(\v+ \h)^2
\\
&= {1\over 2}\,(\partial^\mu\, \h)^\dagger\,(\partial_\mu\, \h)+ {q^2_w\over 8}\,((\W^+_\mu)^\dagger\,\W^{+\mu}+(\W^-_\mu)^\dagger\,\W^{-\mu})\,(\v+ \h)^2 
\\
&+{q^2_w+q^2_4\over 8}\,(\Z^\circ_\mu)^\dagger\,\Z^{\circ \mu} \,(\v+ \h)^2,
    \end{aligned}
\end{equation*}
where, as usual, $\dagger$ denotes the complex transpose. Since $$\mathcal{V}(\phi)=-a_1\,\phi^\dagger\,\phi+a_2\,\phi^\dagger\,\phi=-a_1\,\left({\v+ \h\over \sqrt{2}} \right)^2+a_2\,\left({\v+ \h\over \sqrt{2}} \right)^4=a_1\, \h^2+\cdots $$
the action $\qS_{\SM }$ (see Definition \ref{6.1.12}) is given by
\begin{equation}
    \label{ec.6.7}
    \begin{aligned}
        \int_{\R^4}\, &{1\over 2}\,[(\partial^\mu\, \h)^\dagger\,(\partial_\mu\, \h)-\m^2_{H}\, \h^2]
        \\
        &+ {\m^2_+\over 2}\,(\W^+_\mu)^\dagger\,\W^{+\mu}+{\m^2_-\over 2}\,(\W^-_\mu)^\dagger\,\W^{-\mu}+{\m^2_z\over 2}\,(\Z^\circ_\mu)^\dagger\,\Z^{\circ \mu}+{\m^2_4\over 2}\,(\A_\mu)^\dagger\,\A^{\mu}
        \\
        &+ \cdots,
    \end{aligned}
\end{equation}
where 
\begin{equation}
    \label{ec.6.8}
    \m_{H}:=\sqrt{2\,a_1},\qquad \m_+:=\m_-:={q_w\,\v\over 2}, \qquad \m_z:={(q^2_w+q^2_4)\,\v\over 2}, \qquad \m_4=0. 
\end{equation}

On the other hand, as in Proposition \ref{prop3.4.1}, for the quantum field strength $F^\omega$ we have
\begin{equation}
    \label{ec.6.9}
    \begin{aligned}
F^\omega(\alpha_1)\;\;\longleftrightarrow\;\;\mathbbm{F}^+_{\mu\nu},\quad &F^\omega(\alpha_2)\;\;\longleftrightarrow\;\;\mathbbm{F}^-_{\mu\nu},\quad F^\omega(\alpha_3)\;\;\longleftrightarrow\;\;\mathbbm{F}^\circ_{\mu\nu}\\ &F^\omega(\alpha_4)\;\;\longleftrightarrow\;\;\mathbbm{F}^4_{\mu\nu},
    \end{aligned}
\end{equation}
with 
\begin{equation}
    \label{ec.6.10}
    \mathbbm{F}^+_{\mu\nu}=\partial_\mu \W^+_\nu-\partial_\nu \W^+_\mu-\,{i\,q^2_w\over\sqrt{q^2_w+q^2_4} }\, (\W^+_\mu\,\Z^\circ_\nu-\W^+_\nu\,\Z^\circ_\mu)-\,{i\,q_w\,q_4\over\sqrt{q^2_w+q^2_4} }(\W^+_\mu\,\A_\nu-\W^+_\nu\,\A_\mu),
\end{equation}
\begin{equation}
    \label{ec.6.11}
    \mathbbm{F}^-_{\mu\nu}=\partial_\mu \W^-_\nu-\partial_\nu \W^-_\mu+\,{i\,q^2_w\over\sqrt{q^2_w+q^2_4} }\, (\W^-_\mu\,\Z^\circ_\mu-\W^-_\nu\,\Z^\circ_\mu)+\,{i\,q_w\,q_4\over\sqrt{q^2_w+q^2_4} }(\W^-_\mu\,\A_\nu-\W^-_\nu\,\A_\mu),
\end{equation}
\begin{equation}
    \label{ec.6.12}
    \mathbbm{F}^\circ_{\mu\nu}=\partial_\mu \Z^\circ_\nu-\partial_\nu \Z^\circ_\mu+\,{i\,q^2_w\over\sqrt{q^2_w+q^2_4} }\, (\W^+_\mu\,\W^-_\nu-\W^+_\nu\,\W^-_\mu),
\end{equation}
\begin{equation}
    \label{ec.6.13}
    \mathbbm{F}^4_{\mu\nu}=\partial_\mu \A_\nu-\partial_\nu \A_\mu+\,{i\,q_w\,q_4\over\sqrt{q^2_w+q^2_4} }\, (\W^+_\mu\,\W^-_\nu-\W^+_\nu\,\W^-_\mu).
\end{equation}
Furthermore, the action $\qS_{\YM }$ of Definition (\ref{4.def1}) is given by (see equation (\ref{ec.4.21.2}))
\begin{equation}
    \label{ec.6.14}
    \begin{aligned}
        -{1\over 4}\int_{\R^4} \mathbbm{F}^+_{\mu\nu}\,\mathbbm{F}^{+\mu\nu}+\mathbbm{F}^-_{\mu\nu}\,\mathbbm{F}^{-\mu\nu}+\mathbbm{F}^\circ_{\mu\nu}\,\mathbbm{F}^{\circ \mu\nu}+\mathbbm{F}^4_{\mu\nu}\,\mathbbm{F}^{4 \mu\nu}.
    \end{aligned}
\end{equation}
This complete the characterization of the  non--commutative geometrical Yang--Mills--Higgs action in the bases $\alpha^\#_{4D}$.  

Now, consider (see equation (\ref{operators}))
\begin{equation}
\label{ec.6.15}
  T^\pm:={T_1\pm i\,T_2\over\sqrt{2}},\qquad Q:=T_3+Y.
\end{equation}
By equation (\ref{ec.5.29}), for a general  non--commutative geometrical field $$\Phi\;\;\longleftrightarrow \;\;\phi=\begin{pmatrix}
  \phi_1 \\
  \phi_2
\end{pmatrix},\qquad \phi_1,\,\phi_2\,\in\, B=C^\infty_\C(\R^4),$$ we get
\begin{equation}
    \label{ec.6.16}
    \begin{aligned}
        \widehat{D}_\mu\phi&=\partial_\mu\phi-i\,q_w\,(\W^+_\mu\,T^++\W^-_\mu\,T^-)\,\phi-i\,\left({q^2_w\,T_3-q^2_4\,Y\over \sqrt{q^2_w+q^2_4}}\right)\,\Z^\circ_\mu\,\phi
        -\,i\,e\,Q\,\A_\mu\,\phi, 
    \end{aligned}
\end{equation}
where
\begin{equation}
    \label{ec.6.17}
    e:={q_w\,q_4\over \sqrt{q^2_w+q^2_4}}.
\end{equation}

\begin{Remark}
    \label{higgs1}
   As in the classical case, we have (\cite{diff1}):
\begin{enumerate}
    \item The field $\h$ $\in$ $B$ is what is usually called the Higgs boson. Its self--interaction gives rise to its mass $\m_H$ (the Higgs mass).
    \item The fields $\W^+_\mu$ and $\W^-_\mu$ are called the $W$--boson fields. Their interaction with $\h$ gives rise to their masses $\m_+$ and $\m_-$, respectively.
    \item The field $\Z^\circ_\mu$ is called the $Z$--boson field. Its interaction with $\h$ gives rise to its mass $\m_z$.
    \item The field $\A_\mu$ is the photon field. It does not interact with $\h$ and is therefore massless ($\m_4=0$).
    \item The eigenvalues $t_3$ of $T_3$ are called the weak isospin, the eigenvalues $y$ of $Y$ are called the hypercharge and the eigenvalues $q$ of $Q$ correspond to the electric charge.
   \item The constant $e$ is known as the electromagnetic coupling constant, or as the elementary electric charge.
   \item For the non--commutative geometrical Higgs field $\Phi$ of equation (\ref{higgsfield}), the weak isospin is $$t_3= -{1\over 2},$$ the hypercharge is $$ y={1\over 2}$$ and hence its electric charge is $$q=0.$$ These are the correct quantum numbers of the Higgs field.
   \item As usual, by taking $$\widetilde{W}^+_\mu:=W^+_\mu T^+,\qquad \widetilde{W}^-_\mu:=W^-_\mu T^-,$$ we find that $$ [T_3,\widetilde{W}^+_\mu]=\widetilde{W}^+_\mu, \qquad [T_3,\widetilde{W}^-_\mu]=-\widetilde{W}^-_\mu.$$ Hence, we can conclude that the gauge boson fields $W^+_\mu$, $W^-_\mu$ have weak isospin $$t_3=\pm 1,$$ respectively. In addition, $$[Y,\widetilde{W}^+_\mu]=[Y,\widetilde{W}^-_\mu]=0$$ and therefore, we can conclude that the fields $W^+_\mu$, $W^-_\mu$ have hypercharge $$y=0,$$ respectively. These imply that  the fields $W^+_\mu$, $W^-_\mu$ have electric charge $$q=0,$$ respectively. These are the correct quantum numbers of the $W$ bosons.
   \item Similarly, the quantum numbers of the gauge boson field $Z^\circ$ are $$t_3=0, \qquad y=0, \qquad q=0,$$ which are the correct ones.
   \item As in the classical case, this breaks the $SU(2)\times U(1)$ symmetry of the model (or, to be more precisely, the non--commutative geometrical $SU(2)$ symmetry) and leaves only a residual $U(1)$ symmetry.
\end{enumerate}
\end{Remark}

In summary, as expected, the Higgs mechanism  for the right structure of $\Mor(\delta^W,\Delta_P)$ coincides with the \emph{classical} Higgs mechanism of the electroweak interaction.

\section{Fermionic Fields}
Unfortunately, at the time of writing, there is, to the best of the author’s knowledge, no general theory for non--commutative geometrical fermionic fields based on dualization via the pull--back of \emph{classical} principal bundles, as is done for gauge boson fields and scalar matter fields in~\cite{sald2}. Therefore, for this part of the work, we adopt an alternative approach drawing on other results in non--commutative geometry.

For the fermionic sector of the electroweak theory in non--commutative geometry, we follow the development of A.~Connes presented in Section~6 of~\cite{con}. In that reference, Connes proves that his formulation of the \emph{non--commutative} Dirac operator coincides with the \emph{classical} Dirac theory, provided that the underlying spaces are \emph{classical} (smooth manifolds) and that the \emph{classical} gauge covariant derivative is used. Moreover, he shows that the fermionic kinetic action 
\begin{equation}
    \label{used1}
    \int_{\R^4} i\,\psi^\dagger\;\overline{\sigma}^\mu\;D_\mu\,\psi\,\dvol, 
\end{equation}
where $\bar{\sigma}^\mu=(\sigma^1,\sigma^2,\sigma^3,\Id)$ can be recovered within his framework. Once the fermionic action and the Higgs action is introduced, we can also introduce the Yukawa action
\begin{equation}
    \label{used2}
   \int_{\R^4} i\,\mathrm{y}'\,\psi^\dagger\,\h\,\psi\,\dvol,
\end{equation}
where $\mathrm{y}'$ is the Yukawa coupling.

Since Connes' theory is  more well--known than Durdevich’s formulation of quantum principal bundles, and to maintain a reasonable length for this paper, we will not present the foundational aspects of this framework here. Nevertheless, the reader can consult \cite{con} for further details.

On the other hand, in reference \cite{saldcon} it is  proven that Durdevich's theory of qpb's can be coupled to A. Connes' theory. In this way, we only need to show that our model reproduces the correct gauge covariant derivatives of fermionic fields to guarantee that the theory reproduces the fermionic sector of the electroweak theory.

\subsection{General Theory}

Let $$\zeta=(P,B,\Delta_P)$$ be a quantum principal $A$--bundle with a differential calculus. Also, consider $$\delta^V:V\longrightarrow V\otimes A$$ be a (unitary) finite--dimensional corepresentation of $A$. Then, according to Section 3 of reference \cite{sald1}, the space 
\begin{equation}
    \label{ec.5.1.1}
    \Mor(\delta^V,\Delta_P):=\{T:V \longrightarrow  P\mid T \mbox{ is linear and }
    (T\otimes \id_A)\circ \delta^V=\Delta_P \circ T\}.
\end{equation}
is a finitely generated projective left  $B$--module. Hence, in accordance with the Serre--Swan theorem, we will refer to
\begin{equation}
    \label{ec.5.2}
    E^V_\l:=\Mor(\delta^V,\Delta_P)
\end{equation}
as the associated left quantum vector bundle of $\zeta$ with respect to $\delta^V$ (abbreviated associated left qvb).

On the other hand, there exists a left $\Omega^\bullet(B)$--module isomorphism (see Section 3 of reference \cite{sald1})
\begin{equation}
    \label{ec.5.3.1}
 \Upsilon_{V}: \Mor(\delta^V,\Delta_\Hor) \longrightarrow \Omega^\bullet(B)\otimes_B E^V_\l.
\end{equation}

As in the previous section, the space $$\Mor(\delta^V,\Delta_\Hor):=\{\tau:V \longrightarrow  \Hor^\bullet\,P\mid \tau \mbox{ is linear and }
    (\tau\otimes \id_A)\circ \delta^V=\Delta_\Hor \circ \tau\}$$  will be interpreted as basic quantum differential forms of type $\delta^V$ of $P$ and the space $$\Omega^\bullet(B)\otimes_B E^V_\l$$ will be interpreted as left  qvb--valued differential forms of $B$.

Let $\omega$ be a qpc. According to \cite{stheve}, the covariant derivative satisfies 
\begin{equation*}
    \begin{aligned}
        D^\omega  \,\in\, \Mor(\Delta_\Hor, \Delta_\Hor):=\{f:\, &\Hor^\bullet\,P\longrightarrow \Hor^\bullet\,P\\
        &\mid f \mbox{ is linear and }
    (f\otimes \id_A)\circ \Delta_\Hor=\Delta_\Hor \circ f \}.
    \end{aligned}
\end{equation*}
and hence, we have
\begin{equation}
    \label{ec.5.5}
    D^\omega:\Mor(\delta^V,\Delta_\Hor)\longrightarrow \Mor(\delta^V,\Delta_\Hor), \qquad \tau\longmapsto D^\omega(\tau)
\end{equation}
where $D^\omega(\tau)(v)=D^\omega(\tau(v))$ for all $v$ $\in$ $V$.

In this way, we define the gauge quantum linear connection (abbreviated guage qlc) on $E^V_\l$ as 
\begin{equation}
    \label{ec.5.6}
    \nabla^\omega_V: E^V_\l\longrightarrow \Omega^1(B)\otimes_B E^V_\l, \qquad T\longmapsto \Upsilon_V(D^\omega(T)).
\end{equation}
It is clear that $\nabla^\omega_V$ is linear and by equation (\ref{2.f32}), it satisfies the left Leibniz rule.

\begin{Remark}
    \label{remacon}
In general, by equation (\ref{2.f32.1}), the map $\widehat{D}^\omega$ does not induce a qlc on $E^V_\l$; while by equation (\ref{2.f32}), the map $D^\omega$ does not induce a qlc on $E^V_\r$. Notice that if $\omega$ is regular, then $D^\omega=\widehat{D}^\omega$ and hence $\widehat{D}^\omega$ induces a qlc on $E^V_\l$ and $D^\omega$ induces a qlc on $E^V_\r$.
\end{Remark}

\subsection{The Model: Left--Handed Lepton Doublets}
Consider the quantum principal $\SU(2)$--bundle $$\zeta_{4D}=(P:=B\otimes \SU(2),B,\Delta_P:=\id_B\otimes \Delta)$$ defined in equation (\ref{ec.3.50}), equipped with the differential calculus given in equation (\ref{ec.3.51}) and consider the $\SU(2)$--corepresentation $\delta^W$ of equation (\ref{ec.5.32}). 

According to Section 5 of reference \cite{sald1}, the left associated qvb (see equation (\ref{ec.5.2}))
\begin{equation}
    \label{5.49.6}
    E^W_\l:=\Mor(\delta^W,\Delta_P)=\{\Psi:W \longrightarrow  P\mid \Psi \mbox{ is linear and }
    (\Psi\otimes \id_{\SU(2)})\circ \delta^W=\Delta_P \circ \Psi\}
\end{equation}
is trivial, i.e., it is a free finite--dimensional left $B$--module. Elements of $E^V_\l$ will be called left non--commutative geometrical fields. In addition, by defining
\begin{equation}
    \label{5.49}
    g_{11}:=\alpha,\qquad g_{21}:=\gamma,\qquad g_{12}:=-\gamma^\ast,\qquad g_{22}:=\alpha^\ast,
\end{equation}
the linear maps 
\begin{equation}
    \label{5.50}
    \Psi^L_i:W\longrightarrow P\qquad \mbox{ such that }\qquad \Psi^L_i(\bar{e}_j)=\mathbbm{1}\otimes g_{ij}
\end{equation}
for $i=1,2$, forms a left $B$--basis of $E^W_\l$ (\cite{sald1}). 
Thus, every non--commutative geometrical field $\Psi$ $\in$ $E^W_\l$ is of the form
\begin{equation}
    \label{ec.5.41.5}
    \Psi=\sum^{2}_{k=1} \psi_k \,\Psi^H_k\quad \mbox{ with }\quad \psi_1,\,\psi_2 \;\in\; B=C^\infty_\C(\R^4).
\end{equation}

In light of Section 3 of reference \cite{sald1}, 
the left $\Omega^\bullet(B)$--module isomorphism  of equation (\ref{ec.5.3.1}) for $\delta^W$ is given by
\begin{equation}
    \label{5.52}
    {\Upsilon}_W(\Lambda)= \sum^2_{k=1}  \xi^\Lambda_k \otimes_B \Psi^L_k \qquad \mbox{ with }\qquad \xi^\Lambda_k:=\sum^2_{i=1} \Lambda(\bar{e}_i)(\mathbbm{1}\otimes g^\ast_{ki})\,\in\,\Omega^\bullet(B)
\end{equation}
for all $$\Lambda\,\in\ \Mor(\delta^W,\Delta_\H):=\{\Lambda: W\longrightarrow \Hor^\bullet\,P\mid \Lambda \mbox{ is linear and }
    (\Lambda\otimes \id_{\SU(2)})\circ \delta^W=\Delta_\Hor \circ \Lambda \}.$$

\begin{Theorem}
    \label{teorema2}
    The gauge qlc ${\nabla}^\omega_W$ is given by  $${\nabla}^\omega_W(\Psi)={\Upsilon}_W({D}^\omega(\Psi))=\sum^2_{k=1}  \xi^{{D}^\omega\Psi}_k\otimes_B \Psi^L_k$$ for all $\Psi$ $\in$ $E^W_\r$, where
    \begin{equation}
    \label{5.53}
    \begin{aligned}
        \xi^{D^\omega\Psi}_1&=d\psi_1-{i\,q_w\over 2}\,A^\omega(\beta_1)\,\psi_2-{q_w\over 2}\,A^\omega(\beta_2)\,\psi_2- {i\,q_w\over 2}\,A^\omega(\beta_3)\,\psi_1+{i\,q_4\over 2}\,A^\omega(\beta_4)\,\psi_1,
    \\
    \xi^{D^\omega\Psi}_2&=
    d\psi_2-{i\,q_w\over 2}\,A^\omega(\beta_1)\,\psi_1+{q_w\over 2}\,A^\omega(\beta_2)\,\psi_1+{i\,q_w\over 2}\,A^\omega(\beta_3)\,\psi_2+{i\,q_4\over 2}\,A^\omega(\beta_4)\,\psi_2. 
    \end{aligned}
\end{equation}
\end{Theorem}

\begin{proof}
    Let $\omega$ be a qpc and consider its covariant derivative $D^\omega$ (see equation (\ref{ec.3.8})). Since $$dg=g^{(1)}\pi(g^{(2)}) $$ for all $g$ $\in$ $\G$ (\cite{stheve}), we obtain 
\begin{eqnarray*}
    D^{\omega}(\psi\otimes g)&=&d(\psi\otimes g)-(\psi\otimes g^{(1)})\,\omega(\pi(g^{(2)}))
    \\
    &=&
    d\psi\otimes g+\psi\otimes dg-(\psi\otimes g^{(1)})\,[(A^\omega\otimes \id_\G)\ad(\pi(g^{(2)}))+\mathbbm{1}\otimes \pi(g^{(2)})]
    \\
    &=&
    d\psi\otimes g-(\psi\otimes g^{(1)})\,[(A^\omega\otimes \id_\G)\ad(\pi(g^{(2)}))]
    \\
    &=&
    d\psi\otimes g-(\psi\otimes g^{(1)})\,[((A^\omega\circ \pi)\otimes \id_\G)\Ad(g^{(2)})]
    \\
    &=&
    d\psi\otimes g-(\psi\otimes g^{(1)})\,[(A^\omega(\pi(g^{(3)}))\otimes S(g^{(2)})g^{(4)}]
    \\
    &=&
    d\psi\otimes g-\psi\,A^{\omega}(\pi(g^{(1)}))\otimes g^{(2)}
\end{eqnarray*}
for all $\psi$ $\in$ $B$. 

By equation (\ref{ec.5.need}), it follows that
\begin{eqnarray*}
D^\omega(\Psi(\bar{e}_1))=\sum^2_{k=1}D^\omega(\psi_k\,\Psi_k(\bar{e}_1))&=&\sum^2_{k=1}D^\omega(\psi_k\,\Psi_k(\bar{e}_1))
\\
&=&\sum^2_{k=1}D^\omega(\psi_k\otimes g_{k1})
\\
&=&
\sum^2_{k=1} d\psi_k\otimes g_{k1}-\psi_k\,A^{\omega}(\pi(g^{(1)}_{k1}))\otimes g^{(2)}_{k1}
\\
&=&
\sum^2_{k=1} d\psi_k\otimes g_{k1}-\sum^2_{k,l=1}\psi_k\,A^{\omega}(\pi(g_{kl}))\otimes g_{l1}
\\
&=&
\sum^2_{k=1} [d\psi_k-\sum^2_{l=1}\psi_l\,A^{\omega}(\pi(g_{lk}))]\otimes g_{k1}
\\
&=&
[d\psi_1-\psi_1\,A^{\omega}(\pi(g_{11}))-\psi_2\,A^{\omega}(\pi(g_{21}))]\otimes g_{11}
\\
&+&
[d\psi_2-\psi_1\,A^{\omega}(\pi(g_{12}))-\psi_2\,A^{\omega}(\pi(g_{22}))]\otimes g_{21}.
\end{eqnarray*}
Similarly, we have
\begin{eqnarray*}
    D^\omega(\Psi(\bar{e}_2))&=&\sum^2_{k=1} [d\psi_k-\sum^2_{l=1}\psi_l\,A^{\omega}(\pi(g_{lk}))]\otimes g_{k2}
    \\
    &=&
    [d\psi_1-\psi_1\,A^{\omega}(\pi(g_{11}))-\psi_2\,A^{\omega}(\pi(g_{21}))]\otimes g_{12}
    \\
    &+&
    [d\psi_2-\psi_1\,A^{\omega}(\pi(g_{12}))-\psi_2\,A^{\omega}(\pi(g_{22}))]\otimes g_{22}.
\end{eqnarray*}
Thus, by equation (\ref{5.52}) we get
$$\xi^{D^\omega\Psi}_1=\sum^2_{k=1}D^\omega(\Psi(\bar{e}_k))(\mathbbm{1}\otimes g^\ast_{1k})=
    d\psi_1-\psi_1\,A^{\omega}(\pi(g_{11}))-\psi_2\,A^{\omega}(\pi(g_{21}))\;\in\; \Omega^1(B), $$
$$\xi^{D^\omega\Psi}_2=\sum^2_{k=1}D^\omega(\Psi(\bar{e}_k))(\mathbbm{1}\otimes g^\ast_{2k})=d\psi_2-\psi_1\,A^{\omega}(\pi(g_{12}))-\psi_2\,A^{\omega}(\pi(g_{22})) \;\in\; \Omega^1(B).$$ 

According to Section 2.1, we obtain
$$\pi(g_{11})=\pi(\alpha)={i\,q_w\over 2}\beta_3-{i\,q_4\over 2}\beta_4,\qquad \pi(g_{22})=\pi(\alpha^\ast)=-{i\,q_w\over 2}\beta_3-{i\,q_4\over 2}\beta_4,$$ $$\pi(g_{21})=\pi(\gamma)={i\,q_w\over 2}\beta_1+{q_w\over 2}\beta_2,\qquad \pi(g_{12})=-\pi(\gamma^\ast)={i\,q_w\over 2}\beta_1-{q_w\over 2}\beta_2$$  
and therefore,
\begin{equation*}
    \xi^{D^\omega\Psi}_1=
    d\psi_1-{i\,q_w\over 2}\,\psi_2\,A^\omega(\beta_1)-{q_w\over 2}\,\psi_2\,A^\omega(\beta_2)- {i\,q_w\over 2}\,\psi_1\,A^\omega(\beta_3)+{i\,q_4\over 2}\,\psi_1\,A^\omega(\beta_4),
\end{equation*}
\begin{equation*}
    \xi^{D^\omega\Psi}_2=
    d\psi_2-{i\,q_w\over 2}\,\psi_1\,A^\omega(\beta_1)+{q_w\over 2}\,\psi_1\,A^\omega(\beta_2)+{i\,q_w\over 2}\,\psi_2\,A^\omega(\beta_3)+{i\,q_4\over 2}\,\psi_2\,A^\omega(\beta_4). 
\end{equation*}
\end{proof}

Let us denote by $$\sigma_W: \Omega^\bullet(B)\otimes_B E^W_\l\longrightarrow E^W_\r\otimes_B \Omega^\bullet(B)$$ the canonical flip map. Then, by comparison Theorem (\ref{teorema1}) with Theorem (\ref{teorema2}), it is clear that 
\begin{equation}
    \label{ec.5.55.5}
    \sigma_W\circ \nabla^\omega_W=\widehat{\nabla}^\omega_W
\end{equation}
if and only if $A^\omega(\beta_4)=0$, i.e., if and only if $\omega$ is regular (see Proposition \ref{regular}). Equation (\ref{ec.5.55.5}) shows that the gauge qlc's $\nabla^\omega_W$, $\widehat{\nabla}^\omega_W$ can be combined into a bimodule qlc on $\Mor(\delta^W,\Delta_P)$ if only if $\omega$ is regular and in any other case, they are not related.

As before, let us define
\begin{equation}
    \label{operators1}
    T_1:={\sigma_1\over 2}, \qquad T_2:={\sigma_2\over 2},\qquad  T_3:={\sigma_1\over 2},\qquad Y:=-{\Id_2\over 2}=-{1\over 2}\,\begin{pmatrix}
1 & 0  \\
0 & 1  \\
\end{pmatrix},
\end{equation}
where $\sigma_1$, $\sigma_2$, $\sigma_3$ are the Pauli matrices. Hence, in index notation, let us define
\begin{equation}
    \label{ec.cova1}
    \begin{aligned}
        D_\mu\psi_1&:=\partial_\mu\psi_1-{i\,q_w\over 2}\,A^1_\mu\,\psi_2-{q_w\over 2}\,A^2_\mu\,\psi_2-{i\,q_w\over 2}\,A^3_\mu\,\psi_1+{i\,q_4\over 2}\,A^4_\mu\,\psi_1,
        \\
       D_\mu\psi_2 &:=\partial_\mu\psi_2-{i\,q_w\over 2}\,A^1_\mu\,\psi_1+{q_w\over 2}\,A^2_\mu\,\psi_1+{i\,q_w\over 2}\,A^3_\mu\,\psi_2+{i\,q_4\over 2}\,A^4_\mu\,\psi_2
    \end{aligned}
\end{equation}
for $\mu=0,1,2,3$. These formulas can be summarized as
\begin{equation}
    \label{ec.5.29.1}
    D_\mu \psi:=\partial_\mu \psi- i\,q_w\,A^1_\mu\,T_1\,\psi- i\,q_w\,A^2_\mu\,T_2\,\psi- i\,q_w\,A^3_\mu\,T_3\,\psi-i\,q_4\,A^4_\mu\,Y\,\psi,
\end{equation}
where 
$$\psi:=\begin{pmatrix}
    \psi_1  \\
    \psi_2
\end{pmatrix}. $$

As before, we define 
\begin{equation}
    \label{ec.electric}
    Q=T_3+Y.
\end{equation}
The eigenvalues $t_3$ of $T_3$ are called the weak isospin, and the eigenvalues $y$ of $Y$ are called the hypercharge. Thus, the eigenvalues $q$ of $Q$ correspond to the electric charge.

\begin{Remark}
    \label{remacovaleft}
    If we identify a left non--commutative geometrical field $\Psi$ with $\psi$, i.e., 
\begin{equation}
    \label{5.54}
    \Psi\;\;\longleftrightarrow \;\; \psi
\end{equation}
we can identify 
\begin{equation}
    \label{5.57}
    \nabla^\omega_W\Psi\;\;\longleftrightarrow \;\; D_\mu\psi,\qquad \xi^{D^\omega\psi}_1\;\;\longleftrightarrow\;\; D_\mu\psi_1, \qquad \xi^{D^\omega\psi}_2\;\;\longleftrightarrow\;\; D_\mu\psi_2.
\end{equation}
The most striking feature is that $\nabla^\omega_W\Psi$ coincides exactly with the gauge covariant derivative (also called the gauge linear connection) on the left--handed lepton doublet fields $$ \psi=\begin{pmatrix}
    \psi_1  \\
    \psi_2
\end{pmatrix}=\left\{\begin{pmatrix}
    \nu_{\mathrm{e}}  \\
    \mathrm{e}
\end{pmatrix}_\l, \; \begin{pmatrix}
    \nu_{\mu}  \\
    \mu
\end{pmatrix}_\l,\; \begin{pmatrix}
    \nu_{\tau}  \\
    \tau
\end{pmatrix}_\l \right\}$$ in the electroweak theory, including of course, the correct corresponding value of the hypercharge $$y=-{1\over 2}.$$ In addition, $$\psi_1\;\;\longleftrightarrow\;\; \begin{pmatrix}
    \psi_1  \\
    0
\end{pmatrix}=\left\{\begin{pmatrix}
    \nu_{\mathrm{e}}  \\
    0
\end{pmatrix}_\l, \; \begin{pmatrix}
    \nu_{\mu}  \\
    0
\end{pmatrix}_\l,\; \begin{pmatrix}
    \nu_{\tau}  \\
    0
\end{pmatrix}_\l \right\} $$ always has weak isospin $$t_3={1\over 2} $$ and hence, its electric charge is $$q=0.$$ Moreover, 
$$\psi_2\;\;\longleftrightarrow\;\; \begin{pmatrix}
    0  \\
    \psi_2
\end{pmatrix}=\left\{\begin{pmatrix}
    0  \\
    \mathrm{e}
\end{pmatrix}_\l, \; \begin{pmatrix}
    0  \\
    \mu
\end{pmatrix}_\l,\; \begin{pmatrix}
    0  \\
    \tau
\end{pmatrix}_\l \right\} $$ always has weak isospin $$t_3=-{1\over 2} $$ and hence, its electric charge is $$q=1.$$ These are the correct quantum numbers for left--handed lepton fields (\cite{diff1}).
\end{Remark}

\subsection{The Model: Right--Handed Lepton Singlets}

Unfortunately, the description of the rest of fermions in not so \emph{natural} as the description of the Higgs field and left--handed lepton doublet fields. The reason is that, in the \emph{classical} case, the rest of fermion fields are described as sections of the associated vector bundle of $$\pi_\class:\R^4\times G\longrightarrow \R^4,\qquad G=SU(2)\times U(1)$$ with respect to the tensor product of a $SU(2)$--representation and a $U(1)$--representation that is not the standard one.

However, in our model, we only have access to the standard representations of $U(1)$. In order to solve this, we will use the fact that the space of qpc's $\mathfrak{qpc}(\zeta_{4D})$ is in bijection with the space of quantum gauge potentials (see equation (\ref{ec.3.55})).

\begin{Definition}
    \label{Y_e}
    We define the map $${\mathrm{Lep}_\r}: \mathfrak{qpc}(\zeta_{4D})\longrightarrow \mathfrak{qpc}(\zeta_{4D}), \qquad \omega\longmapsto {\mathrm{Lep}_\r}(\omega),$$
    where ${\mathrm{Lep}_\r}(\omega)$ is the qpc defined by the quantum gauge potential  $$A^{{\mathrm{Lep}_\r}(\omega)}:\su(2,\C)^\#_{4D}\longrightarrow \Omega^1(B)$$ given by $$ A^{{\mathrm{Lep}_\r}(\omega)}(\beta_1)=A^{{\mathrm{Lep}_\r}(\omega)}(\beta_2)=A^{{\mathrm{Lep}_\r}(\omega)}(\beta_3)=0,\qquad A^{{\mathrm{Lep}_\r}(\omega)}(\beta_4)=-2\,A^\omega(\beta_4),$$ where $A^\omega$ is the quantum gauge potential associated with $\omega$. 
\end{Definition}
In this way, for a given element $\Psi$ of $E^W_\l$  $$\Psi=\sum^2_{i=1}\psi_k\,\Psi_k \qquad \mbox{ with }\qquad \psi_1=\psi_2,$$ we will consider the association $$(\omega,\Psi)\;\;\longmapsto \;\; \nabla^{{\mathrm{Lep}_\r}(\omega)}_W(\Psi).$$ By equation (\ref{ec.cova1}), we get
\begin{equation}
    \label{ec.5.29.2}
    \widehat{D}_\mu \psi_1=\partial_\mu \psi_1-i\,q_4\,A^4_\mu\,\psi_1
\end{equation}
and
\begin{equation}
    \label{ec.5.29.3}
    \widehat{D}_\mu \psi=\partial_\mu \psi-i\,q_4\,A^4_\mu\,Y\,\psi,\qquad \psi= \begin{pmatrix}
    \psi_1  \\
    \psi_1
\end{pmatrix}, \qquad  Y=\begin{pmatrix}
    1 & 0  \\
    0 & 1
\end{pmatrix}
\end{equation}
As before, we define 
\begin{equation}
    \label{ec.electric1}
    Q=T_3+Y.
\end{equation}
\begin{Remark}
    \label{remacovaright1}
The most striking feature is that $\nabla^{{\mathrm{Lep}_\r}(\omega)}_W\Psi$ for every qpc $\omega$ coincides exactly with the gauge covariant derivative (also called the gauge linear connection) on the right--handed lepton singlet field $$ \psi_1=\{\mathrm{e}_\r,\;\mu_\r,\;\tau_\r  \}$$ in the  electroweak theory, including of course, the correct corresponding value of the hypercharge $$y=1.$$ Clearly $$t_3=0$$ and hence $$q=1.$$  These are the correct quantum numbers for right--handed lepton singlets (\cite{diff1}).
\end{Remark}

\subsection{The Model:  Left--Handed Quarks Doublets}

As in the previous subsection, we have

\begin{Definition}
    \label{Y_q}
    We define the map $${\mathrm{Quark}_\l}: \mathfrak{qpc}(\zeta_{4D})\longrightarrow \mathfrak{qpc}(\zeta_{4D}), \qquad \omega\longmapsto {\mathrm{Quark}_\l}(\omega),$$
    where $\mathrm{Quark}_\l(\omega)$ is the qpc defined by the quantum gauge potential  $$A^{{\mathrm{Quark}_\l}(\omega)}:\su(2,\C)^\#_{4D}\longrightarrow \Omega^1(B)$$ given by $$A^{{\mathrm{Quark}_\l}(\omega)}(\beta_j)=A^{\omega}(\beta_j) \quad \mbox{ for }\quad j=1,2,3 \quad \mbox{ and }\quad  A^{{\mathrm{Quark}_\l}(\omega)}(\beta_4)=-{1\over 3}\,A^\omega(\beta_4),$$ where $A^\omega$ is the quantum gauge potential associated with $\omega$. 
\end{Definition}
In this way, for a given element $\Psi$ of $E^W_\l$  $$\Psi=\sum^2_{i=1}\psi_k\,\Psi_k \;\;\longleftrightarrow \;\;\psi=\begin{pmatrix}
    \psi_1   \\
    \psi_2 
\end{pmatrix}$$ we will consider the association $$(\omega,\Psi)\;\;\longmapsto \;\; \nabla^{{\mathrm{Quark}_\l}(\omega)}_W(\Psi).$$ By equation (\ref{ec.5.29.1}), we get
\begin{equation}
    \label{ec.5.29.55}
    D_\mu \psi:=\partial_\mu \psi- i\,q_w\,A^1_\mu\,T_1\,\psi- i\,q_w\,A^2_\mu\,T_2\,\psi- i\,q_w\,A^3_\mu\,T_3\,\psi-i\,q_4\,A^4_\mu\,Y\,\psi,
\end{equation}
with
\begin{equation}
    \label{ec.5.29.4}
  Y={1\over 6}\begin{pmatrix}
    1 & 0  \\
    0 & 1
\end{pmatrix}
\end{equation}
As before, we define 
\begin{equation}
    \label{ec.electric2}
    Q=T_3+Y.
\end{equation}
\begin{Remark}
    \label{remacovaleft2}
The most striking feature is that $\nabla^{{\mathrm{Quark}_\l}(\omega)}_W\Psi$ coincides exactly with the gauge covariant derivative (also called the gauge linear connection) on the left--handed quark doublet fields $$ \psi=\begin{pmatrix}
    \psi_1  \\
    \psi_2
\end{pmatrix}=\left\{\begin{pmatrix}
    u  \\
    d
\end{pmatrix}_\l, \; \begin{pmatrix}
    c  \\
    s
\end{pmatrix}_\l,\; \begin{pmatrix}
    t  \\
    b
\end{pmatrix}_\l \right\}$$ in the electroweak theory, including of course, the correct corresponding value of the hypercharge $$y={1\over 6}.$$ In addition, $$\psi_1\;\;\longleftrightarrow\;\; \begin{pmatrix}
    \psi_1  \\
    0
\end{pmatrix}=\left\{\begin{pmatrix}
    u  \\
    0
\end{pmatrix}_\l, \; \begin{pmatrix}
    c  \\
    0
\end{pmatrix}_\l,\; \begin{pmatrix}
    t  \\
    0
\end{pmatrix}_\l \right\} $$ always has weak isospin $$t_3={1\over 2} $$ and hence, its electric charge is $$q={2\over 3}.$$ Moreover, 
$$\psi_2\;\;\longleftrightarrow\;\; \begin{pmatrix}
    0  \\
    \psi_2
\end{pmatrix}=\left\{\begin{pmatrix}
    0  \\
    d
\end{pmatrix}_\l, \; \begin{pmatrix}
    0  \\
    s
\end{pmatrix}_\l,\; \begin{pmatrix}
    0  \\
    b
\end{pmatrix}_\l \right\} $$ always has weak isospin $$t_3=-{1\over 2} $$ and hence, its electric charge is $$q=-{1\over 3}.$$ These are the correct quantum numbers for left--handed quark fields (\cite{diff1}).
\end{Remark}

\subsection{The Model:  Right--Handed Quarks Singlets}

As in the previous subsection, we have

\begin{Definition}
    \label{Y_qr}
    We define the map $$\mathrm{Quark}_\r: \mathfrak{qpc}(\zeta_{4D})\longrightarrow \mathfrak{qpc}(\zeta_{4D}), \qquad \omega\longmapsto \mathrm{Quark}_\r(\omega),$$
    where $\mathrm{Quark}_\r(\omega)$ is the qpc defined by the quantum gauge potential  $$A^{\mathrm{Quark}_\r(\omega)}:\su(2,\C)^\#_{4D}\longrightarrow \Omega^1(B)$$ given by $$ A^{\mathrm{Quark}_\r(\omega)}(\beta_1)=A^{\mathrm{Quark}_\r(\omega)}(\beta_2)=A^{\mathrm{Quark}_\r(\omega)}(\beta_3)=0,\qquad A^{\mathrm{Quark}_\r(\omega)}(\beta_4)=-{4\over 3}\,A^\omega(\beta_4),$$ where $A^\omega$ is the quantum gauge potential associated with $\omega$. 
\end{Definition}
In this way, for a given element $\Psi$ of $E^W_\l$  $$\Psi=\sum^2_{i=1}\psi_k\,\Psi_k \qquad \mbox{ with }\qquad \psi_1=\psi_2,$$ we will consider $$(\omega,\Psi)\;\;\longmapsto \;\; \nabla^{\mathrm{Quark}_\r(\omega)}_W(\Psi).$$ By equation (\ref{ec.cova1}), we get
\begin{equation}
    \label{ec.5.29.2.1}
    \widehat{D}_\mu \psi_1=\partial_\mu \psi_1-{i\,2\,q_4\over 3}\,A^4_\mu\,\psi_1.
\end{equation}
and
\begin{equation}
    \label{ec.5.29.3.1}
    \widehat{D}_\mu \psi=\partial_\mu \psi-i\,q_4\,A^4_\mu\,Y\,\psi,\qquad \psi= \begin{pmatrix}
    \psi_1  \\
    \psi_1
\end{pmatrix}, \qquad  Y={2\over 3}\begin{pmatrix}
    1 & 0  \\
    0 & 1
\end{pmatrix}
\end{equation}
As before, we define 
\begin{equation}
    \label{ec.electric1.1}
    Q=T_3+Y.
\end{equation}
\begin{Remark}
    \label{remacovaright1.1}
The most striking feature is that $\nabla^{\mathrm{Quark}_\r(\omega)}_W\Psi$ for every qpc $\omega$ coincides exactly with the gauge covariant derivative (also called the gauge linear connection) on the right--handed quark singlet field $$ \psi_1=\{u_\r,\;c_\r,\;t_\r  \}$$ in the electroweak theory, including of course, the correct corresponding value of the hypercharge $$y={2\over 3}.$$ Clearly $$t_3=0$$ and hence $$q={2\over 3}.$$  These are the correct quantum numbers for $\{u_\r,c_\r,t_\r  \}$ (\cite{diff1}).
\end{Remark}

By using the map
$$\mathrm{Quark}'_\r: \mathfrak{qpc}(\zeta_{4D})\longrightarrow \mathfrak{qpc}(\zeta_{4D}), \qquad \omega\longmapsto \mathrm{Quark}'_\r(\omega),$$
    where $\mathrm{Quark}'_\r(\omega)$ is the qpc defined by the quantum gauge potential  $$A^{\mathrm{Quark}'_\r(\omega)}:\su(2,\C)^\#_{4D}\longrightarrow \Omega^1(B)$$ given by $$ A^{\mathrm{Quark}'_\r(\omega)}(\beta_1)=A^{\mathrm{Quark}'_\r(\omega)}(\beta_2)=A^{\mathrm{Quark}'_\r(\omega)}(\beta_3)=0,\qquad A^{\mathrm{Quark}'_\r(\omega)}(\beta_4)={2\over 3}\,A^\omega(\beta_4),$$
we find that that $\nabla^{\mathrm{Quark}'_\r(\omega)}_W\Psi$ for every qpc $\omega$ coincides exactly with the gauge covariant derivative (also called the gauge linear connection) on the right--handed quark singlet field $$ \psi_1=\{d_\r,\;s_\r,\;b_\r  \}$$ in the classical electroweak theory, including of course, the correct corresponding value of the hypercharge $$y=-{1\over 3}.$$ Clearly $$t_3=0$$ and hence $$q=-{1\over 3},$$ as expected (\cite{diff1}).  

\begin{Remark}
     Proposition~\ref{gauge2} also holds for fermionic fields and non--commutative geometrical fermionic fields, and the proof is completely similar. Hence, our model also has the necessary symmetries for the action of equation (\ref{used1}).
\end{Remark}

\begin{Remark}
    Since all the equations of our model coincides with the ones of the electroweak interaction, it follows that our equations are also Lorentz--invariant.
\end{Remark}

    As we commented in the introductory section, the purpose of this paper is not to provide new information on the electroweak interaction, rather, the purpose of this paper is to show that the non--commutative geometrical formulation of the Yang--Mills theory presented in~\cite{sald1,sald2,saldcon} allows one to describe the electroweak theory using solely a quantum principal $SU(2)$--bundle. Since our model satisfactorily reproduces all the actions, symmetries and field equations of the electroweak theory, in principle, all the physics obtained after a second quantization and renormalization procedure  can also be reproduced. 
    
    However, the reader should be careful: our model is not an ordinary $SU(2)$--Yang--Mills theory with $3$ gauge boson fields. Rather, it is a non--commutative geometrical $SU(2)$--Yang--Mills theory with $4$ gauge boson fields.\\

{\bf Note to the reviewer:} I understand that the manuscript is relatively long. This is because of all the subsections titled \emph{General Theory}. However, as I commented in the introductory section, I really think that these subsections are necessary because I believe that a paper should be reasonably self--contained, especially when the underlying theory is not widely known (unfortunately, this is the case for Durdevich’s formulation). Since the reader must be familiar with this framework in order to properly understand the model I am proposing, a clear and sufficiently detailed summary of the parts of Durdevich’s theory relevant to this paper is necessary.

If after this brief explanation, you still consider that the paper is too long, I would be grateful if you could indicate which parts of these subsections you believe could be removed without affecting the rest of the paper. I will gladly make the necessary modifications to ensure that the manuscript fully addresses your comments on this matter, as well as any additional concerns regarding the manuscript that you may have.

\end{document}